\documentclass[12pt]{article}%
\usepackage{amssymb,amsmath,xcolor,graphicx,xspace,colortbl, rotating}
\usepackage{amsfonts}
\usepackage{amsthm}
\usepackage{caption}
\usepackage[mag=1000,portrait]{geometry}
\usepackage{indentfirst}
\usepackage[onehalfspacing]{setspace}
\usepackage{subcaption}
\usepackage{booktabs}
\usepackage{amsmath}
\usepackage{amssymb}
\usepackage{graphicx}%
\usepackage[justification=centering]{caption}
\providecommand{\U}[1]{\protect\rule{.1in}{.1in}}
\DeclareGraphicsExtensions{.pdf,.eps,.ps,.png,.jpg,.jpeg}
\providecommand{\U}[1]{\protect\rule{.1in}{.1in}}
\newtheorem{theorem}{Theorem}

\newtheorem{assumption}{Assumption}
\newtheorem{corollary}{Corollary}

\newtheorem{remark}{Remark}

\DeclareMathOperator*{\argmin}{arg\,min}

\usepackage{placeins}

\begin{document}
\makeatletter
\let\SuppSavedMaketitle\maketitle
\let\SuppSaved@maketitle\@maketitle
\let\SuppSavedTitle\title
\let\SuppSavedAuthor\author
\let\SuppSavedDate\date
\let\SuppSavedThanks\thanks
\let\SuppSavedAnd\and
\makeatother

\title{Efficient GMM and Weighting Matrix under Misspecification} 
\date{October 15, 2026\footnote{We thank seminar participants at Seoul National University and the James MacKinnon Conference at Aarhus University for helpful comments and suggestions. } }

\author{Byunghoon Kang\footnote{Department of Economics,  Lancaster University Management School. Email: b.kang1@lancaster.ac.uk. } }
\maketitle	
	
\begin{abstract}
This paper develops efficient GMM estimation under misspecified moment conditions. The influence function of the standard GMM estimator under misspecification depends on both the moment conditions and their Jacobian, motivating a new class of estimators based on augmented, recentered moment conditions. Standard GMM is a special case and generally suboptimal. By optimally weighting the augmented system, we obtain a misspecification-efficient (ME) estimator with the smallest asymptotic variance for the same GMM estimand. In linear models, its asymptotic variance takes the textbook form $(G'W^{*}G)^{-1}$, where $W^{*}$ is the inverse variance of the residualized moments after projection on the Jacobian $G$. We propose a feasible bootstrap estimator, which is a misspecification-robust and efficient version of the Hall and Horowitz (1996)  bootstrap GMM estimator. Finally, we  establish uniform local asymptotic minimax  bounds over a class of weighting matrices. We illustrate the proposed methods in simulation and two empirical examples.\\

\end{abstract}
	
\section{Introduction}

The generalized method of moments (GMM) offers a robust framework for estimation and inference, and is widely used  in applied econometrics. When the model is over-identified, where the number of available moment conditions exceeds the number of parameters, the GMM estimator depends on the choice of weighting matrix. Under correct specification, where a unique true parameter satisfies the population moment conditions exactly equal to zero, the efficient weighting matrix is the inverse of the variance of  the moments, which achieves the smallest asymptotic variance among GMM estimators (Hansen (1982), Chamberlain (1987)). 

However, economic models are inherently structural approximations, and the overidentifying restriction test is often rejected due to various forms of misspecification; candidate instruments may be invalid, lag structures or functional forms may be misspecified or treatment effect heterogeneity implies that no parameter satisfies all moment equations exactly. A growing literature therefore studies the inference for the pseudo-true value, defined as the population minimizer of the objective function (e.g., White (1982), Hall and Inoue (2003), M\"{u}ller (2013), Hansen and Lee (2021), and Andrews and Kwon (2024)).\footnote{In the asset pricing literature, such as Kan and Robotti (2009) and Gospodinov, Kan, and Robotti (2014), the GMM pseudo-true value can be defined as a minimizer of the distance metric of Hansen and Jagannathan (1997).} Under misspecification, the asymptotic variance of the GMM estimator has different forms than the standard GMM variance estimator, and ignoring this can lead to inconsistent standard error estimation and asymptotically invalid confidence intervals for the pseudo-true values, which was also highlighted in Andrews, Chen, and Tecchio (2025).

More importantly, misspecification fundamentally changes the role of the weighting matrix, and has different implications than under the correct specification setup. The choice of the weighting matrix no longer merely affects the statistical efficiency of the estimator, but it also affects what is being estimated; different weighting matrices correspond to different target parameters.\footnote{Pseudo-true values may differ from quantities of economic interest (M\"{u}ller (2013), Andrews, Barnhard, and Carlson (2024), Andrews, Barahona, et al. (2025)). The main focus of this paper is on valid and efficient inference for a given GMM or two-stage least squares (TSLS) estimand, instead of identification and its causal interpretation; see, e.g.,  Koles\'{a}r (2013), Mogstad et al. (2021), Blandhol et al. (2025),  and S\l oczy\'{n}ski (2024) for the latter direction in the linear instrumental variables (IV) model under treatment effect heterogeneity.} This dual role of the weighting matrix under misspecification raises a number of questions that the standard GMM literature does not directly address. Is the ``standard" optimal weight, defined under correct specification, still efficient under misspecification? Does any other class of estimators exist with a smaller asymptotic variance than the standard GMM under misspecification for the same target? How should we think about asymptotic efficiency across different weighting matrices when each weighting matrix targets a different estimand? This paper addresses these challenging questions and provides theoretical and some empirical guidance for researchers using overidentified GMM models. 

First, for a fixed weighting matrix $W$ and the associated GMM estimand (pseudo-true value) $\theta_W$, we observe that the influence function of the ``standard" GMM estimator under misspecification can be expressed as a linear combination of augmented moment conditions, which stack the original moment functions and their Jacobian, with recentering. This motivates a new class of M-estimator based on the augmented moments that includes the ``standard" misspecified-GMM as a special case. We then define the ``misspecification-efficient" (ME) estimator as the estimator with the smallest asymptotic variance within this class, which is consistent for the \textit{same} GMM pseudo-true value. By applying the optimal weighting, derived from the inverse of the variance matrix of the augmented moments, the proposed ME estimator captures the correlation between the moments and the Jacobian and extracts additional identifying information from the moment misspecification, whereas the standard GMM is suboptimal with particular linear combination of the available moments dictated by the usual GMM first-order condition.

The ME estimator is not directly feasible because it requires the population means of the original moments and the Jacobian for recentering, i.e., the degree of misspecification and the identification strength. However, our asymptotic result of the ME estimator is constructive in the sense that the variance of the ME estimator can be easily estimated and reported as an oracle efficiency bound; showing researchers what could be gained if these quantities were known. Under the linear model, the bound is $ V_{ME}(W) = (G^{\prime} W^{*}G)^{-1}$, where $G$ is the Jacobian and $W^{*} = \Sigma_{11, 2}^{-1}$ is the inverse of the variance of residualized moments after projection on the Jacobian, instead of the conventional efficient weighting matrix $W_0 = \Sigma_{11}^{-1}$ that is  the inverse of the variance of the original moments. This bound is invariant to $W$ under the linear model; different $W$ target different $\theta_W$, but do not affect how precisely they can be estimated. It also coincides with the textbook ``efficient GMM" formula $(G'W^{*}G)^{-1}$, which offers a different angle on the well-known downward bias of the efficient-GMM standard errors (Windmeijer (2005)); the formula is valid for the oracle ME estimator, but too ``small'' for the standard GMM estimator under misspecification when $W^* \approx W_0$, even in large samples. We therefore recommend reporting  $V_{ME}(W) $  alongside the misspecification-robust standard errors  (Andrews, Chen, and Tecchio (2025)) as a transparent efficiency measure that does not require committing to a specific $W$. Under homoskedasticity, the ratio of the GMM variance to the ME bound  is approximately determined by the moment-Jacobian correlation and the $J$ statistic scaled by the identification strength (see Remark \ref{rem:gap}). In general nonlinear models, $V_{ME}(W)$ varies with $W$.

Second, we propose two feasible bootstrap procedures. The bootstrap ME-GMM estimator is based on the GMM criterion function that recenters both the original moments and the Jacobian at their sample analogs evaluated at a preliminary GMM estimate, and uses the efficient weighting matrix of the augmented moments. We show that its bootstrap distribution approximates that of the oracle ME estimator, and it can be used to estimate the efficiency bounds. The proposed estimator can be viewed as a misspecification-robust and efficient version of the Hall and Horowitz (1996) recentered bootstrap, which is only valid under correct specification. We also consider the  ``double recentered" (DR) bootstrap, which jointly perturbs the moment conditions and the Jacobian as in the first-order condition of the  misspecified GMM, and show that it mimics the limiting distribution of the standard GMM estimator regardless of model specification.  Under correct specification, the additional Jacobian recentering is a higher-order correction and the DR bootstrap nests the original Hall and Horowitz (1996) bootstrap, and under misspecification it correctly approximates the Jacobian variation term that the Hall and Horowitz (1996) bootstrap omits. We also consider  a repeated sample-splitting ME estimator in the spirit of Angrist and Krueger (1995) in Supplemental Appendix \ref{sec:splitsample}. 

Third,  we develop uniform semiparametric efficiency bounds over a class of weighting matrices $W \in \mathcal{W}$. The notion of semiparametric efficiency needs to be defined  differently under (global) misspecification, because $\theta_W$ changes with $W$, whereas the semiparametric efficiency framework typically assumes a single true parameter. We view each $W$ as a different candidate model/design, and do not aim to choose an optimal $W$. For example, a  more precisely estimated LATE (local average treatment effect) for one subpopulation is not necessarily better than a noisier LATE for another; they may answer different causal questions. We first derive the efficiency bound for $\theta_W$ with a fixed $W$, and then establish uniform local asymptotic minimax bounds over $W\in \mathcal{W}$. These bounds are useful both for researchers who target a specific $\theta_W$ (e.g., TSLS), and for those who are agnostic about $W$, for example, when comparing one-step, two-step, and iterated GMM, each corresponds to a different pseudo-true value under misspecification.

We illustrate the proposed methods with Monte Carlo simulations and two empirical applications.  The simulation results confirm that the oracle ME is meaningfully tighter, and the DR percentile CI delivers valid coverage, with lengths similar to the misspecification-robust Wald CI.  In the returns to schooling examples of Card (1995) and Angrist and Krueger (1991), the ME efficiency  bound is roughly 3-18\% smaller than the  misspecification-robust standard errors under the efficient weight $W^*= \Sigma_{11,2}^{-1}$ or the TSLS weight $(Z'Z)^{-1}$, and the gap is larger in Card (1995), where the moment-Jacobian correlation is larger.  Overall, we find that the  ME efficiency bound is most informative when (i) $W$ is an efficient weight ($W^*$ or $W_0$), rather than $W=I$ or other naive choices, and (ii) the identification strength is moderate, so that the Jacobian sampling variation is non-negligible. 

Our paper complements a growing literature on sensitivity analysis and robust inference under misspecification (e.g., Andrews et al. (2017), Armstrong and Koles\'{a}r (2021), Bonhomme and Weidner (2022), Christensen and Connault (2023), and Armstrong et al. (2025)) and on the  relationship between the pseudo-true and true values  (e.g., Andrews, Barnhard, and  Carlson (2024), Andrews, Barahona et al. (2025); see Andrews, Chen, and Tecchio (2025) for a review with recommendations for practitioners). While these papers focus on bias, sensitivity, or optimal inference, our paper analyses efficient estimation in misspecified GMM, and provides a feasible  efficiency measure that researchers can report alongside robust standard errors. Our bootstrap procedures relate to the recentered GMM bootstrap (Hall and Horowitz (1996), Brown and Newey (2002)) and its misspecification-robust alternative (Lee (2014)), and to the use of Jacobian recentering and moment-Jacobian orthogonalization in the weak-identification literature (e.g., Kleibergen (2005), Dovonon and Gon\c{c}alves (2017), Lee and Liao (2018), Kleibergen and Zhan (2025)).

The remainder of the paper is organized as follows. Section \ref{sec:setup} introduces the setup and Section \ref{sec:AM} develops the ME estimator. Section \ref{sec:bootstrap} analyses the  bootstrap procedures and Section \ref{sec:semi_efficiency} provides the uniform efficiency bounds. Section \ref{sec:simulation} reports the Monte Carlo evidence, Section \ref{sec:empirical} provides empirical results, and Section \ref{sec:conclusion} concludes. The Appendix includes all proofs of the results in the main paper.  The Supplemental Appendix contains additional materials. Section \ref{sec:linear_iv} provides details on the linear IV example, Section \ref{sec:worstcase_variance} provides the worst-case variance bounds over misspecification and identification strength in the spirit of Conley, Hansen, and Rossi (2012). Sections \ref{sec:splitsample}-\ref{sec:decomp_ME} contain the split-sample ME estimator, additional simulations, the dynamic panel application  of Arellano and Bond (1991), and the derivation of the decomposition of the ME efficiency gain, respectively.

\subsection{Notation}\label{sec:notation}
 \noindent $||A ||$ denotes the spectral norm. $A \geq B$ denotes that $A-B$ is positive semidefinite for real symmetric matrices A and B. Let $o_p(\cdot)$ and $O_p(\cdot)$ denote the usual stochastic order symbols, and $\overset{p}{\longrightarrow}$ and $ \overset{d}{\longrightarrow}$ denote convergence in probability, convergence in distribution, respectively. Superscript $\phantom{}^{*}$ denotes a probability or moment computed under the bootstrap distribution conditional on the original data set. Let $\xrightarrow{p^{*}}$ and $\xrightarrow{d^{*}}$ denote the convergence in probability, in probability, and the convergence in distribution, in probability, respectively. In addition, we write $\xi_{n}=o_{p^{*}}(1)$ if $\xi_{n}\xrightarrow{p^{*}}0$ and $\xi_{n}=O_{p^{*}}(1)$ if $\xi_{n}$ is bounded in probability, in probability.

\section{The Setup}\label{sec:setup}

Suppose that we have data  on observations $\{ X_i\}_{i=1}^{n}$ randomly drawn from the unknown probability distribution $F_0$, with a data vector $X$. Let $g(X, \theta)$ be an $m\times 1$ vector of moment functions and $\theta$ be a  $p\times 1$ parameter vector. The model is based on the moment conditions, and a popular estimator of $\theta$ in this setup is the generalized method of moments (GMM) estimator of Hansen (1982), 
\begin{align}\label{gmm_definition}
\widehat{\theta}_{GMM} (W)  &= \argmin_{\theta} g_n(\theta)^{\prime} W g_n(\theta)
\end{align}
where $g_n(\theta) = n^{-1}\sum_{i=1}^{n} g (X_i, \theta)$ and  $W$ is the weighting matrix. Under regularity conditions, the GMM estimator converges to the pseudo-true value $ \theta_{W} $, which is defined as a minimizer of the population objective function
\begin{align}\label{pseudo_definition}
 \theta_{W} = \argmin_{\theta} g(\theta)^{\prime} W g(\theta) 
 \end{align}
where $g(\theta) = E[g_n(\theta)]$. Assuming global identification, $ \theta_{W} $ satisfies the following just-identified equation from the first order condition (FOC) of \eqref{pseudo_definition};
\begin{align}\label{pseudo_definition2}
G(\theta_{W})'W g(\theta_{W}) = 0
\end{align}
where $G(\theta) = E [\frac{\partial g(X_i, \theta)}{\partial\theta^{\prime}}]$ is the population Jacobian. When the model is correctly specified, i.e., there exists a unique $\theta_0$ that satisfies the original moment equations
\begin{align}\label{correctspe_definition}
g(\theta_0) = E[g_n (\theta_0)] = 0, 
\end{align}
the pseudo-true value $\theta_W= \theta_0$ for all $W$, and it does not depend on $W$.  Even if the moment condition is (globally) misspecified, i.e.,  $E[g_n (\theta)] \neq0,~\forall\theta \in \Theta$, the FOC condition \eqref{pseudo_definition2} still holds. However, the GMM pseudo-true values generally depend on the choice of $W$ under misspecification when the model is overidentified ($m>p$), as it satisfies the linear combinations of non-zero  population moments $g(\theta)$ as in \eqref{pseudo_definition2}, so  varying $W$ changes the parameter being estimated.  

In this paper, our objective is to provide valid inference procedures for the pseudo-true value $\theta_W$.  It is difficult to give general economic interpretation to the pseudo-true parameter values, and instead, researchers typically analyze and interpret them case-by-case. For example, in the linear IV model, moment misspecification can arise due to the treatment effect heterogeneity, and it is common to focus attention on the TSLS estimand, which corresponds to $\theta_W$ with $W = E[Z_i Z_i^{\prime}]^{-1}$ and the instrumental variables $Z_i$. In this setup, $\theta_W$ can be interpreted as a weighted LATE under appropriate assumptions; see Angrist and Imbens (1995).

The influence function of the ``standard" GMM estimators under misspecification can be expressed in terms of augmented moment conditions, which include (i) the original moment restriction, (ii) their Jacobian, and (iii) with recentering. This result, while implicit in the form of the asymptotic variance of the misspecified-GMM in the literature, motivates the new class of estimators based on the augmented moments.

To see this, the GMM estimator has the following expression under misspecification with  the standard regularity conditions (see the proof of Theorem \ref{thm:dist_AM} for details), 
\begin{equation}
\sqrt{n} (\widehat{\theta}_{GMM} (W) - \theta_{W} ) =  - (A (W, \theta_{W} ) \Gamma (\theta_{W} ))^{-1} A (W, \theta_{W} )\frac{1}{\sqrt{n}} \sum_{i=1}^{n} \Big(\psi(X_i, \theta_{W} ) - E[\psi(X_i, \theta_{W} )]\Big) + o_p(1), \label{eqn_influence}
\end{equation}
where 
\begin{eqnarray}
\underbrace{\psi(X_i, \theta)}_{m(p+1) \times 1}&=& \begin{bmatrix}
g(X_i, \theta)  \\
vec[G(X_i, \theta)^{\prime}]\nonumber
\end{bmatrix}, \quad G(X_{i}, \theta)=\frac{\partial g(X_i, \theta)}{\partial\theta^{\prime}}, \quad F(X_i, \theta) =\frac{\partial vec [G(X_i, \theta)^{\prime}]}{\partial\theta^{\prime}}, \\
 \underbrace{\Gamma (\theta)}_{m(p+1) \times p}  &=&  E [\frac{\partial \psi(X_i, \theta) }{\partial \theta^{\prime}} ] = \begin{bmatrix}
G(\theta)\\
F(\theta)
\end{bmatrix}  \label{eqn:AGamma}, \qquad 
\underbrace{A (W, \theta)}_{p \times m(p+1)} = [\underbrace{G(\theta)^{\prime} W}_{p \times m} \quad \underbrace{g(\theta)^{\prime} W \otimes I_p}_{p \times mp} ].
\end{eqnarray}

We emphasize that the Jacobian  term $G(X_i, \theta)$ already enters the influence function under misspecification without imposing extra moments. Based on \eqref{eqn_influence}, we can derive an asymptotic distribution of the GMM estimators under misspecification,
\begin{align}
\sqrt{n}(\widehat{\theta}_{GMM} (W)-\theta_{W} ) \overset{d}{\longrightarrow}  N(0, V_{GMM} (W, \theta_{W})) \label{asymptotic_dist}
\end{align}
where 
\[
V_{GMM}(W, \theta) = (A (W , \theta)\Gamma (\theta ))^{-1} A (W , \theta) \Sigma (\theta ) A (W , \theta)^{\prime} ( ( A (W , \theta)\Gamma (\theta))^{\prime})^{-1}, 
\]
and $\Sigma (\theta) =E [ \psi(X_i, \theta) \psi(X_i, \theta)^{\prime}]  -E [ \psi(X_i, \theta)] E [ \psi(X_i, \theta)]^{\prime} $. The asymptotic distribution results of the misspecified-GMM in \eqref{asymptotic_dist} are well-known in the literature; see Imbens (1997), Hall and Inoue (2003),  Hansen and Lee (2021) for derivations of same result under misspecification for one-step/two-step GMM, and the iterated GMM estimator.\footnote{One can generally interpret the overidentified GMM models as a just-identified GMM estimator with an augmented parameter vector (Newey and McFadden (1994)), and Imbens (1997) finds the connection with the misspecification-robust variance. However, they did not consider the explicit asymptotic variance form of the misspecified GMM estimator and the choice of weighting matrix, which are the main focus of this paper.}

\section{Efficient Estimator under Misspecification}\label{sec:AM}

By inspection of the influence function in \eqref{eqn_influence}, we observe that the misspecified-GMM estimator is related to the M-estimator where the original moment functions and the Jacobian moments are ``stacked" to form a vector of augmented moment conditions.  We consider the following recentered moment function with augmented moments, $\widetilde{\psi}(X_i, \theta):=\widetilde{\psi}(X_i, \theta; W)  = \psi(X_i, \theta) - E[\psi(X_i, \theta_{W})]$, where $\psi(X_i, \theta)$ is defined in \eqref{eqn:AGamma}, and  $\theta_W$ is  the GMM pseudo-true value in \eqref{pseudo_definition} for fixed $W$.  By construction, the population augmented moments $E[\widetilde{\psi}(X_i,  \theta)] $ are zero when $\theta$ is evaluated at $\theta_{W}$. 

We can then transform this to the just-identified moment  $\Lambda E[\widetilde{\psi}(X_i, \theta_{W})] = 0$ using the transform matrix $\Lambda  = [\Lambda_1 \quad \Lambda_2]$, where $\Lambda_1, \Lambda_2 $ are $p \times m$, and $p\times mp$ matrices, respectively. Then, we consider the following general class of \textit{M-estimator} (or Z-estimator), $\widetilde{\theta} (\Lambda)$, to denote an M-estimator based on the sample analogs of augmented moment conditions
\begin{equation}
\Lambda \widetilde{\psi}_n(\widetilde{\theta} (\Lambda)) = 0 \label{eqn_mestimation_sample}
\end{equation}
where $\widetilde{\psi}_n (\theta) = \frac{1}{n} \sum_{i=1}^{n} \widetilde{\psi}(X_i, \theta)$.  Alternative choices of $\Lambda$ are associated with alternative estimators that are consistent for $\theta_W$. This estimator is not feasible without estimation of $E[\psi(X_i, \theta_{W})]$, but we postpone the discussion to focus first on the efficient choice of $\Lambda$ and the asymptotic distribution of the oracle estimator.

\subsection{Asymptotic Distribution of the  M-Estimator with Augmented Moments}\label{subsec:AM_distribution}

 The following assumptions are used to show the asymptotic distribution of $\widetilde{\theta} (\Lambda)$.

\begin{assumption}\label{assump:AM}
\begin{enumerate}
\item  The observations $ \{ X_{i} \}_{i=1}^{n}$ are independent and identically distributed (i.i.d.).
\item  For $W >0$, $\theta_{W}  =  \argmin_{\theta} g(\theta)^{\prime} W g(\theta)$ is unique, and $\theta_{W}$ is in the interior of the compact parameter space $\Theta \subset \mathbb{R}^{p}$, where $g(\theta) \neq 0$ for  all $\theta \in \Theta$. 

\item $g(X, \theta)$ and $G(X, \theta)$ are continuous at each $\theta \in \Theta$ with probability one, and $E[\sup_{\theta} || g(X_i, \theta)||] < \infty, E[\sup_{\theta} || G(X_i, \theta)||] < \infty, E[\sup_{\theta} || F(X_i, \theta)||] < \infty$.

\item $g(X, \theta)$ is twice continuously differentiable in $\theta$ with probability one, and $E[||g(X, \theta_W) ||^{2}] <\infty$, $E[||G(X, \theta_W) ||^{2}] <\infty$. 
\item  $\Lambda E[\widetilde{\psi}(X_i,  \theta)] \neq 0 $ for $\theta \neq \theta_W$. $\Lambda \Gamma (\theta_{W})$ exists and is nonsingular.

\end{enumerate}
\end{assumption}

Assumption \ref{assump:AM}.3 ensures uniform convergence of the $\widetilde{\psi}_n (\theta)$ and the population objective function $E[\widetilde{\psi}_n (\theta)]$ to show consistency. Assumption \ref{assump:AM}.4 is for the asymptotic normality of the normalized sum of $\widetilde{\psi}_n (\theta_{W})$.  The first part of Assumption \ref{assump:AM}.5 is an identification condition. A necessary and sufficient condition for (local) identification is that $\Lambda \Gamma(\theta_W) = \Lambda_1 G(\theta_W) + \Lambda_2 F(\theta_W)$ has  full column rank.\footnote{This is also related to the second-order identification condition because even if the Jacobian $G(\theta_W)$ is zero (rank deficient), the model can still be identified by nonzero $\Lambda_2 F(\theta_W)$.}  For example, when the moment functions are linear in $\theta$, this reduces to the condition that $\Lambda \Gamma (\theta_W) = \Lambda_1 G$ has  full column rank, where  $G = E[G(X_i)]$ does not depend on $\theta$. When $\Lambda_1$ takes the form of $\Lambda_1 = G'W$ with $W>0$, it reduces to the standard GMM rank condition, which is implied by the global identification of $\theta_W$ in Assumption \ref{assump:AM}.2.

\begin{theorem}\label{thm:dist_AM}
Under Assumption \ref{assump:AM},
\begin{align}
\sqrt{n}(\widetilde{\theta} (\Lambda) - \theta_{W} ) \overset{d}{\longrightarrow}  N(0, V_{M}(\Lambda, \theta_{W})),
\end{align}
where 
\[
V_{M} (\Lambda, \theta_W) = (\Lambda \Gamma (\theta_W))^{-1} \Lambda \Sigma (\theta_{W} ) \Lambda ^{\prime} ( (\Lambda \Gamma (\theta_W))^{\prime})^{-1}
\]
and $\Sigma (\theta) =E [ \psi(X_i, \theta) \psi(X_i, \theta)^{\prime}]  -E [ \psi(X_i, \theta)] E [ \psi(X_i, \theta)]^{\prime} $.
\end{theorem}

Theorem \ref{thm:dist_AM} shows that $\widetilde{\theta} (\Lambda)$ is consistent for the same GMM pseudo-true value $\theta_{W}$, and the augmented-M estimator is a sufficiently general class of estimators to include the ``standard" misspecified-GMM as a special case. 
When $\Lambda  = A(W, \theta_W) = [G(\theta_W)^{\prime} W \quad g(\theta_W)^{\prime} W \otimes I_p]$ as in  \eqref{eqn_influence}-\eqref{eqn:AGamma},  $\widetilde{\theta} (\Lambda)$ has the same influence function representation as the original misspecified-GMM estimator $\widehat{\theta}_{GMM} (W)$ and has the same asymptotic distribution as in \eqref{asymptotic_dist}.

\subsection{Efficient Choice of $\Lambda$ and the Efficient Estimator under Misspecification}\label{subsec:ME_estimator}

More importantly, Theorem \ref{thm:dist_AM} justifies the use of our new class of estimators, through the following corollaries,  as we can consider an ``efficient" choice of $\Lambda$ to minimize the asymptotic variance  $V_{M} (\Lambda, \theta_{W})$. The optimal choice of $\Lambda $ takes the following form\footnote{When $\Sigma(\theta_W)$ is singular, we can use the generalized inverse $\Sigma^{-}$ and all the arguments below will still be valid.  $\Sigma$ is singular when some components of moments have an overlap, which occurs for example, when $g(X_i, \theta) = ( X_i-\theta , (X_i-\theta)^2 - 1)^{\prime}$ as $G(X_i, \theta) = (-1, -2 (X_i-\theta))^{\prime}$. In this case, the Jacobian carries no information beyond the original moments, but $(\Gamma'\Sigma^{-}\Gamma)$ may still be nonsingular.}
\begin{eqnarray}
\Lambda^{*}_{W} =\big( \frac{\partial }{\partial \theta^{\prime} } E[\widetilde{\psi} (X_i, \theta)] \big)^{\prime} |_{\theta = \theta_W} Var [\widetilde{\psi} (X_i, \theta_W)]^{-1} = \Gamma(\theta_W)^{\prime} \Sigma (\theta_W)^{-1}, \label{optimal_weighting}
\end{eqnarray}
by the standard theory of optimal estimating equations (Godambe (1960)) or semiparametric efficiency of GMM.

With $\Lambda^{*}_{W} $, we use the optimal linear combinations of the augmented moments, while the specific choice $\Lambda  = A(W, \theta_W)$ based on the FOC of the GMM objective is inefficient. When the estimator $\widetilde{\theta} (\Lambda)$ is constructed with $\Lambda^{*}_{W}$ (or any consistent estimator of $\Lambda^{*}_{W}$), we label it as the \textit{misspecification-efficient} (ME) estimator, $\widehat{\theta}_{ME} (W) = \widetilde{\theta} (\Lambda^{*}_{W})$.

\begin{corollary}\label{cor:dist_ME} Suppose Assumption \ref{assump:AM} holds for $\Lambda =\Lambda^{*}_{W}  = \Gamma(\theta_W)^{\prime} \Sigma (\theta_W)^{-1}$. Then, we have
\begin{equation}
\sqrt{n}(\widehat{\theta}_{ME} (W) - \theta_{W} ) \overset{d}{\longrightarrow}  N(0, V_{ME}(W)), 
\label{optimal_variance} 
\end{equation}
where $V_{ME}(W)  = (\Gamma(\theta_W)^{\prime} \Sigma (\theta_W)^{-1}\Gamma(\theta_W))^{-1}$, and  the following holds for any $\Lambda$ that satisfies Assumption \ref{assump:AM},
\[
V_{M} (\Lambda, \theta_W) \geq V_{ME}(W).
\]
Further, the asymptotic variance of the ME estimator has the following form,
\begin{equation}
V_{ME}(W)   = (G(\theta_W)^{\prime} \Sigma_{11}^{-1} G(\theta_W) +F_G(\theta_W)^{\prime}  \Sigma_{22, 1}^{-1} F_G(\theta_W))^{-1}, \label{optimal_variance_general}
\end{equation}
where $F_G(\theta_{W}) = F(\theta_{W}) - \Sigma_{21} \Sigma_{11}^{-1} G(\theta_{W})$, $\Sigma_{22,1} = (\Sigma_{22} - \Sigma_{21}\Sigma_{11}^{-1} \Sigma_{12})$, and $\Sigma (\theta_W)  = \begin{bmatrix} \Sigma_{11} & \Sigma_{12} \\ \Sigma_{21}&\Sigma_{22}\end{bmatrix}$.
\end{corollary}
\vspace{0.16cm}

Corollary \ref{cor:dist_ME} shows that the misspecification-efficient  estimator has the smallest asymptotic variance in the class of M-estimators that we consider with  the augmented moment conditions $\widetilde{\psi}(X_i,  \theta) $. Since  misspecified-GMM estimator is the augmented-M estimator  $\widetilde{\theta} (\Lambda)$ with $\Lambda  = A(W, \theta_W) = [G(\theta_W)^{\prime} W \quad g(\theta_W)^{\prime} W \otimes I_p]$, Corollary \ref{cor:dist_ME} implies that
\begin{equation}
V_{GMM}(W, \theta_W) = V_{M} (A(W, \theta_W), \theta_W) \geq  V_{ME}(W).\label{variance_comparison_1}
\end{equation}

Corollary \ref{cor:dist_ME} also provides some insights on the ``conventional" asymptotic variance of the efficient GMM under correct specification. The ME variance in \eqref{optimal_variance_general} decomposes as the correct-specification efficient GMM variance contribution $G(\theta_W)^{\prime} \Sigma_{11}^{-1} G(\theta_W)$ (evaluated at $\theta_W$ here instead of $\theta_0$) and the curvature term contribution from the second derivative of the moments $F$ after residualization $(F_G(\theta_W)^{\prime}  \Sigma_{22, 1}^{-1} F_G(\theta_W))$,  and we can deduce that
\begin{equation}
(G(\theta_W)^{\prime} \Sigma_{11}^{-1} G(\theta_W))^{-1} \geq  V_{ME}(W), \label{variance_comparison}
\end{equation}
which implies that using the extra moment conditions from Jacobian does not worsen the asymptotic efficiency of estimators even under correct specification.\footnote{Under correct specification, $E[g (X_i, \theta_0)]  = 0$, and  $V_{GMM}(W, \theta_0)$ in \eqref{asymptotic_dist} reduces to the classical form $ (G'WG)^{-1} G'W \Sigma_{11} W G (G'WG)^{-1}$ evaluated at $\theta_0$, where $\Sigma_{11}$ is the $m \times m$ upper-left submatrix  of $\Sigma(\theta_0)$. This is minimized by the ``standard" optimal-weight matrix $W_0 =   \Sigma_{11}^{-1}$, which gives $(G'  \Sigma_{11}^{-1} G)^{-1}$.}

When  $F_G(\theta_{W}) = F(\theta_{W}) - \Sigma_{21} \Sigma_{11}^{-1} G(\theta_{W}) = 0$, we have $ V_{ME}(W) = (G(\theta_W)^{\prime} \Sigma_{11}^{-1} G(\theta_W))^{-1}$, so the ME estimator provides no efficiency gain.  In the linear IV model, where $F(\theta_W) = 0$ and $G$ does not depend on $\theta$, this condition holds when $\Sigma_{12} = 0$, i.e., the covariance between the moment function $g(X_i, \theta)$ and the Jacobian $G(X_i)$ is zero. However, $\Sigma_{12}$ (and thus $F_G(\theta_{W})$) is generally nonzero both under correct specification and misspecification, implying that the ME estimator provides strictly positive efficiency gains in most empirically relevant settings.\footnote{To see this, consider the simple linear IV model where $g(X_i, \theta) =Z_i (Y_i - X_i\theta) $, and  $\theta$ is a scalar. Define $ \varepsilon_i = Y_i - X_i \theta_W, \sigma_{\varepsilon X} = E[\varepsilon_i X_i |Z_i]$,  $\sigma_{\varepsilon Z}  =E[Z_i \varepsilon_i ]$, and $X_i = \pi^{\prime} Z_i + v_i, E[Z_i v_i] = 0$.  Then, $\Sigma_{12} = - (E[\sigma_{\varepsilon X}  Z_i Z_i^{\prime}] - \sigma_{\varepsilon Z}  E[X_i Z_i^{\prime} ])  $. Under correct specification ($\sigma_{\varepsilon Z} = 0)$, $\Sigma_{12} =-E[\sigma_{\varepsilon X}  Z_i Z_i^{\prime}]$ is generally nonzero whenever there is endogeneity, $\sigma_{\varepsilon X}  \neq 0$. Under misspecification ($\sigma_{\varepsilon Z} \neq 0)$, $\Sigma_{12}$ is also generally nonzero.} 
 Intuitively, variations of the Jacobian affect the distribution of the GMM estimator under misspecification, and thus improved efficiency can be achieved by imposing efficient weighting of the augmented moments accounting for the correlations between $g(X_i, \theta)$ and $G(X_i, \theta)$.

\subsection{ME Estimator under the Linear Model}\label{subsec:ME_estimator_linear}

Under the linear model, we have further insights on the optimal choice of $\Lambda^{*}_{W}$ and the asymptotic variance of the ME estimator. This will facilitate the discussions on  the efficient weighting matrix under misspecification, compared with  the ``conventional" optimal weighting matrix  that was defined under correct specification.
\vspace{0.4cm}

\noindent \textbf{Example: Linear IV}. 
Consider a linear instrumental variable regression $Y_i = X_i^{\prime} \theta + \varepsilon_i $ where $Y_i$ is an outcome variable, $X_i$ is a potentially endogenous regressor, and $Z_i$ is a $m\times 1$ vector of instrumental variables. Let $Y, X, Z$ be  $n \times 1, n\times p$, and $n\times m$ data matrices with row $i$ equal to $Y_i^{\prime}, X_i^{\prime}, Z_i^{\prime}$, respectively. The moment function is $ g(\cdot, \theta) = Z_i (Y_i - X_i^{\prime} \theta)$, and we consider moment misspecification $E[Z_i (Y_i - X_i^{\prime} \theta)] \neq 0 $. With a $m\times m$ matrix $W>0$, the standard GMM estimator is 
\[
\widehat{\theta}_{GMM}(W) = (X'Z WZ'X)^{-1} X'Z W Z'Y
\]
with the GMM pseudo-true value,
\[
\theta_W = (E[X'Z] W E[Z'X])^{-1}  (E[X'Z] W E[Z'Y]).
\]
When $W = E[Z'Z]^{-1}$,  $\theta_W$ is the TSLS estimand.  Based on \eqref{eqn_mestimation_sample} and using the linearity of the moment condition, and the optimal $\Lambda^{*}_{W}$ in Corollary \ref{cor:optimal_matrix} below, the ME  has a closed-form as follows
\begin{multline}
\widehat{\theta}_{ME} (W) = (\frac{1}{n}X'Z \Sigma_{11, 2}^{-1} Z'X)^{-1} \biggl( X'Z \Sigma_{11, 2}^{-1} (\frac{1}{n}Z'Y -  E[g(X_i, \theta_{W})])  \\- X'Z  \Sigma_{11,2}^{-1} \Sigma_{ 12} \Sigma_{ 22}^{-1} ( \frac{1}{n}vec(X'Z) - E[vec(X_iZ_i')])  \biggr) \label{def:ME_linear_estimator_efficient} 
\end{multline}
where  $\Sigma_{ 11, 2} = (\Sigma_{11} - \Sigma_{12} \Sigma_{22}^{-1} \Sigma_{ 21}),  \Sigma (\theta_W)  = \big[\begin{smallmatrix} \Sigma_{11} & \Sigma_{12} \\ \Sigma_{ 21}&\Sigma_{22}\end{smallmatrix}\big]$, with $\Sigma (\theta) = E [ \big( \begin{smallmatrix} Z_i Y_i-Z_i X_i'\theta \\ vec(X_iZ_i^{\prime}) \end{smallmatrix} \big) \big(\begin{smallmatrix} Z_i Y_i-Z_i X_i'\theta \\ vec(X_iZ_i^{\prime}) \end{smallmatrix}\big)^{\prime}]  -E  \big(\begin{smallmatrix} Z_i Y_i-Z_i X_i'\theta \\ vec(X_iZ_i^{\prime}) \end{smallmatrix}\big) E  \big(\begin{smallmatrix} Z_i Y_i-Z_i X_i'\theta \\ vec(X_iZ_i^{\prime}) \end{smallmatrix}\big)^{\prime} $. Sample analogs of  $\Sigma(\theta_W)$ can also  be used here. In the linear IV example, the ME estimator \eqref{def:ME_linear_estimator_efficient} can be viewed as efficient linear combinations of the recentered GMM/TSLS estimator and the Jacobian moment conditions.

Assuming that the data-generating process is $Y_i = \theta_0 X_i + Z_{2i} +\varepsilon_i, X_i = Z_{1i} + 2Z_{2i}+ v_i$, $Z_i = (Z_{1i}, Z_{2i})^{\prime} \overset{i.i.d.}{\sim} N(0, I_2), (\varepsilon_i, v_i) \sim N(0, [1, 0.5; 0.5, 1])$ with $\theta_0 = 0$. When we consider the weight matrix $W =(\begin{smallmatrix} 1 &\rho\\ \rho &1 \end{smallmatrix})$ for $\rho \in (-1, 1)$, the GMM pseudo-true value is $\theta_W(\rho) = \frac{\rho + 2}{5 + 4\rho}$, which is a strictly decreasing function of $\rho \in (-1, 1)$, and there is a one-to-one mapping from $\rho$ to $\theta_W(\rho)$.  By using the formula above, the asymptotic variance of the GMM estimator as a function of $\rho$ is:
\[ V_{GMM}(\rho) = \frac{2(56\rho^4 + 131\rho^3 + 210\rho^2 + 232\rho + 100)}{(5 + 4\rho)^4}. \]
The variance of the ME estimator can be calculated as $V_{ME} = \frac{969}{6116} \approx 0.1584$ (which does not depend on $\rho$), and it can be verified that  $V_{GMM}(\rho) > V_{ME}$ for all $\rho$. At  $\rho = 0$ (TSLS), the variance simplifies to $V_{GMM}(0) = 0.32$, and the efficiency gain is over 50\% against TSLS. The efficiency gain ($1- \frac{V_{ME}}{V_{GMM}(\rho)}$) grows as $\rho$ decreases; for example, it  is 73\% and 88\% for $\rho = -0.5$ and $-0.9$, respectively.\footnote{Although we do not report the results here for brevity, we also considered a general nonlinear instrumental variables (NLIV) model, where the moment function is $g(X_i, \theta) = Z_i (Y_i - \exp(X_i'\theta))$. This model has been considered in count data models with endogenous regressors (e.g., Windmeijer and Santos Silva (1997)). Using numerical methods, we verified that $V_{GMM} > V_{ME}$, and $V_{ME}$ varies with $W$ in the nonlinear model. }\\

When the model is linear, i.e., $G(X_i) = \frac{\partial g(X_i, \theta)}{\partial\theta^{\prime}}$ does not depend on $\theta$, and $F(X_i, \theta) = \frac{\partial vec [G(X_i, \theta)^{\prime}]}{\partial\theta^{\prime}}= 0$,  we can further simplify the form of  $\Lambda^{*}_{W}$ and the asymptotic variance of the ME estimator. \\ 

\begin{corollary}\label{cor:optimal_matrix} 

If $g(X, \theta)$ is linear in $\theta$, we have
\begin{equation}
\Lambda^{*}_{W} = [G^{\prime} \Sigma_{11,2}^{-1} \quad - G^{\prime} \Sigma_{11, 2}^{-1} \Sigma_{12} \Sigma_{22}^{-1}]   \label{optimal_Lambda}
\end{equation}
where $G = E[ G(X_i)], \Sigma_{11, 2} = (\Sigma_{11} - \Sigma_{12} \Sigma_{22}^{-1} \Sigma_{21})$, $\Sigma (\theta_W) =\begin{bmatrix} \Sigma_{11} & \Sigma_{12} \\ \Sigma_{21}&\Sigma_{22}\end{bmatrix} $. The asymptotic variance of the ME estimator is 
\begin{equation}
V_{ME} (W)= (G^{\prime} \Sigma_{11, 2}^{-1}G)^{-1},\label{optimal_variance_linear}\\ 
\end{equation}
and $\Sigma_{11,2}$ is invariant to $\theta_W$.
\end{corollary}
\vspace{0.5cm}
\noindent

Under the linear model, Corollary \ref{cor:optimal_matrix} shows that the asymptotic variance of the ME estimator is $(G^{\prime} \Sigma_{11, 2}^{-1}G)^{-1}$, where $\Sigma_{11, 2} = \Sigma_{11} - \Sigma_{12} \Sigma_{22}^{-1} \Sigma_{21} $.  We refer to $W^* = \Sigma_{11,2}^{-1}$ as the \textit{misspecification-efficient} weighting matrix, in analogy with the conventional (correct-specification) efficient weighting matrix $W_0 = \Sigma_{11}^{-1}$. While $W_0$ is the inverse of the variance of the original moments $g(X_i,\theta)$, $W^* = \Sigma_{11,2}^{-1}(\theta)$ is the inverse of the variance of $g(X_i, \theta) - \Sigma_{12}\Sigma_{22}^{-1} vec(G(X_i)')$, which is the residualized moment after the projection onto the Jacobian. Even under correct specification, the ME bound $(G^{\prime} \Sigma_{11, 2}^{-1}G)^{-1}$ is strictly smaller than the asymptotic variance of the GMM estimator with the conventional efficient weighting matrix $W_0$ when $\Sigma_{12} \neq 0$ as we discussed.

Corollary \ref{cor:optimal_matrix} shows that the conventional efficient GMM variance formula $(G^{\prime} WG)^{-1}$ achieves the efficiency bounds of the oracle-ME estimator, when $W^* = \Sigma_{11,2}^{-1}(\theta_W)$. The caveat is that the GMM estimator itself with the choice $W^*$ does not achieve the oracle ME bound as shown in Corollary \ref{cor:dist_ME}, $V_{GMM}(W^*, \theta_W) \ge (G'\Sigma_{11,2}^{-1}G)^{-1} $ under misspecification. It follows that the conventional formula for efficient GMM is too ``small", even in large samples, as it will be valid for the ME estimator, but not  for the standard GMM estimator. Researchers have found that the standard errors of the efficient GMM are often severely downward biased.\footnote{Many researchers thus routinely use the finite-sample correction of Windmeijer (2005). Hwang et al. (2022) further show, under the linear model, that misspecification-robust GMM variance will not only be valid under misspecification, but also work as a finite-sample correction even if the model is correctly specified.} Our results show that reporting the conventional efficient GMM variance formula in the linear GMM will still be useful as an oracle efficiency bound; showing researchers what could be gained by imposing moments from the Jacobian $G(X_i, \theta)$, but should not be used directly for inference with the standard GMM itself under misspecification. When  $\Sigma_{12}=0$,  $W^* = W_0  = \Sigma_{11}^{-1}$, and the misspecification-efficient weighting matrix reduces to the conventional one.

We can easily report the oracle ME bound, $V_{ME} (W)=(G^{\prime} \Sigma_{11, 2}^{-1}G)^{-1}$, as an efficiency bound alongside conventional/misspecification-robust standard errors with no significant cost. We can consistently estimate $V_{ME} (W)= (G^{\prime} \Sigma_{11, 2}^{-1}G)^{-1}$ regardless of the model specification based on the sample estimator 
\[
\widehat\Sigma_{11,2} = \frac{1}{n}\sum_i (r_i(\widehat{\theta}) - \bar r_n(\widehat{\theta}))(r_i(\widehat{\theta}) - \bar r_n(\widehat{\theta}))^{\prime}
\]
with $r_i(\theta) = g_i(\theta) - \widehat\Sigma_{12}\widehat\Sigma_{22}^{-1}vec(G_i(\theta)'), \bar r_n(\theta) = 1/n \sum r_i(\theta)$ and the preliminary GMM estimator $\widehat{\theta} = \widehat{\theta}_{GMM} (W)$.\footnote{Under misspecification, we must use this centered covariance estimator because the uncentered covariance estimator $ \frac{1}{n}\sum_i r_i(\widehat{\theta}) r_i(\widehat{\theta})^{\prime}$ is not consistent for $\Sigma_{11,2}$.}

Furthermore, Corollary \ref{cor:optimal_matrix} shows that under the linear model, $V_{ME} = (G'\Sigma_{11, 2}^{-1}G)^{-1}$ is the same for all different pseudo-true values $\theta_W$, due to the invariance property of $\Sigma_{11, 2}$. This implies that  misspecification affects what is being estimated, but not how precisely it can be estimated. This is because the linear GMM model has a ``location-like" structure in the sense that $g(X_i,\theta) = g(X_i, 0) + G(X_i)\theta$, so shifting $\theta$ moves the moment vector along the Jacobian direction. The invariance property holds because $\Sigma_{11, 2}$ is the variance of the residual from the population regression of $g(X_i,\theta)$ on $vec(G(X_i)')$, so all the $\theta$-dependent parts are projected out. Practically, this implies that $V_{ME} $ in the linear model can be consistently estimated from any preliminary GMM without committing to a particular weight matrix.

In general nonlinear settings, above arguments fail to hold as the Jacobian $G(X_i, \theta)$, and  $\Sigma_{11,2}(\theta)$ change with $\theta$. However, the distribution of the oracle-ME estimator $\widehat{\theta}_{ME} (W) $ and the oracle bounds $V_{ME}(W)$ can be approximated by the bootstrap, and it can be viewed as a misspecification-robust and efficient version of the Hall and Horowitz (1996) recentered bootstrap; see Section \ref{sec:bootstrap}.

\begin{remark}[Decomposition of the ME efficiency gain]\label{rem:gap}
Consider the simple linear IV model, $Y_i = X_i \theta_0 + Z_i'\eta + \varepsilon_i, X_i = Z_i'\pi + v_i$, where $\theta_0$ is a scalar with  $m$ excluded instruments $Z_i$. Suppose $(\varepsilon_i, v_i)$ are conditionally homoskedastic given $Z_i$, and the first-stage and the direct effects of instruments are local, $\pi = O(n^{-1/2}), \eta = O(n^{-1/2})$. Then, we show (in Supplemental Appendix \ref{sec:decomp_ME}) that we have the following decomposition for $W \propto Q_{ZZ}^{-1}, Q_{ZZ} = E[Z_iZ_i']$ (which covers TSLS, $W_0, W^*$ under homoskedasticity) 
\begin{equation}\label{eq:gain_decomp}
\frac{V_{GMM}(W, \theta_W)}{V_{ME}} -1 \approx  \frac{\rho^2}{1-\rho^2}  + \frac{J_n}{\mu_n^2 (1-\rho^2)}
\end{equation}
where $\rho$ is the correlation between $\varepsilon_{W} = Y_i - X_i\theta_W$ and $X_i$, $J_n = n g(\theta_W)'\Sigma_{11}^{-1} g(\theta_W)$, and $\mu_n^2 = n \pi'Q_{ZZ}\pi/\sigma_v^2$. We can report this measure approximating $\mu_n^2$ by $mF$ with the number of excluded instruments $m$ and the first-stage F statistic, and $J_n$ by the reported Hansen $J$ statistics. So, the ME efficiency gain depends on the moment-Jacobian correlation and  also on the degree of misspecification relative to the strength of identification, rather than the degree of misspecification alone. We also show that the misspecification-robust SE exceeds the conventional SE by the factor $\sqrt{1 + J_n/(mF)}$, while the ME bound lies below the conventional SE by the factor $\sqrt{1-\rho^2}$. See Supplemental Appendix \ref{sec:decomp_ME} for detailed derivations and Section \ref{sec:empirical} for the empirical results and how the decomposition in \eqref{eq:gain_decomp} can be helpful in these examples.
\end{remark}

\section{Misspecification-Efficient Bootstrap GMM Estimators}
\label{sec:bootstrap}

We first consider a related class of estimators, misspecification-efficient (ME) GMM based on the augmented moment conditions. We then propose the bootstrap GMM estimator, which can be considered as a misspecification-robust and efficient version of the Hall and Horowitz (1996) recentered bootstrap GMM estimator.

The ME-GMM estimator, $\widehat{\theta}_{ME-GMM} (\Delta)$, solves
\begin{equation}\label{def:ME-GMM}
\widehat{\theta}_{ME-GMM} (\Delta) = \argmin_{\theta}  \Biggl(n^{-1}\sum_{i=1}^{n} \widetilde{\psi}(X_i, \theta) \Biggr)^{\prime}  \Delta \Biggl(n^{-1}\sum_{i=1}^{n} \widetilde{\psi}(X_i, \theta) \Biggr)
\end{equation}\vspace{-0.3cm}
where $\Delta >0$ is a fixed weighting matrix, and 
$
\widetilde{\psi}(X_i, \theta)  = \begin{pmatrix}
g(X_i, \theta) - E[g(X_i, \theta_W)]\\
 vec(G(X_i, \theta)^{\prime}) - E[vec (G(X_i, \theta_{W})^{\prime})]
\end{pmatrix}.
$
Under the same assumptions (Assumption \ref{assump:AM}.2-\ref{assump:AM}.4), we can easily show that the ME-GMM with the efficient weighting matrix $\Delta = E[\widetilde{\psi}(X_i, \theta_W) \widetilde{\psi}(X_i, \theta_W)^{\prime}]^{-1} = \Sigma(\theta_W)^{-1}$ is asymptotically equivalent to the ME estimator, $\widehat{\theta}_{ME} (W)$.  We note that Lee and Liao (2018) consider the same GMM problem under correct specification $(E[g(X_i, \theta_W)] = 0)$ and a singular Jacobian $(E[G(X_i, \theta_{W})]  =0) $, where the recentering terms are known to be zero.

Here, we consider the bootstrap GMM estimator which mimics \eqref{def:ME-GMM} with the efficient weighting matrix $\Delta = \Sigma(\theta_W)^{-1}$.  Let  $\{X_i^*\}_{i=1}^n$ be the bootstrap samples, drawn independently  from the original sample $\{X_i\}_{i=1}^n$ with replacement.  With re-centering, the bootstrap ME-GMM estimator solves
\begin{equation}\label{def:ME-GMM_bootstrap}
\widehat{\theta}_{ME}^{*}(W) = \argmin_{\theta}  \Biggl(n^{-1}\sum_{i=1}^{n} \widetilde{\psi}^{*}(X_i^{*}, \theta) \Biggr)^{\prime}  \Delta^{*} \Biggl(n^{-1}\sum_{i=1}^{n} \widetilde{\psi}^{*}(X_i^{*}, \theta) \Biggr),
\end{equation}
where
\begin{equation*}
\widetilde{\psi}^{*}(X_i^{*}, \theta)  = \begin{pmatrix}
g(X_i^{*}, \theta) - g_n(\widehat{\theta}_W)\\
 vec(G(X_i^{*}, \theta)^{\prime}) - vec (G_n(\widehat{\theta}_W)^{\prime})
\end{pmatrix}, 
\end{equation*}
\[
g_n(\theta) = n^{-1}\sum_{i=1}^{n} g (X_i, \theta), \quad G_n(\theta) = n^{-1}\sum_{i=1}^{n} G (X_i, \theta), 
\]
\begin{equation}
\Delta^{*} = \Biggl(E^{*} [ \psi(X_i^{*}, \widehat{\theta}_W) \psi(X_i^{*}, \widehat{\theta}_W)^{\prime}]  -E^{*} [ \psi(X_i^{*}, \widehat{\theta}_W)] E^{*} [ \psi(X_i^{*}, \widehat{\theta}_W)]^{\prime}\Biggr)^{-1},
\end{equation}
and $\widehat{\theta}_W = \widehat{\theta}_{GMM}(W)$ is the  GMM estimator.  The recentering by $g_n(\widehat{\theta}_W), vec (G_n(\widehat{\theta}_W))$ ensures $E^*[g(X_i^*,\widehat{\theta}_W)  - g_n(\widehat{\theta}_W)] = 0$,  and the bootstrap moments  $\widetilde{\psi}^{*}(X_i^{*}, \theta)$ have mean zero at $\widehat{\theta}_W$ conditional on the data. $\Delta^{*}$ is a symmetric positive definite weighting matrix that is a bootstrap sample version of the efficient weighting matrix $\Delta = \Sigma(\theta_W)^{-1}$.

The proposed bootstrap estimator utilizes the recentering of the Jacobian estimator to ensure that the bootstrap expectation of the recentered Jacobian moments evaluated at the estimator is zero, which is essential to correctly mimic the  asymptotic distribution under misspecification. This ``double recentering" idea has also been considered in weak-identification literature to accurately reflect the behavior of the original Jacobian (e.g., Dovonon and Gon\c{c}alves (2017)), yet under correct specification.

The standard nonparametric bootstrap without recentering provides first-order valid inference under correct specification (Hahn (1996)), however, the standard bootstrap does not achieve an asymptotic refinement in overidentified models. This is because the bootstrap sample mean of the moment function evaluated at the estimator, $E^{*}[g(X_i^*, \widehat{\theta}_W)] = g_n(\widehat{\theta}_W)$, is not necessarily equal to zero, and the recentering has been considered important to achieve higher-order improvements under correct specification (Hall and Horowitz (1996), Brown and Newey (2002)).  Hall and Horowitz (1996, HH hereafter) proposed a recentered bootstrap  GMM estimator,
\begin{equation}
\widehat{\theta}_{HH}^{*} = \argmin_{\theta}  \Biggl(n^{-1}\sum_{i=1}^{n} g(X_i^{*}, \theta) - g_n(\widehat{\theta}_W) \Biggr)^{\prime}  W_{0}^{*} \Biggl(n^{-1}\sum_{i=1}^{n} g(X_i^{*}, \theta) - g_n(\widehat{\theta}_W) \Biggr)
\end{equation}
where $W_0^*$ is the bootstrap version of the efficient weighting matrix under correct specification, $W_0 = \Sigma_{11} (\theta)^{-1} $. $\widehat{\theta}_{HH}^{*} $ can be considered as a special case of the proposed bootstrap ME-GMM estimator $\widehat{\theta}_{ME}^{*}$  by ignoring the Jacobian augmented parts and only recentering the original sample moments, i.e.,  $\Delta^{*}_{11} =W_0^*, \Delta_{12}^{*} =  0$, and  $\Delta_{22}^{*} = 0$. 

The conventional GMM and the HH bootstrap estimator would have different asymptotic distributions under misspecification because the Hall-Horowitz recentered bootstrap always creates a correctly specified bootstrap world, regardless of whether the actual model is misspecified or not.  The Hall-Horowitz bootstrap eliminates the additional source of Jacobian sampling variability that arises under misspecification, while the proposed ME bootstrap can consider these variations, and furthermore it considers more efficient combinations of the sampling uncertainty for the original and Jacobian moments. In the linear IV model, the bootstrap ME-GMM estimator has a closed form and is an efficient linear combination of the HH recentered bootstrap TSLS estimator and the recentered Jacobian; see Supplemental Appendix \ref{sec:linear_iv}.

We show below that the asymptotic distribution of the bootstrap estimator $\widehat{\theta}_{ME}^{*}(W)$ is identical to that of the original ME estimator $\widehat{\theta}_{ME}(W) $. We consider the following assumptions for the bootstrap validity.  Define bootstrap sample mean of moments and Jacobian as $
g_n^*(\theta) = n^{-1}\sum_{i=1}^{n} g (X_i^*, \theta), G_n^*(\theta) = n^{-1}\sum_{i=1}^{n} G (X_i^*, \theta)$.

\begin{assumption}\label{assump:boot}
\begin{enumerate}
\item $\widehat{\theta}_{ME}^{*} (W) \overset{p^{*}}{\rightarrow}  \theta_W.$
\item $\sup_{\theta \in \Theta} \| g_n^*(\theta) - g_n(\theta)\| \overset{p^{*}}{\rightarrow} 0$, and  $\sup_{\theta \in \Theta} \| G_n^*(\theta) - G_n(\theta)\| \overset{p^{*}}{\rightarrow} 0$.
\item  $\sqrt{n}\big(g_n^*(\widehat{\theta}_W) - g_n(\widehat{\theta}_W),\; \mathrm{vec}(G_n^*(\widehat{\theta}_W)' - G_n(\widehat{\theta}_W)')\big) \overset{d^{*}}{\rightarrow} N(0,\Sigma(\theta_W))$, where $ \widehat{\theta}_W = \widehat{\theta}_{GMM}(W)$, and $\Sigma(\theta_W)$ is the same covariance matrix as in Theorem \ref{thm:dist_AM}.
\end{enumerate}
\end{assumption}

These conditions are standard assumptions under i.i.d.  sampling; see, e.g., Gin\'{e} and Zinn (1990), and Hahn (1996).  Assumption \ref{assump:boot}.1 on the bootstrap consistency is a high-level assumption, which can be verified from more primitive conditions using the standard bootstrap consistency of the extremum estimators, combined with the uniform convergence conditions in Assumption \ref{assump:boot}.2. Assumption \ref{assump:boot}.3 is implied by Assumptions \ref{assump:AM}.1-\ref{assump:AM}.4 under the i.i.d. setup.

\begin{theorem}\label{thm:bootME} Suppose that Assumptions \ref{assump:AM}.1-\ref{assump:AM}.4 hold. Further, we assume that $\theta = \theta_W$ is a unique solution to $E[\widetilde{\psi}(X_i, \theta)] = 0 $, and $\Delta$  is positive definite. Then, we have
\begin{align*}
\sqrt{n}(\widehat{\theta}_{ME-GMM} (\Delta)- \theta_{W} ) \overset{d}{\longrightarrow}  N(0,(\Gamma(\theta_W)^{\prime} \Delta \Gamma (\theta_W))^{-1} \Gamma(\theta_W)^{\prime} \Delta \Sigma (\theta_{W} )  \Delta \Gamma(\theta_W)( \Gamma(\theta_W)^{\prime} \Delta \Gamma (\theta_W))^{-1}).
\end{align*}
When $\Delta = \Sigma(\theta_W)^{-1}$, the asymptotic variance reduces to $(\Gamma(\theta_W)^{\prime} \Sigma (\theta_W)^{-1}\Gamma(\theta_W))^{-1}$, provided it is non-singular. In addition, suppose that Assumption \ref{assump:boot} holds, then
\begin{align*}
\sqrt{n}(\widehat{\theta}_{ME}^*(W) - \widehat{\theta}_{GMM} (W) ) \overset{d^*}{\longrightarrow}  N(0,(\Gamma(\theta_W)^{\prime} \Sigma (\theta_W)^{-1}\Gamma(\theta_W))^{-1}).
\end{align*}
\end{theorem}

Theorem \ref{thm:bootME} shows that the asymptotic distribution of the bootstrap estimator $\sqrt{n}(\widehat{\theta}_{ME}^{*}(W)  - \widehat{\theta}_{GMM}(W))$ is equal to that of $\sqrt{n}(\widehat{\theta}_{ME} (W)- \theta_{W} )$. Theorem \ref{thm:bootME} allows us to approximate the sampling variation of the oracle-ME estimator $\widehat{\theta}_{ME} (W)$, and the theorem justifies the use of the bootstrap to obtain the minimax bounds that will be established in Section \ref{sec:semi_efficiency} (Theorem \ref{thm:semiparametric_efficiency}).

\subsection{Double-Recentered Bootstrap Estimator}
\label{sec:MRHH}
We note, however, that Theorem \ref{thm:bootME} does not guarantee the validity of the percentile bootstrap CI based on $\widehat{\theta}_{ME}^{*}(W) $ because it is centered at $\widehat{\theta}_{GMM}(W)$, which has a different asymptotic distribution. Using the double-recentering idea from the above section, we propose a double-recentered (DR) bootstrap that jointly perturbs the moment conditions and the Jacobian. The DR bootstrap correctly approximates the asymptotic distribution of the GMM estimator regardless of specification status, nesting the Hall and Horowitz (1996) bootstrap as a special case under correct specification. 

Lee (2014) shows that the standard nonparametric GMM bootstrap with the analytic misspecification-robust variance estimator works regardless of misspecification unlike Hall and Horowitz (1996), and argues that the recentering can be detrimental. In our paper, however, we clarify this argument that additional considerations of re-centered Jacobian can achieve asymptotic validity under misspecification. While the standard GMM bootstrap estimator without recentering will be valid combined with the analytic misspecification-robust SE, the Hall and Horowitz (1996) bootstrap does not achieve asymptotic validity under misspecification because it  \textit{only} recenters the original moment functions.\footnote{Since the main focus of the paper is on the efficiency results for $\theta_W$ and minimax bounds, we only considered the GMM estimator with a fixed $W$. Thus, we do not need to consider variations from the estimated weight matrix, which was typically considered for the construction of the misspecification-robust variance in the literature. Although it is beyond the scope of this paper, the idea of the ``double recentering" in this paper can be easily  extended to the  ``triple recentering" with the estimated $W_n$ (one-step) or the $W_n(\theta)$ (two-step) case by augmenting the weighting matrix into the $\psi(\cdot)$. }

We define the double-recentered bootstrap estimator  $\widehat{\theta}_{DR}^{*}$ as the solution to the bootstrap analog of the augmented estimating equation \eqref{eqn_mestimation_sample} with the particular choice of  $\Lambda  = A(W, \theta_W)= [G(\theta_W)'W,\; g(\theta_W)'W \otimes I_p] $:

\begin{equation}\label{eq:drhh}
A_n(W, \widehat{\theta}_{W})\widetilde{\psi}_n^* (\widehat{\theta}_{DR}^{*})  = 0,
\end{equation}
where 
\begin{equation*}
\widetilde{\psi}_n^* (\theta) = \frac{1}{n} \sum_{i=1}^{n} \widetilde{\psi}^{*}(X_i^{*}, \theta)  = \begin{pmatrix}
g_n^*( \theta) - g_n(\widehat{\theta}_W)\\
 vec(G_n^{*}(\theta)^{\prime}) - vec (G_n(\widehat{\theta}_W)^{\prime})
\end{pmatrix}, 
\end{equation*}
and 
\[
A_n(W, \widehat{\theta}_{W})= \big[G_n(\widehat{\theta}_{W})'W\;,\;\; g_n(\widehat{\theta}_{W})'W \otimes I_p\big].
\]
While the HH bootstrap only recenters the original moments $g_n^*$ at $g_n(\widehat{\theta}_{W})$, the DR bootstrap recenters the full augmented vector $\psi_n^*$ at its sample counterpart $\psi_n$, including Jacobian estimation variability, appropriately weighted by  
$\Lambda_n = A_n(W, \widehat{\theta}_{W})$ to mimic the distribution of  $\widehat{\theta}_{W}$.\footnote{We note that the DR bootstrap estimator is different from the bootstrap M-estimator based on the just-identified moment functions with an augmented parameter as in Imbens (1997). One can show that the bootstrap M-estimator based on the just-identified augmented moment functions corresponds to the standard GMM bootstrap estimator without recentering.} 

For the implementation, we first compute the original sample GMM estimator $\widehat{\theta}_{W}$, and calculate $A_n(W, \widehat{\theta}_W)$, which requires $g_n(\widehat{\theta}_{W})$ and $G_n(\widehat{\theta}_{W})$.  Then, for each bootstrap replication $b = 1, \ldots, B$: 
\begin{enumerate}
\item Draw $\{X_i^*\}_{i=1}^n$ with replacement; 
\item With $A_n(W, \widehat{\theta}_W)$, solve (\ref{eq:drhh}) for $\widehat{\theta}_{DR}^{*}$ via Newton--Raphson initialized at $\widehat{\theta}_{W}$. For the linear model, $\widehat{\theta}_{DR}^{*}$ has a closed form (see Supplemental Appendix \ref{sec:linear_iv});
\item Confidence intervals are constructed as standard percentile or percentile-$t$ intervals based on $t^* = \frac{\widehat{\theta}_{DR}^{*} -\widehat{\theta}_{W}}{s.e.(\widehat{\theta}_{DR}^{*} )}$ from the $B$ bootstrap draws.
\end{enumerate}

\begin{corollary}\label{cor:bootDRHH} Suppose that Assumptions \ref{assump:AM}.1-\ref{assump:AM}.5, and Assumption \ref{assump:boot} hold (replacing consistency Assumption \ref{assump:boot}.1 with $\widehat{\theta}_{DR}^*$). Then, 
\begin{align*}
\sqrt{n}(\widehat{\theta}_{DR}^* - \widehat{\theta}_{GMM} (W) ) \overset{d^*}{\longrightarrow}  N(0,V_{GMM}(W, \theta_W)), 
\end{align*}
where $V_{GMM}(W, \theta) $ is defined in \eqref{asymptotic_dist}.
\end{corollary}

We recommend $\widehat{\theta}_{DR}^* $ and the percentile  confidence interval for valid inference for $\theta_W$. In addition, we recommend reporting separately  $\widehat{V}_{ME}(W) $ and $\sup_{W \in \mathcal{W}} \widehat{V}_{ME}(W)$ based on the recentered bootstrap ME estimator $\widehat{\theta}_{ME}^*$ or using analytical ME estimates with the standard GMM point estimate $\widehat{\theta}_{GMM}(W)$ as evidence of the attainable oracle bound by showing researchers what could be gained if a feasible ME point estimate were available with known misspecification and identification. This provides honest bounds that acknowledge the researcher doesn't know which $W$ (and hence $\theta_W$) is ``right."

\section{Semiparametric Efficiency under Misspecification}\label{sec:semi_efficiency}

Based on the results from the previous sections, we derive the uniform (over $W\in \mathcal{W})$ asymptotic minimax bounds of the misspecified moment condition models.  We can generally view the misspecified GMM model as a semiparametric model such that we consider a family of probability distributions of the observed data $X$, indexed by the finite-dimensional parameter of interest $\theta$ for some weighting matrix $W$. In practice, researchers can choose different weighting matrices $W \in \mathcal{W}$ based on the application they consider, where $\mathcal{W}$ denotes a class of $m \times m$ positive definite matrices.

However,  the notion of semiparametric efficiency and optimality needs to be defined differently under (global) misspecification. As noted earlier, the fundamental difficulty here is that the pseudo-true value $\theta_W$ changes with $W \in \mathcal{W}$ in our misspecified moment condition models.  The semiparametric efficiency framework typically  assumes the existence of a well-defined ``true" parameter, and the efficiency bound is computed at this target parameter.  For example, the GMM estimator is semiparametrically efficient under correct specification (Chamberlain (1987)) for $\theta = \theta_0$ under the moment conditions $E[g (X_i, \theta)]  = 0$ with an ``efficient" weighting matrix $W = \Sigma_{11} (\theta_0)^{-1}$, the inverse of the variance matrix of the original moments.

We may compare two or more different weighting matrices, and may want to choose a ``better" $W$ and pseudo-true $\theta_W$ than the others, based on the  pseudo-distance measure such as $J$ statistics, or based on the goodness-of-fit measures or the information criteria (e.g., Rivers and Vuong (2002), Marmer and Otsu (2012)).  We can also consider  the optimal weights $W^{\dagger} = \argmin_{W} V_{GMM} (W, \theta_{W})$ that minimize the asymptotic variance of misspecified-GMM in \eqref{asymptotic_dist} or $W^{\dagger} = \argmin_{W} V_{ME}(W)$ that minimizes the asymptotic variance of misspecification-efficient GMM in \eqref{optimal_variance}.  However, as noted in Hall and Pelletier (2011), economic theory dictates an appropriate choice of the weighting matrix in some cases and the researcher may choose $W$ based on which $\theta_W$ they care about substantively, and in the absence of these economic considerations, the choice of $W$ and the relative ranking over $\mathcal{W}$ can be arbitrary. 

In this paper, we treat $\mathcal{W}$ as fixed, but in practice $\mathcal{W}$ can be chosen by the researchers. For example, researchers may focus only on the identity matrix and diagonal matrices due to computational costs, or may impose some restrictions to obtain particular causal interpretations.
\vspace{0.4cm}

\noindent \textbf{Example: Linear IV (continued)}. Suppose we focus on a scalar $\theta$ ($p=1$), instruments $Z$ consist of $m$ single instruments $Z_{1}, \cdots Z_{m}$. Then,  the GMM pseudo-true value can be written as a linear combination of the single-instrument IV estimands $\theta_j =E[Z_j^{\prime} Y]/E[Z_j^{\prime} X] $,
\[
\theta_W = \frac{E[X'Z] W E[Z'Y]}{E[X'Z] W E[Z'X]} = \sum_{j=1}^{m} w_j \theta_{j}
\]
where the weight $w_j$ for the $j$-th instrument is given by:
\[
w_j = \frac{\pi' W e_j e_j' \pi}{\pi' W \pi} =\frac{( \sum_{k=1}^{m} W_{jk} \pi_k) \pi_j}{{\pi' W \pi}}
\]
with the $j$-th unit vector $e_j$ and $\pi = E[Z^{\prime}X]$. The weights sum to one, but $w_j$ is negative if $\pi_j$ and $ \sum_{k=1}^{m} W_{jk} \pi_k$ have opposite signs. When $W = E[Z'Z]^{-1}$,  $\theta_W$ is the TSLS estimand and Angrist and Imbens (1995) and Koles\'{a}r (2013) characterize $\theta_W$ as a weighted average of LATE under monotonicity conditions (see also Andrews (2019)). However, the TSLS  weights are not guaranteed to be positive. Researchers may restrict the weighting matrices to the class with positive weights,
\[
\mathcal{W}^{+} = \{ W \in \mathcal{W}: \pi^{\prime} W e_j e_j' \pi \geq 0 \textnormal{ for all } j=1, \cdots, m, \pi = E[Z'X]  \}, 
\]
which preserves the sign of the instrument strength, $\text{sign}\left( (W \pi)_j \right) = \text{sign}(\pi_j)$, and includes any  diagonal positive definite matrix (e.g., identity, or $\text{diag}((E[Z'Z])^{-1})$). The TSLS weighting matrix is in this class only under strict conditions when we have covariates (e.g., Blandhol et al. (2025) and S\l oczy\'{n}ski (2024)). Under treatment effect heterogeneity,  $W 
\in \mathcal{W}^{+}$ guarantees  that $\theta_W$ lies  between the minimum and maximum of the single-instrument LATE estimands $\theta_j$, and the bounds below are uniform over this class of causal estimands.

\subsection{Uniform Asymptotic Minimax Bounds}
\label{sec:uniform_bounds}

To formally state our results, we first define the class of distributions and the parameter spaces we consider.  We also provide more specific definitions of correct/misspecified models. 

Suppose that we observe an i.i.d. sample $\{ X_i \}_{i=1}^{n}$ drawn from the true (unknown) probability distribution $F_0$. 
Let $\mathcal{F}$ be the space of all probability distributions and define $\mathcal{F}^{C}$ as follows;
\begin{eqnarray}
\mathcal{F}^{C} &=&  \{ (F, \theta) \in \mathcal{F} \times \Theta:  (i)\phantom{a}  \theta \textnormal{ is in an open set } \Theta \subset \mathbb{R}^{p} \textnormal{ such that  } g(X, \beta) \textnormal{ and } \partial g(X, \beta)/\partial \beta^{\prime} \nonumber \\ 
 && \quad \textnormal{ are continuous   for  }(X, \beta) \in \mathcal{X} \times \Theta; (ii) \phantom{a} E_F g(X, \theta) = 0; \label{def:sets_correct} \\
 && \quad (iii) \phantom{a} E_F [g(X, \theta) g(X, \theta)^{\prime}] \textnormal{ is positive-definite}; (iv)\phantom{a}  rank \phantom{a} E_F[\partial g(X, \theta)/\partial \theta^{\prime}] = p \} , \nonumber
\end{eqnarray}
i.e., $\mathcal{F}^{C}$ is the  set of probability distributions $F$, such that for some $\theta \in \Theta$, $( F, \theta)$ satisfies the regularity and the moment conditions in $\mathcal{F}^{C} $.  Correct specification assumes $(F_0, \theta_0) \in \mathcal{F}^{C}$ such that the true probability distribution $F_0$  with unique parameter value $\theta_0$ satisfies conditions above. Chamberlain (1987) shows that $\widehat{\theta}_{GMM}( \Sigma_{11} (\theta_0)^{-1})$ achieves semiparametric efficiency bounds in the sense of H\'{a}jek's (1972) local asymptotic minimaxity for the neighborhood  of $(F_0, \theta_0) \in \mathcal{F}^{C}$.

In this paper, we consider the misspecified moment condition models and the pseudo-true values associated with the choice of weighting matrix $W$. We define that the moment function $g(\cdot)$ is misspecified if 
\[
(F_0, \theta) \notin \mathcal{F}^{C} \textnormal{ for all } \theta \in \Theta.
\]

Then, for any $W \in \mathcal{W}$, where $\mathcal{W}$ is a set of positive definite matrices, we can define the pseudo-true values  as the minimizer of a goodness-of-fit measure of moment condition $g(\cdot)$ to the true data distributions $F_0$ using the weighted-norm associated with the weighting matrix $W$;
\[
 \theta_{W} = \argmin_{\theta} E_{F_0} g(X_i, \theta)'WE_{F_0} g(X_i, \theta). 
\]

To consider the uniform minimax bounds, we define the extended class of models $\mathcal{G}$ as the set of tuples, $\mathcal{G} = \{ (W, (F, \theta)) : W \in \mathcal{W}, (F, \theta) \in \mathcal{F}^{W} \}$, where $\mathcal{F}^{W}$ is the space of probability distributions that is consistent with $\theta_W$, i.e., the set of probability distributions $F$, such that for some $\theta \in \Theta$, $(F, \theta)$ satisfies the following regularity conditions, 
\begin{eqnarray}
\mathcal{F}^{W} &=&   \{ (F, \theta) \in \mathcal{F} \times \Theta:  (i) \phantom{a}  \theta \textnormal{ is in an open set } \Theta \subset \mathbb{R}^{p} \textnormal{ such that  } \psi(X, \beta) \textnormal{ and } \partial \psi(X, \beta)/\partial \beta^{\prime} \nonumber \\ 
 && \quad \textnormal{ are continuous   for  }(X, \beta) \in \mathcal{X} \times \Theta ;  (ii) \phantom{a} E_F \tilde{\psi}(X, \theta) = 0, \textnormal{ where }  \nonumber \\
 && \quad \tilde{\psi}(X, \theta) = \psi(X, \theta) - E_{F_0}[\psi(X, \theta_W) ],  \theta_W = \argmin_{\theta} E_{F_0} g(X_i, \theta)'WE_{F_0} g(X_i, \theta),  \label{def:sets_misspec}\\
 &&  \quad \textnormal{assuming it exists and unique};  (iii) \phantom{a} E_F [\tilde{\psi}(X, \theta) \tilde{\psi}(X, \theta)^{\prime}] \textnormal{ is positive-definite}; 
\nonumber \\
 && \quad  (iv) \phantom{a} rank \phantom{a} E_F[\partial \psi(X, \theta)/\partial \theta^{\prime}] = p \}. \nonumber 
\end{eqnarray}

Under the same identification assumption (Assumption \ref{assump:AM}), we have  $(F_0, \theta_W) \in \mathcal{F}^{W}$. Condition $(i)$  requires that the original function $g(X, \theta)$ is twice continuously differentiable. Condition $(iv)$ requires that $\Gamma(\theta) = [G(\theta); F(\theta)]$ has rank $p$. This is a weaker condition than standard GMM identification (condition $(iv)$ in $\mathcal{F}^C$) because it allows identification to come from the curvature ($F(\theta)$) in addition to the first-order derivative ($G(\theta)$). $\mathcal{F}^{W} $ is defined through a recentering that depends on the pseudo-true value $\theta_W$ (and thus implicitly on $F_0$), and this is innocuous for the local asymptotic minimax analysis; it concerns the minimax risk over neighborhoods of $(F_0,  \theta_W)$, and within such a neighborhood, the recentering is fixed.

We introduce further notation on a neighborhood of the probability distribution and the loss function $l(\cdot)$. Let $\Pi$ be the set of all probability measures $F$ defined on the Borel sets in a Euclidean space with the support of a probability measure $\mathcal{X}$. Given an $F \in \Pi$, the basic neighborhoods of $F$ are sets of the form:
\[
\left\{Q \in \Pi : \left|\int f_{j}dQ - \int f_{j}dF\right| < \epsilon_{j}, j=1,\dots,k\right\}
\]
where $\epsilon_{j} > 0$, $k$ is some positive integer, and $f_j: \mathcal{X} \rightarrow \mathbb{R}$ are measurable functions such that $\int |f_{j}| dF < \infty$. An arbitrary neighborhood of $F$ is formed by taking unions of sets of this form. Let $\mathcal{L}$ be the class of loss functions $l \in \mathcal{L}$ that satisfy the following conditions for all $u, v \in \mathbb{R}$: (i) $l(u)=l(|u|)$, (ii) $|u| \le |v|$ implies $l(u) \le l(v)$, (iii) $\int_{-\infty}^{\infty}l(u)\exp\left(-\frac{1}{2}\lambda u^{2}\right)du < \infty$ for all $\lambda > 0$, (iv) $l(0)=0$. Many standard loss functions satisfy these conditions, for example, quadratic $l(u) = u^2$, absolute function $l(u) = |u|$, and polynomial loss $l(u) = |u|^p$ for $ p > 0$.

The next theorem shows the semiparametric efficiency bounds in the sense of H\'{a}jek's (1972) local asymptotic minimaxity  for the neighborhood  of $(W, (F_0, \theta_W)) \in \mathcal{G} = \{ (W, (F, \theta)) : W \in \mathcal{W}, (F, \theta) \in \mathcal{F}^{W} \}$.

\begin{theorem}\label{thm:semiparametric_efficiency}
Suppose that $(F_{0},\theta_W)$ satisfies the conditions in $\mathcal{F}^W$ for all $W\in \mathcal{W}$. Let $\Delta^{W}$ be any neighborhood of $(F_{0},\theta_W)$, and let $\Gamma^{W}$ be the subset of $\Delta^W$ such that the conditions in $\mathcal{F}^W$ are satisfied for all $(F,\theta) \in \Gamma^W$. Let $\mathcal{G}^{\Gamma} = \{ (W, (F, \theta)) : W \in \mathcal{W}, (F, \theta) \in \Gamma^W \}$. For any pair $(F, \theta) \in \Gamma^W$, let $\theta_{1}$ be the first component of $\theta$ and let $T_{n}(X_{1},\dots,X_{n})$ be any (measurable) estimator for $\theta_{1}$. Then for any loss function $l(\cdot) \in \mathcal{L}$, we have
\[
\liminf_{n\rightarrow\infty} \sup_{(W, (F,\theta))\in \mathcal{G}^{\Gamma}} E_{F}\left\{l\left[\sqrt{n}(T_{n}-\theta_{1})\right]\right\} \ge \ \int_{-\infty}^{\infty}l(\sigma u)d\Phi(u),
\]
where $\Phi(u)$ is the cumulative distribution function of the standard normal distribution, $\sigma^2 = \sup_W \sigma_{W}^2 <\infty$, $\sigma_{W}^2$ is the (1,1) element of $V_{ME}(W) = (\Gamma(\theta_W)^{\prime} \Sigma (\theta_W)^{-1}\Gamma(\theta_W))^{-1}$.
\end{theorem}
\vspace{0.2cm}
Theorem \ref{thm:semiparametric_efficiency} states that the asymptotic variance of any semiparametric estimator under the model class  $\mathcal{G} = \{ (W, (F, \theta)) : W \in \mathcal{W}, (F, \theta) \in \mathcal{F}^{W} \}$ is no smaller than the bounds with the least favorable weighting matrix $W$ that maximizes risk over the class $\mathcal{W}$.  This bound can be achieved by the worst-case variance of the ME estimator $\widehat{\theta}_{ME} (W), W \in \mathcal{W}$, which  can be feasibly approximated by the analytic formula or the bootstrap method. Since $V_{ME}$ doesn't depend on $W$ in the linear IV model, the ME estimator eliminates the question of which $W$ is most efficient among $W \in \mathcal{W}$, and the uniformity can be obtained with no extra cost. 

The bound in Theorem \ref{thm:semiparametric_efficiency} is based on the model class $\mathcal{F}^{W} $, where the recentering part $E_{F_0}[\psi(X, \theta_W)]$ is held fixed.  One can straightforwardly show that the misspecified-GMM estimator achieves a different semiparametric efficiency bound $V_{GMM}(W, \theta_{W})$ in a class of models that restricts $F$ only through the definition of $\theta_W$ (by replacing moment condition ($ii$) in $\mathcal{F}^C$ (equation \eqref{def:sets_correct}) with the local first order condition $E_F[ \partial g(X, \theta)/\partial \theta']^{\prime}W E_F [g(X, \theta)] = 0$); see also Mukhin (2019). Equivalently, one can use the augmented parameter space approach as in Imbens (1997) with the regularity (smoothness) conditions to obtain the same  efficiency bound $V_{GMM}(W, \theta_{W})$. However, as shown in Corollary \ref{cor:dist_ME},  the optimal linear combinations of the augmented moments $\Lambda^*_{W} = \Gamma(\theta_W)^{\prime} \Sigma (\theta_W)^{-1}$ exploit the correlation between the moments and the Jacobian, which the standard GMM combination $\Lambda   = [G(\theta_W)^{\prime} W \quad g(\theta_W)^{\prime} W \otimes I_p]$, dictated by the FOC of the GMM objective, ignores. Using the moment vector $\widetilde{\psi}(X, \theta)$ provides an over-identified system for $\theta$, exploiting more information from the degree of misspecification and the identification strength. 

\section{Monte Carlo Simulation}
\label{sec:simulation}

We compare the finite-sample performance of the standard GMM and the proposed misspecification-efficient estimator in a linear IV model. The data generating process is
\begin{equation}
Y_i = X_i \theta + Z_i'\eta +\varepsilon_i, \qquad X_i = Z_i'\pi + v_i,
\end{equation}
with $Z_i = (Z_{1i},Z_{2i})' \sim_{iid} N(0,I_2)$ and $(\varepsilon_i, v_i) \sim N(0,[1, 0.5; 0.5, 1])$. We set $\theta=1$, and vary the ``degree of misspecification" through the direct effect $\eta=(0,\delta)'$ with $\delta \in \{0,0.5,1,2\}$; $\delta=0$ is the correctly specified benchmark, $\delta > 0$ introduces an exclusion violation through $Z_2$. The sample sizes are $n \in \{200,500, 1000\}$. We set the first-stage coefficient $\pi=(\tilde\pi,2\tilde\pi)'$ so that the scaled concentration parameter ($\mu^2/m$)  is 50 or 10, corresponding to moderately strong and weak instrument setups for $n=200$. The pseudo-true value $\theta_W$ is $W$-dependent and moves away from $\theta=1$ as $\delta$ changes. We consider $W = I,  \widehat\Sigma_{11,2}^{-1}$, and the moderately strong instrument case here in the main paper, and the weak instrument design results, together with additional results under the conventional optimal weighting $W = \widehat\Sigma_{11}^{-1}$ are reported in Supplemental Appendix \ref{sec:app_weak_sim}. For each estimator, we report the standard deviation (SD), empirical $95\%$ coverage (Cov), and median CI length (Len) over $2,000$ Monte Carlo replications, using $B=1,000$ bootstrap draws. SD is normalized relative to the SD of standard GMM with the correctly specified benchmark $(n=200, \delta=0)$.

Table \ref{tab:sim} reports results for the one-step $(W = I)$ and the two-step GMM $(W = \widehat\Sigma_{11,2}^{-1})$. For each $(n,\delta)$, we report the following estimators and standard errors (SE) with Wald CI: (i) standard GMM estimator with the conventional SE; (ii) GMM with the misspecification-robust SE; (iii) the oracle ME estimator with efficient  $V_{ME}(W)$.  We also consider the following bootstrap estimators: (iv) HH (Hall and Horowitz (1996)) percentile bootstrap CI; and (v) the DR (double-recentered) percentile bootstrap CI proposed in Section \ref{sec:MRHH}. Since the oracle ME results are identical under the two weighting matrices, we report them only once.

As the degree of misspecification ($\delta$) increases, we observe several consistent patterns across DGPs; 1) the SD of GMM increases, 2) the Wald CI of GMM with conventional SE undercovers, while the misspecification-robust SE restores coverage by correctly capturing the SD of GMM, and 3) the oracle ME dominates in SD and CI length across all designs, while retaining coverage (coverage 94-96\% across all $\{ n, \delta \}$), consistent with Corollaries \ref{cor:dist_ME} and \ref{cor:optimal_matrix}. Note that the SD of standard GMM grows substantially with $\delta$, whereas the SD of oracle ME grows only modestly (e.g., SD of GMM $0.43 \rightarrow 1.38$, SD of ME $0.40 \rightarrow 0.79$, when $\delta$ changes from 0 to 2, $n=1000$); the oracle ME estimator takes into account increases in variance through optimal Jacobian recentering rather than letting the asymptotic variance inflate. The  efficiency gains of the oracle ME over GMM ($1-V_{ME}/V_{GMM}$, $n=1000$) are $25\%$, $56\%$, $67\%$ for $W=I$, and $14\%$, $50\%$, $63\%$ for $W=\widehat\Sigma_{11,2}^{-1}$, when $\delta=0.5, 1, 2$, respectively. We also observe that the efficient weighting choice $W=\widehat\Sigma_{11,2}^{-1}$ reduces the SD of GMM.

Furthermore, the DR percentile CI has coverage close to the nominal level in all designs, while the HH bootstrap exhibits coverage distortions under misspecification. Under $W= \widehat\Sigma_{11,2}^{-1}$, the weight does not depend on the preliminary estimator in the linear model, and the corresponding Wald CI with misspecification-robust SE is close to the nominal level, and comparable to the DR percentile CI. By contrast, under the standard optimal weighting matrix $W= \widehat\Sigma_{11}^{-1}$ (Table \ref{tab:sim_s11} in Supplemental Appendix \ref{sec:app_weak_sim}), the weight depends on the preliminary estimator, and the corresponding Wald CI under-covers (0.90-0.92 for $\delta=1,2$), whereas the DR CI remains close to the nominal level.

The simulation results confirm the efficiency of the oracle ME estimator, and overall suggest that the DR bootstrap provides a valid procedure, and a reasonable alternative to the standard GMM with analytic misspecification-robust standard errors.

\begin{table}[!htbp]
\footnotesize
\setlength{\tabcolsep}{3pt}
\centering
\caption{Monte Carlo results under $W = I$ and $W = \widehat\Sigma_{11,2}^{-1}$}
\label{tab:sim}
\begin{tabular}{l *{12}{c}}
\toprule
 & \multicolumn{3}{c}{$\delta=0$} & \multicolumn{3}{c}{$\delta=0.5$} & \multicolumn{3}{c}{$\delta=1$} & \multicolumn{3}{c}{$\delta=2$} \\
\cmidrule(lr){2-4} \cmidrule(lr){5-7} \cmidrule(lr){8-10} \cmidrule(lr){11-13}
Estimator & SD & Cov & Len & SD & Cov & Len & SD & Cov & Len & SD & Cov & Len \\
\midrule
\multicolumn{13}{l}{\textit{Weight matrix: $W = I$}} \\
\multicolumn{13}{l}{\qquad $n=200$} \\
GMM + Conv.\ SE   & 1.00 & 0.945 & 0.386 & 1.01 & 0.916 & 0.351 & 1.62 & 0.867 & 0.476 & 3.10 & 0.869 & 0.907 \\
GMM + Robust SE   & 1.00 & 0.946 & 0.390 & 1.01 & 0.953 & 0.400 & 1.62 & 0.940 & 0.612 & 3.10 & 0.951 & 1.203 \\
HH Bootstrap      & 1.00 & 0.940 & 0.395 & 0.90 & 0.910 & 0.359 & 1.25 & 0.866 & 0.488 & 2.41 & 0.875 & 0.930 \\
DR Bootstrap      & 1.03 & 0.943 & 0.404 & 1.04 & 0.951 & 0.410 & 1.62 & 0.940 & 0.629 & 3.19 & 0.953 & 1.246 \\
\multicolumn{13}{l}{\qquad $n=500$} \\
GMM + Conv.\ SE   & 0.61 & 0.956 & 0.246 & 0.65 & 0.910 & 0.223 & 1.02 & 0.862 & 0.304 & 1.92 & 0.871 & 0.584 \\
GMM + Robust SE   & 0.61 & 0.956 & 0.247 & 0.65 & 0.945 & 0.256 & 1.02 & 0.940 & 0.392 & 1.92 & 0.952 & 0.764 \\
HH Bootstrap      & 0.62 & 0.953 & 0.249 & 0.56 & 0.909 & 0.226 & 0.77 & 0.859 & 0.307 & 1.49 & 0.870 & 0.590 \\
DR Bootstrap      & 0.63 & 0.955 & 0.249 & 0.65 & 0.943 & 0.257 & 1.00 & 0.938 & 0.397 & 1.96 & 0.952 & 0.775 \\
\multicolumn{13}{l}{\qquad $n=1000$} \\
GMM + Conv.\ SE   & 0.43 & 0.955 & 0.175 & 0.45 & 0.912 & 0.159 & 0.71 & 0.863 & 0.217 & 1.38 & 0.859 & 0.416 \\
GMM + Robust SE   & 0.43 & 0.955 & 0.175 & 0.45 & 0.950 & 0.181 & 0.71 & 0.945 & 0.278 & 1.38 & 0.948 & 0.542 \\
HH Bootstrap      & 0.44 & 0.957 & 0.175 & 0.40 & 0.913 & 0.159 & 0.54 & 0.861 & 0.217 & 1.05 & 0.862 & 0.418 \\
DR Bootstrap      & 0.44 & 0.959 & 0.176 & 0.45 & 0.947 & 0.183 & 0.70 & 0.946 & 0.279 & 1.37 & 0.947 & 0.546 \\
\midrule
\multicolumn{13}{l}{\textit{Weight matrix: $W = \widehat\Sigma_{11,2}^{-1}$}} \\
\multicolumn{13}{l}{\qquad $n=200$} \\
GMM + Conv.\ SE   & 1.00 & 0.943 & 0.384 & 0.94 & 0.940 & 0.347 & 1.48 & 0.882 & 0.460 & 2.95 & 0.855 & 0.841 \\
GMM + Robust SE   & 1.00 & 0.944 & 0.388 & 0.94 & 0.951 & 0.374 & 1.48 & 0.938 & 0.557 & 2.95 & 0.943 & 1.116 \\
HH Bootstrap      & 1.00 & 0.941 & 0.393 & 0.90 & 0.935 & 0.355 & 1.20 & 0.880 & 0.472 & 2.22 & 0.853 & 0.862 \\
DR Bootstrap      & 1.02 & 0.942 & 0.401 & 1.02 & 0.957 & 0.402 & 1.54 & 0.954 & 0.599 & 3.00 & 0.944 & 1.168 \\
\multicolumn{13}{l}{\qquad $n=500$} \\
GMM + Conv.\ SE   & 0.61 & 0.953 & 0.245 & 0.61 & 0.933 & 0.223 & 0.94 & 0.876 & 0.295 & 1.84 & 0.856 & 0.541 \\
GMM + Robust SE   & 0.61 & 0.954 & 0.246 & 0.61 & 0.950 & 0.239 & 0.94 & 0.936 & 0.360 & 1.84 & 0.948 & 0.724 \\
HH Bootstrap      & 0.62 & 0.955 & 0.248 & 0.56 & 0.933 & 0.225 & 0.75 & 0.877 & 0.297 & 1.38 & 0.854 & 0.545 \\
DR Bootstrap      & 0.63 & 0.956 & 0.249 & 0.64 & 0.957 & 0.253 & 0.95 & 0.950 & 0.378 & 1.85 & 0.947 & 0.731 \\
\multicolumn{13}{l}{\qquad $n=1000$} \\
GMM + Conv.\ SE   & 0.43 & 0.953 & 0.174 & 0.42 & 0.932 & 0.158 & 0.66 & 0.883 & 0.210 & 1.31 & 0.855 & 0.385 \\
GMM + Robust SE   & 0.43 & 0.953 & 0.175 & 0.42 & 0.950 & 0.170 & 0.66 & 0.945 & 0.256 & 1.31 & 0.950 & 0.513 \\
HH Bootstrap      & 0.44 & 0.954 & 0.175 & 0.40 & 0.933 & 0.159 & 0.53 & 0.881 & 0.211 & 0.97 & 0.852 & 0.385 \\
DR Bootstrap      & 0.44 & 0.955 & 0.176 & 0.45 & 0.962 & 0.179 & 0.67 & 0.949 & 0.267 & 1.30 & 0.946 & 0.517 \\
\midrule
\multicolumn{13}{l}{\textit{ Oracle ME}} \\
\qquad $n=200$   & 0.93 & 0.944 & 0.367 & 0.87 & 0.949 & 0.349 & 1.08 & 0.945 & 0.422 & 1.78 & 0.942 & 0.695 \\
\qquad $n=500$   & 0.58 & 0.945 & 0.232 & 0.56 & 0.945 & 0.221 & 0.69 & 0.944 & 0.267 & 1.09 & 0.952 & 0.440 \\
\qquad $n=1000$  & 0.40 & 0.958 & 0.164 & 0.39 & 0.948 & 0.156 & 0.47 & 0.950 & 0.189 & 0.79 & 0.951 & 0.311 \\
\bottomrule
\end{tabular}
\begin{flushleft}
\scriptsize
\textit{Notes.} We consider (i) standard GMM estimator with the conventional SE; (ii) GMM with the misspecification-robust SE; (iii) Hall and Horowitz (1996) (HH) percentile bootstrap; and (iv) the Double-Recentered percentile bootstrap  (DR), under $W = I$ and $W = \widehat\Sigma_{11,2}^{-1}$, respectively. We also report the oracle ME estimator with efficient SE, $V_{ME}(W)$, whose results are identical under the two weighting matrices (Corollary \ref{cor:optimal_matrix}). For GMM and Oracle ME, ``SD" is the MC standard deviation of the point estimate across MC draws; for HH and DR bootstrap, ``SD" is the average bootstrap standard deviation across MC draws. ``SD" is reported relative to the correctly-specified benchmark SD (GMM with $\delta = 0, W = I$, $n=200$), so values $<1$ ($>1$) indicate higher (lower) precision than this benchmark. ``Cov" is empirical $95\%$ coverage; ``Len" is median CI length. Based on $2{,}000$ Monte Carlo replications, $B=1{,}000$ bootstrap draws.
\end{flushleft}
\end{table}

\section{Empirical Examples}
\label{sec:empirical}

We illustrate our proposed methods with two main empirical applications: the returns to schooling examples of Card (1995) and Angrist and Krueger (1991). We also consider the dynamic panel data model of the employment equation of Arellano and Bond (1991)  in Supplemental Appendix \ref{sec:empirical_ab}.

\subsection{Card (1995)}\label{sec:card}

We replicate the well-known college-proximity IV analysis in Card (1995), who uses the National Longitudinal Survey of Young Men for 1976. The sample size is $n=3010$. The outcome is log wage and the regressors include education and some demographic/geographic variables. The instruments include two binary indicators for living near a four-year public and private college, respectively. 

Table \ref{tab:emp_card} compares several estimators across four specifications and three weighting choices: $W=I$, $W=(Z'Z)^{-1}$ (TSLS), and $W=\widehat\Sigma_{11,2}^{-1}$ (efficient).\footnote{In both applications we consider here, we also considered the conventional optimal weight $W=\widehat\Sigma_{11}^{-1}$ for each specification. The results are almost identical and thus  we do not report them  in the paper.} The first column (1) corresponds to the baseline specification with experience, experience squared,  an indicator for Black, indicators for residence in an SMSA in 1976, and residence in the South in 1976. Column (3) treats experience and its square as endogenous, adding age and age squared to the instrument set. These specifications correspond to Card (1995,  Table 3). In Columns (2) and (4), we additionally control for more covariates including cubic and quartic experience, and  all pairwise interactions between baseline covariates, motivated by the rich-covariates condition of Blandhol et al. (2025), required for the TSLS estimand to admit a weakly causal interpretation. 

Table \ref{tab:emp_card} shows that GMM point estimates are very similar across specifications and $W$ (0.159-0.176 in four columns with three $W$). The conventional and misspecification-robust SEs are  very similar in this example (within about 2\%). The Hansen $J$ test does not reject in any specifications (p-values $\ge 0.35$); however, misspecification can still exist even when the $J$ test fails to reject. Although the ME efficiency bound is essentially similar across specifications (0.0343-0.0366), it is 9-19\% smaller than the conventional/misspecification-robust SE. The ME-GMM bootstrap SD tracks the bound closely, and the DR percentile CI is comparable to the Wald CI with misspecification-robust SEs across specifications. 

As in Mogstad and Torgovitsky (2024), we find that a Ramsey (1969) RESET test rejects the null of rich-covariates (testing whether $E[Z|X]$ is linear in $X$) in basic specifications (1) and (3), rejects  at 5\% but not at 1\% in Column (2) ($p = 0.041$), and the condition holds in  Column (4) with saturated exogenous controls. Since the point estimates are similar across all specifications, the same caveat as in Mogstad and Torgovitsky (2024) applies: two estimates can be very similar even when one corresponds to a non-negatively weighted causal estimand and the other does not.

\begin{table}[hp]
\footnotesize
\centering
\caption{Returns to schooling: Card (1995)}
\label{tab:emp_card}
\begin{tabular}{l cc cc}
\toprule
 & \multicolumn{2}{c}{ $exp/exp^2$ as exogenous} & \multicolumn{2}{c}{$exp/exp^2$ as endogenous} \\
\cmidrule(lr){2-3} \cmidrule(lr){4-5}
 & (1) & (2) & (3) & (4) \\
\midrule
ME efficiency bound & (0.0366) & (0.0347) & (0.0347) & (0.0343) \\
\midrule
\multicolumn{5}{l}{\textit{Panel A: $W = I$ (identity)}} \\
GMM point                & 0.1763 & 0.1734 & 0.1674 & 0.1663 \\
\quad Conv.\ SE          & (0.0448) & (0.0412) & (0.0430) & (0.0418) \\
\quad Misspec-robust SE  & (0.0445) & (0.0410) & (0.0427) & (0.0416) \\
ME-GMM Boot SD           & (0.0375) & (0.0358) & (0.0372) & (0.0374) \\
DR Boot $95\%$ CI          & [0.096,0.289] & [0.104,0.269] & [0.092,0.291] & [0.099,0.292] \\
\midrule
\multicolumn{5}{l}{\textit{Panel B: $W = (Z'Z)^{-1}$ (TSLS)}} \\
GMM point                & 0.1611 & 0.1620 & 0.1597 & 0.1616 \\
\quad Conv.\ SE          & (0.0405) & (0.0382) & (0.0408) & (0.0404) \\
\quad Misspec-robust SE  & (0.0412) & (0.0386) & (0.0415) & (0.0408) \\
ME-GMM Boot SD           & (0.0382) & (0.0357) & (0.0368) & (0.0370) \\
DR Boot $95\%$ CI          & [0.090,0.259] & [0.091,0.258] & [0.090,0.282] & [0.096,0.273] \\
\midrule
\multicolumn{5}{l}{\textit{Panel C: $W = \widehat\Sigma_{11,2}^{-1}$ (efficient)}} \\
GMM point                & 0.1606 & 0.1619 & 0.1594 & 0.1615 \\
\quad Conv.\ SE          & (0.0404) & (0.0381) & (0.0408) & (0.0404) \\
\quad Misspec-robust SE  & (0.0413) & (0.0387) & (0.0417) & (0.0407) \\
ME-GMM Boot SD           & (0.0377) & (0.0361) & (0.0376) & (0.0361) \\
DR Boot $95\%$ CI          & [0.089,0.265] & [0.091,0.254] & [0.090,0.275] & [0.094,0.281] \\
\midrule
\multicolumn{5}{l}{\textit{Panel D: Diagnostics}} \\
First-stage $F$    & 13.5 & 15.7 & 8.6 & 9.0 \\
 $J$-test $p$-value     & 0.351 & 0.410 & 0.461 & 0.592\\
RESET $p$-value    & 0.000 & 0.041 & 0.000 & - \\
\# excluded IVs   & 2 & 2 & 4 & 4 \\
\# exogenous controls    & 6 & 18 & 4 & 8 \\
\bottomrule
\end{tabular}
\begin{flushleft}
\scriptsize
\textit{Notes.} We report returns to education coefficients with standard errors or bootstrap SD (in parentheses). Column (1) includes the baseline controls with indicators for Black, residence in an SMSA in 1976, residence in the South in 1976, and experience and experience squared. Column (2) additionally controls for cubic and quartic terms of experience, and all interactions between covariates. Columns (3) and (4) treat experience and its square as endogenous, adding age and age squared to the instrument sets. Column (3) includes baseline specifications and Column (4) adds pairwise interactions. Specifications correspond to Card (1995, Table 3). $n=3010$. Bootstrap size: $B=2,000$. The ME efficiency bound is the oracle ME standard error $\widehat V_{ME}(W)^{1/2}/\sqrt{n}$ with $V_{ME}(W)=(G'\Sigma_{11,2}^{-1}G)^{-1}$ (Corollary \ref{cor:optimal_matrix}), which does not depend on $W$. DR bootstrap percentile 95\% CI in Section \ref{sec:MRHH} is reported in brackets. \end{flushleft}
\end{table}

\subsection{Angrist and Krueger (1991)}\label{sec:app_AK}

Now, we consider the quarter-of-birth instruments example of Angrist and Krueger (1991).  We use the US Census 1980 extract of men born 1930-1939 ($n=329,509$). The outcome is the log wage and the endogenous variable is education with exogenous controls (year-of-birth, race, SMSA, marital status, region, state). The instruments are quarter-of-birth (QOB) with the interactions. In Table \ref{tab:emp_ak}, Column (1) uses the baseline QOB dummies $(m=3)$; (2) uses QOB$\times$YOB interactions ($m=30$); (3) adds QOB$\times$State interactions ($m=180$) which may exhibit many and many-weak instrument issues as noted in Bound et al. (1995).  These specifications correspond to Angrist and Krueger (1991, Table V Columns 2 and 4 and Table VII Column 8). For each case, we consider the basic controls and the flexible specification that adds all interactions between controls. 

The results in Table \ref{tab:emp_ak} parallel Table \ref{tab:emp_card} in the Card (1995) example. The misspecification-robust SE only marginally exceeds the conventional SE, and the gap is larger for many instrument specifications.   The ME-GMM bootstrap SD tracks the ME efficiency bound closely, in line with the simulation evidence.  The ME efficiency bound is substantially smaller than the conventional/misspecification-robust SE when $W=I$. When we use TSLS weights ($(Z'Z)^{-1}$) and efficient weights ($\widehat\Sigma_{11,2}^{-1}$), the ME efficiency bounds are close to the conventional SE, but still 3-18\% smaller than the misspecification-robust SE as predicted by our theory. By Remark \ref{rem:gap}, the ratio of the misspecification-robust SE to the ME bound is approximately $\sqrt{(1+ J/(mF))/(1-\rho^2)}$. The moment-Jacobian correlation $\rho^2$ is small in this application, and the gap increases with the number of instruments as $J/(mF)$ increases because F falls (weak instruments).

Adding interaction terms barely changes the GMM point estimates. However, the GMM point estimate under $W=I$ is systematically higher ($0.088$-$0.113$) than under TSLS or efficient weighting ($0.079$-$0.099$). Furthermore, the DR bootstrap CI is close to the Wald CIs with the misspecification-robust SE in most cases. The RESET test rejects the rich-covariates condition of Blandhol et al. (2025)  in all specifications, except (1) for the basic instruments. ME-GMM inference still applies to $\theta_W$, but the causal interpretation of the estimates in Table \ref{tab:emp_ak} may be fragile in that case.

\begin{table}[hp]
\footnotesize
\setlength{\tabcolsep}{3pt}
\centering
\caption{Returns to schooling: Angrist and Krueger (1991)}
\label{tab:emp_ak}
\resizebox{\textwidth}{!}{%
\begin{tabular}{l cc cc cc}
\toprule
 & \multicolumn{2}{c}{(1) QOB (m=3)} & \multicolumn{2}{c}{(2) + QOB $\times$ YOB (m=30)} & \multicolumn{2}{c}{(3) + QOB $\times$ State (m=180)} \\
\cmidrule(lr){2-3} \cmidrule(lr){4-5} \cmidrule(lr){6-7}
 & basic & flexible & basic & flexible & basic & flexible \\
\midrule
ME efficiency bound & (0.0204) & (0.0207) & (0.0164) & (0.0166) & (0.0095) & (0.0096) \\
\midrule
\multicolumn{7}{l}{\textit{Panel A: $W = I$ (identity)}} \\
GMM point                & 0.1127 & 0.1116 & 0.0919 & 0.0904 & 0.0906 & 0.0881 \\
\quad Conv.\ SE          & (0.0250) & (0.0252) & (0.0193) & (0.0194) & (0.0148) & (0.0149) \\
\quad Misspec-robust SE  & (0.0252) & (0.0254) & (0.0210) & (0.0211) & (0.0162) & (0.0163) \\
ME-GMM Boot SD           & (0.0204) & (0.0202) & (0.0151) & (0.0154) & (0.0080) & (0.0080) \\
DR Boot $95\%$ CI      & [0.063,0.166] & [0.062,0.168] & [0.051,0.135] & [0.050,0.134] & [0.058,0.124] & [0.057,0.121] \\
\midrule
\multicolumn{7}{l}{\textit{Panel B: $W = (Z'Z)^{-1}$ (TSLS)}} \\
GMM point                & 0.0990 & 0.0981 & 0.0806 & 0.0792 & 0.0831 & 0.0809 \\
\quad Conv.\ SE          & (0.0207) & (0.0210) & (0.0165) & (0.0166) & (0.0098) & (0.0098) \\
\quad Misspec-robust SE  & (0.0210) & (0.0213) & (0.0178) & (0.0179) & (0.0115) & (0.0114) \\
ME-GMM Boot SD           & (0.0202) & (0.0210) & (0.0151) & (0.0151) & (0.0081) & (0.0079) \\
DR Boot $95\%$ CI      & [0.058,0.142] & [0.058,0.144] & [0.045,0.117] & [0.042,0.114] & [0.061,0.106] & [0.058,0.103] \\
\midrule
\multicolumn{7}{l}{\textit{Panel C: $W = \widehat\Sigma_{11,2}^{-1}$ (efficient)}} \\
GMM point                & 0.0991 & 0.0982 & 0.0821 & 0.0807 & 0.0836 & 0.0814 \\
\quad Conv.\ SE          & (0.0207) & (0.0210) & (0.0165) & (0.0166) & (0.0096) & (0.0096) \\
\quad Misspec-robust SE  & (0.0210) & (0.0213) & (0.0177) & (0.0178) & (0.0116) & (0.0113) \\
ME-GMM Boot SD           & (0.0201) & (0.0205) & (0.0149) & (0.0150) & (0.0081) & (0.0081) \\
DR Boot $95\%$ CI      & [0.060,0.141] & [0.056,0.141] & [0.044,0.117] & [0.044,0.116] & [0.061,0.106] & [0.059,0.103] \\
\midrule
\multicolumn{7}{l}{\textit{Panel D: Diagnostics}} \\
First-stage $F$     & \multicolumn{2}{c}{30.5} & \multicolumn{2}{c}{4.75} & \multicolumn{2}{c}{2.43} \\
Hansen $J$ $p$-value     & 0.310 & 0.348 & 0.821 & 0.838 & 0.774 & 0.808 \\
RESET $p$-value    & 0.831 & 0.870& 0.000 & 0.000 & 0.000 & 0.000 \\
\bottomrule
\end{tabular}}
\begin{flushleft}
\scriptsize
\textit{Notes.} We report returns to education coefficients with standard errors or bootstrap SD (in parentheses). Controls include year-of-birth, race, SMSA, marital status, region, and state dummies. Column (1) uses the baseline QOB dummies $(m=3)$; (2) uses QOB$\times$YOB interactions ($m=30$), (3) adds QOB$\times$State interactions ($m=180$). ``Basic" includes baseline specifications and ``flexible" adds all pairwise dummy interactions. Specifications correspond to Angrist and Krueger (1991, Table V Columns 2 and 4 and Table VII Column 8). $n=329,509$. Bootstrap size: $B=2,000$. The ME efficiency bound is the oracle ME standard error $\widehat V_{ME}(W)^{1/2}/\sqrt{n}$ with $V_{ME}(W)=(G'\Sigma_{11,2}^{-1}G)^{-1}$ (Corollary \ref{cor:optimal_matrix}), which does not depend on $W$ in the linear model. DR bootstrap percentile 95\% CI in Section \ref{sec:MRHH} is reported in brackets.
\end{flushleft}
\end{table}

\section{Conclusion}
\label{sec:conclusion}

This paper develops the theory of efficient GMM estimation under moment misspecification. The key insight is that the influence function of the standard misspecified-GMM estimator depends both on the original moment conditions and their Jacobian. By optimally weighting this augmented system, we propose a misspecification-efficient estimator that achieves a smaller asymptotic variance than standard GMM for the same pseudo-true value. We also establish semiparametric efficiency bounds uniform over a class of weighting matrices, and propose a bootstrap implementation based on double recentering of both moment and Jacobian conditions. We illustrate our theory in two cross-sectional IV applications, and find that the efficiency loss is small to moderate, with the TSLS and the efficient weights. 

Several directions remain for future work. First, extending the theory to dependent data with HAC-type weighting matrices would significantly broaden the applicability of the ME estimator to time-series settings. Second, we only consider the one-step GMM with a fixed $W$. Under misspecification, the pseudo-true value changes with each step of the two-step or iterated GMM, whose asymptotic distributions also differ. Existing papers typically use the ``standard" optimal weighting matrix $W_0$, which we show in this paper to be suboptimal under misspecification; it would be interesting to analyze the behavior of the two-step/iterated GMM based on $W^{*} = \Sigma_{11,2}^{-1}(\theta)$ considered in this paper. Third, the relationship between the ME estimator and generalized empirical likelihood (GEL) estimators under misspecification deserves further  investigation. Under misspecification, these estimators have different pseudo-true values and different asymptotic distributions, so they are generally difficult to compare in a unified framework.  Fourth, extending the framework to accommodate weak identification, where the Jacobian $G(\theta_W)$ is near-singular, and the behavior of the GMM estimators may be nonstandard (e.g., Stock and Wright (2000), Andrews and Cheng (2012), Dovonon and Renault (2013)), would be important given the connection to the literature (e.g., Lee and Liao (2018) and Kleibergen (2005)) that already exploits the Jacobian moment conditions. Finally, developing practical guidance on the choice of the weighting matrix class $\mathcal{W}$ would provide applied researchers with concrete recommendations for implementation.

\appendix

\section*{Appendix: Proofs}

\noindent
\textbf{Proof of Theorem \ref{thm:dist_AM}}

Let $G(X_{i}, \theta)=\partial g(X_i, \theta)/\partial\theta^{\prime}$, $G_{n}(\theta)=n^{-1}\sum_{i=1}^{n}G(X_{i}, \theta)$, $G(\theta)=E[G_{n} (\theta)]$. Also, let $ F(X_i, \theta) = \frac{\partial}{\partial \theta^{\prime}} vec (G(X_{i}, \theta)^{\prime}), F_n (\theta)= n^{-1}\sum_{i=1}^{n}F(X_{i}, \theta)$ and $F (\theta)= E[F_n(\theta)]$.  Using the first-order Taylor expansion from the FOC of the GMM estimator, 
\[
G_n(\widehat{\theta}_{GMM} (W))^{\prime} Wg_n(\widehat{\theta}_{GMM} (W)) = 0,\] 
we have 
\begin{equation*}
0 = G_n(\theta_{W} )^{\prime} W g_n(\theta_{W} ) + H_{n} (\widetilde{\theta}) (\widehat{\theta}_{GMM} (W) - \theta_{W} )
\end{equation*}
where
\begin{equation*}
 H_{n}(\theta) = G_n(\theta)^{\prime} W G_n(\theta) +  (g_n(\theta)^{\prime} W \otimes I_{p}) F_n(\theta)
\end{equation*}
and $\widetilde{\theta}$ lies between $\widehat{\theta}_{GMM} (W)$ and $\theta_{W}$. By rearranging,  we have
\begin{eqnarray*}
\widehat{\theta}_{GMM} (W) - \theta_{W}  &=&- H_{n} (\widetilde{\theta})^{-1}G_n(\theta_{W} )^{\prime} W g_n(\theta_{W} ) \\
& =& -H_{n} (\widetilde{\theta})^{-1}  [ G_n(\theta_{W} )^{\prime} W (g_n(\theta_{W} ) - g(\theta_{W} ))  +  (G_n(\theta_{W} ) -  G(\theta_{W} ))^{\prime} W g(\theta_{W} )]\\
& =& -H_{n} (\widetilde{\theta})^{-1}  [ G_n(\theta_{W} )^{\prime} W (g_n(\theta_{W} ) - g(\theta_{W} ))  +  (g(\theta_{W} )^{\prime} W \otimes I_p) vec ((G_n(\theta_{W} ) -  G(\theta_{W} ))^{\prime})] 
\end{eqnarray*}
where we use the FOC of the population GMM,  $G(\theta_{W} )^{\prime} W g(\theta_{W} ) = 0$, in the second equation, and use the properties of the Kronecker product, $vec (AB) = (B^{\prime} \otimes I) vec(A)$.

Under Assumption \ref{assump:AM}, which is the standard regularity conditions for the GMM estimators, we have $\widehat{\theta}_{GMM} (W) \overset{p}{\rightarrow}\theta_{W} , G_n(\theta_{W} ) \overset{p}{\rightarrow} G(\theta_{W} ), H_{n} (\widetilde{\theta}) \overset{p}{\rightarrow} H(\theta_{W} )$, where $H (\theta) = G(\theta)^{\prime} W G(\theta) +  (g(\theta)^{\prime} W \otimes I_{p}) F(\theta) = A(W, \theta) \Gamma (\theta)$. Thus, 
\begin{eqnarray}
\sqrt{n} (\widehat{\theta}_{GMM} (W) - \theta_{W} ) &=& - H(\theta_{W} )^{-1} \frac{1}{\sqrt{n}} \sum_{i=1}^{n} \{  G(\theta_{W} )^{\prime} W [g(X_i, \theta_{W} ) - E g(X_i, \theta_{W} )]  \nonumber \\
& & + (g(\theta_{W} )^{\prime} W \otimes I_p) vec[G(X_i, \theta_{W} )^{\prime} - EG(X_i, \theta_{W} )^{\prime}] \}  + o_p(1) \nonumber \\
&=& - (A (W, \theta_{W} ) \Gamma (\theta_{W} ))^{-1} A (W, \theta_{W} )\frac{1}{\sqrt{n}} \sum_{i=1}^{n} (\psi(X_i, \theta_{W} ) - E[\psi(X_i, \theta_{W} )]) + o_p(1).\nonumber
\end{eqnarray}
Then, the  asymptotic distribution of the GMM estimators in \eqref{asymptotic_dist}  follows by CLT and CMT under Assumption \ref{assump:AM}.

We next focus on the $\widetilde{\theta} (\Lambda)$. A consistency result for  $\widetilde{\theta} (\Lambda)  \overset{p}{\rightarrow}\theta_{W}$ holds under Assumption \ref{assump:AM}. Then, by the Taylor expansion of \eqref{eqn_mestimation_sample} around $\theta_{W}$, we have

\[
0 = \Lambda \widetilde{\psi}_n(\theta_{W}) + \Lambda \Gamma_{n}(\overline{\theta}) (\widetilde{\theta} (\Lambda) - \theta_{W})
\]
where 
 \[
 \Gamma_{n}(\theta) =  \frac{1}{n} \sum_{i=1}^{n} \frac{\partial}{\partial \theta^{\prime}} \psi(X_i, \theta) =  \begin{bmatrix}
G_n(\theta) \\
F_n(\theta)
\end{bmatrix} \\
 \]
 and $\overline{\theta}$ lies between $\widetilde{\theta} (\Lambda)$ and $\theta_{W}$. Then, under Assumption \ref{assump:AM}, 
  \[
\widetilde{\theta} (\Lambda) - \theta_{W} = - (\Lambda \Gamma (\theta_{W}))^{-1} \Lambda  [ \psi_n(\theta_{W}) - E  [\psi (X_i, \theta_{W})] ] + o_p(n^{-1/2})
 \]
where $\psi_n (\theta) = \frac{1}{n} \sum_{i=1}^{n} \psi(X_i, \theta)$, and the asymptotic distribution in Theorem \ref{thm:dist_AM} holds by CLT. \qed 
\vspace{0.3cm}

\noindent
\textbf{Proof of Corollary \ref{cor:dist_ME}}

With the choice of $\Lambda = \Lambda^{*}_{W}$ in \eqref{optimal_weighting}, the asymptotic distribution of the ME estimator in \eqref{optimal_variance} follows by Theorem  \ref{thm:dist_AM} as the asymptotic variance of M-estimator reduces to 
\begin{eqnarray*}
V_{M} (\Lambda^{*}_{W}, \theta_W) &=& (\Gamma(\theta_W)^{\prime} \Sigma (\theta_W)^{-1} \Gamma (\theta_W))^{-1} \Gamma(\theta_W)^{\prime} \Sigma (\theta_W)^{-1} \Sigma (\theta_{W}) \Sigma (\theta_W)^{-1}   \Gamma(\theta_W) \\
& & \quad ( (\Gamma(\theta_W)^{\prime} \Sigma (\theta_W)^{-1} \Gamma (\theta_W))^{\prime})^{-1}\\
&=& (\Gamma(\theta_W)^{\prime} \Sigma (\theta_W)^{-1}\Gamma(\theta_W))^{-1}.
\end{eqnarray*}
We now show that $V_{M} (\Lambda, \theta_W) -(\Gamma(\theta_W)^{\prime} \Sigma (\theta_W)^{-1}\Gamma(\theta_W))^{-1}$ is positive semidefinite. We first note that the following projection matrix is positive semidefinite, 
\[
I - \Sigma(\theta_W)^{-1/2} \Gamma(\theta_W) (\Gamma(\theta_W)^{\prime}\Sigma(\theta_W)^{-1}  \Gamma(\theta_W))^{-1}\Gamma(\theta_W)^{\prime} \Sigma(\theta_W)^{-1/2} \geq 0. 
\]
We also note that if a $L \times L$ square matrix $P$ is positive semidefinite, then  $HPH^{\prime}$ is also positive semidefinite for any $K \times L$ matrix $H$. So, with $H = (\Lambda  \Gamma(\theta_W))^{-1} \Lambda  \Sigma(\theta_W)^{1/2}$, we have
\begin{eqnarray*}
&(\Lambda  \Gamma(\theta_W))^{-1} \Lambda  \Sigma(\theta_W)^{1/2} \big( I - \Sigma(\theta_W)^{-1/2} \Gamma(\theta_W) (\Gamma(\theta_W)^{\prime}\Sigma(\theta_W)^{-1} \Gamma(\theta_W))^{-1}\Gamma(\theta_W)^{\prime} \Sigma(\theta_W)^{-1/2} \big)  \\
&  \Sigma(\theta_W)^{1/2} \Lambda^{\prime} ( \Gamma(\theta_W)^{\prime} \Lambda^{\prime})^{-1}  \geq 0.
\end{eqnarray*}
This implies 
\[V_{M} (\Lambda, \theta_W) - V_{M} (\Lambda^{*}_{W}, \theta_W) = (\Lambda \Gamma (\theta_W))^{-1} \Lambda \Sigma (\theta_{W} ) \Lambda ^{\prime} ( (\Lambda \Gamma (\theta_W))^{\prime})^{-1} - (\Gamma(\theta_W)^{\prime} \Sigma (\theta_W)^{-1}\Gamma(\theta_W))^{-1}\geq 0.
\]

For the last part of Corollary \ref{cor:dist_ME},  we have
\begin{eqnarray*}
\Sigma (\theta_W)^{-1} &=&\begin{bmatrix} \Sigma_{11, 2}^{-1}& - \Sigma_{11}^{-1} \Sigma_{12} \Sigma_{22,1}^{-1}  \\ - \Sigma_{22, 1}^{-1} \Sigma_{21} \Sigma_{11}^{-1} & \Sigma_{22, 1}^{-1}\end{bmatrix}
\end{eqnarray*}
by the inverting partitioned matrices, where $\Sigma_{11, 2} = (\Sigma_{11} - \Sigma_{12}\Sigma_{22}^{-1} \Sigma_{21})$ and $\Sigma_{22,1} = (\Sigma_{22} - \Sigma_{21}\Sigma_{11}^{-1} \Sigma_{12})$, and $
\Sigma (\theta_W) = \begin{bmatrix} \underbrace{\Sigma_{11}}_{m \times m} & \underbrace{\Sigma_{12}}_{m \times mp} \\ \underbrace{\Sigma_{21}}_{mp \times m} & \underbrace{\Sigma_{22}}_{mp \times mp} \end{bmatrix}.$

Then, 
\begin{equation*}
\Lambda^{*}_{W}  = \Gamma(\theta_W)^{\prime} \Sigma (\theta_W)^{-1}  = [G(\theta_W)^{\prime} \Sigma_{11,2}^{-1} - F(\theta_W)^{\prime} \Sigma_{22, 1}^{-1} \Sigma_{21} \Sigma_{11}^{-1}  \quad - G(\theta_W)^{\prime} \Sigma_{11}^{-1} \Sigma_{12} \Sigma_{22,1}^{-1} + F(\theta_W)^{\prime} \Sigma_{22, 1}^{-1}], 
\end{equation*}
and 
\begin{eqnarray}
 \Gamma(\theta_W)^{\prime} \Sigma (\theta_W)^{-1}   \Gamma(\theta_W) &=& G(\theta_W)^{\prime} \Sigma_{11, 2}^{-1}  G(\theta_W) - F(\theta_W)^{\prime} \Sigma_{22, 1}^{-1} \Sigma_{21} \Sigma_{11}^{-1}G(\theta_W)\nonumber  \\ 
 &&- G(\theta_W)^{\prime} \Sigma_{11}^{-1} \Sigma_{12} \Sigma_{22,1}^{-1} F(\theta_W) + F(\theta_W)^{\prime}\Sigma_{22, 1}^{-1} F(\theta_W)\label{eqn:cor11}
\end{eqnarray}

By the following Woodbury matrix identity for nonsingular matrices $A, C$, and $A-BCB'$, 
\[
(A-BCB')^{-1} = A^{-1}  - A^{-1} B (B'A^{-1}B - C^{-1})^{-1} B'A^{-1},
\]
we have
\[
\Sigma_{11, 2}^{-1} = \Sigma_{11}^{-1} - \Sigma_{11}^{-1}\Sigma_{12} (\Sigma_{21} \Sigma_{11}^{-1} \Sigma_{12} -  \Sigma_{22})^{-1}  \Sigma_{21} \Sigma_{11}^{-1} = \Sigma_{11}^{-1} + \Sigma_{11}^{-1}\Sigma_{12} \Sigma_{22, 1}^{-1}   \Sigma_{21} \Sigma_{11}^{-1}  
\]
where the last equality uses the definition of $\Sigma_{22, 1}$. Thus, \eqref{eqn:cor11} can be written as
\begin{eqnarray*}
 G(\theta_W)^{\prime} \Sigma_{11}^{-1} G(\theta_W) + (F(\theta_W) - \Sigma_{21} \Sigma_{11}^{-1} G(\theta_W))^{\prime} \Sigma_{22, 1}^{-1}  (F(\theta_W) - \Sigma_{21} \Sigma_{11}^{-1} G(\theta_W)).
\end{eqnarray*}
This completes the proof. \qed
\vspace{0.3cm}

\noindent
\textbf{Proof of Corollary \ref{cor:optimal_matrix}}

When the model is linear, so that the Jacobian $G(\theta) = G$ does not depend on $\theta$, $F(\theta) = 0$,
\begin{equation*}
\Lambda^{*}_{W}  =  [G^{\prime} \Sigma_{11,2}^{-1}  \quad - G^{\prime} \Sigma_{11, 2}^{-1} \Sigma_{12} \Sigma_{22}^{-1} ], 
\end{equation*}
by Corollary \ref{cor:dist_ME} and $\Sigma_{11, 2}^{-1} \Sigma_{12} \Sigma_{22}^{-1} =  \Sigma_{11}^{-1} \Sigma_{12} \Sigma_{22,1}^{-1}$. Thus the equation \eqref{eqn:cor11} becomes
\begin{eqnarray*}
 \Gamma(\theta_W)^{\prime} \Sigma (\theta_W)^{-1}   \Gamma(\theta_W) &=& G^{\prime} \Sigma_{11, 2}^{-1}  G.
\end{eqnarray*}

Next, we show that $\Sigma_{11,2}(\theta_W)$ doesn't depend on $\theta_W$. 
Fix an arbitrary $\theta_0\in\mathbb{R}^p$ and let $\delta = \theta - \theta_0$.  Since the model is linear, $g(X_i,\theta) = g(X_i,\theta_0) + G(X_i)\delta$, we have
\[
\psi(X_i, \theta)  = \psi_i(\theta) 
= \begin{pmatrix} g(X_i,\theta_0) + G(X_i)\delta \\ vec(G(X_i)') \end{pmatrix}
= \begin{pmatrix} I_m & C \\ 0 & I_{mp} \end{pmatrix} \psi_i(\theta_0)
\]
where $C = (\delta' \otimes I_m)\,K_{p,m}$ and $K_{p,m}$ is the $mp\times mp$ commutation matrix that satisfies $vec(G(X_i)) = K_{p,m}\,vec(G(X_i)')$ (Abadir and Magnus (2005)). $L(\theta) = \bigl(\begin{smallmatrix} I_m & C \\ 0 & I_{mp} \end{smallmatrix}\bigr)$ is block upper-triangular (with identity diagonal blocks) and depends on $\theta$ only through $C$.

Since $\psi_i(\theta) = L(\theta) \psi_i(\theta_0)$, we have $\Sigma(\theta) = L(\theta) \Sigma(\theta_0) L(\theta)'$, and expanding this gives 
\begin{align*}
\Sigma_{11}(\theta) &= \Sigma_{11}(\theta_0) + C\Sigma_{21}(\theta_0) + \Sigma_{12}(\theta_0)C' + C\Sigma_{22}C',
\end{align*}
\begin{align*}
\Sigma_{12}(\theta) &= \Sigma_{12}(\theta_0) + C\Sigma_{22}, \\
\Sigma_{22}(\theta) &= \Sigma_{22}, \\
\Sigma_{12}(\theta)\Sigma_{22}^{-1}\Sigma_{21}(\theta) &= \Sigma_{12}(\theta_0)\Sigma_{22}^{-1}\Sigma_{21}(\theta_0) + \Sigma_{12}(\theta_0)C' + C\Sigma_{21}(\theta_0) + C\Sigma_{22}C'.
\end{align*}
Thus, 
\[
\Sigma_{11, 2}(\theta) = \Sigma_{11}(\theta_0) - \Sigma_{12}(\theta_0)\Sigma_{22}^{-1}\Sigma_{21}(\theta_0),
\]
and all $\theta$-dependent terms cancel out. Thus, $\Sigma_{11, 2}(\theta_W)$ does not depend on $W$. This completes the proof. \qed 

\vspace{0.3cm}
\noindent
\textbf{Proof of Theorem \ref{thm:bootME}}

It is straightforward to show that the ME-GMM estimator is asymptotically equivalent to the augmented M-estimator $\widetilde{\theta} (\Lambda)$ with $\Lambda = \Gamma(\theta_W)^{\prime} \Delta $. Then, the first result  in Theorem \ref{thm:bootME} follows from Theorem \ref{thm:dist_AM}, since  Assumption  \ref{assump:AM}.5 holds by the assumptions of Theorem  \ref{thm:bootME}. 

The bootstrap ME-GMM estimator $\widehat{\theta}_{ME}^{*}$  in \eqref{def:ME-GMM_bootstrap} satisfies the first-order condition,
\[
\Gamma_n^{*} (\widehat{\theta}_{ME}^{*})^{\prime} \Delta^* \widetilde{\psi}_n^* (\widehat{\theta}_{ME}^{*}) =  0
\]
where 
\begin{equation*}
\widetilde{\psi}_n^* (\theta) = \frac{1}{n} \sum_{i=1}^{n} \widetilde{\psi}^{*}(X_i^{*}, \theta)  = \begin{pmatrix}
g_n^*( \theta) - g_n(\widehat{\theta}_W)\\
 vec(G_n^{*}(\theta)^{\prime}) - vec (G_n(\widehat{\theta}_W)^{\prime})
\end{pmatrix}, 
\end{equation*}

and $\Gamma_n^* (\theta) = \frac{1}{n} \sum_{i=1}^{n} \frac{\partial \widetilde{\psi}^{*}(X_i^{*}, \theta)}{\partial \theta'} = \frac{1}{n} \sum_{i=1}^{n} \frac{\partial \psi(X_i^*, \theta) }{\partial \theta^{\prime}}$. Note that 

\begin{eqnarray*}
\Delta^{*} &=& \Biggl(E^{*} [ \psi(X_i^{*}, \widehat{\theta}_W) \psi(X_i^{*}, \widehat{\theta}_W)^{\prime}]  -E^{*} [ \psi(X_i^{*}, \widehat{\theta}_W)] E^{*}  [ \psi(X_i^{*}, \widehat{\theta}_W)]^{\prime}\Biggr)^{-1}\\
&=&\Biggl(\frac{1}{n} \sum_{i=1}^{n} \psi(X_i,\widehat{\theta}_W) \psi(X_i, \widehat{\theta}_W)^{\prime}  - (\frac{1}{n} \sum_{i=1}^{n} \psi(X_i, \widehat{\theta}_W))(\frac{1}{n} \sum_{i=1}^{n} \psi(X_i, \widehat{\theta}_W))^{\prime}   \Biggr)^{-1}
\end{eqnarray*}
depends only on the original sample quantities $\widehat{\theta}_W =  \widehat{\theta}_{GMM} (W)$.

Expanding $\widetilde{\psi}_n^* (\widehat{\theta}_{ME}^{*}) $ around  $\widehat{\theta}_W$, multiplying through by $\sqrt{n}$, and solving gives
\begin{equation}
\sqrt{n} (\widehat{\theta}_{ME}^{*} - \widehat{\theta}_W )=  - (\Gamma_n^{*} (\widehat{\theta}_{ME}^{*})^{\prime} \Delta^*  \Gamma_n^{*} (\tilde{\theta}^*))^{-1}  
\Gamma_n^{*} (\widehat{\theta}_{ME}^{*})^{\prime} \Delta^* \sqrt{n} \widetilde{\psi}_n^* (\widehat{\theta}_{W}) \label{eq:proof_taylor}
\end{equation}
for some $\tilde{\theta}^*$ between $\widehat{\theta}_{ME}^{*}$ and $\widehat{\theta}_W $. Under Assumption \ref{assump:boot}, $\tilde{\theta}^* \overset{p^{*}}{\rightarrow} \theta_W$, and thus $\Gamma_n^{*} (\widehat{\theta}_{ME}^{*}) \overset{p^{*}}{\rightarrow} \Gamma(\theta_W)$ by the bootstrap WLLN (Assumption \ref{assump:boot}.2), and WLLN by Assumption \ref{assump:AM}. Furthermore, we have $\Delta^* = \Sigma_n(\widehat{\theta}_W)^{-1}  \overset{p^{*}}{\rightarrow} \Sigma(\theta_W)^{-1}$. So that, 
\[
(\Gamma_n^{*} (\widehat{\theta}_{ME}^{*})^{\prime} \Delta^*  \Gamma_n^{*} (\tilde{\theta}^*))^{-1}  
\Gamma_n^{*} (\widehat{\theta}_{ME}^{*})^{\prime} \Delta^* \overset{p^{*}} {\rightarrow} (\Gamma(\theta_W)^{\prime} \Sigma(\theta_W)^{-1}\Gamma(\theta_W))^{-1} \Gamma(\theta_W)^{\prime} \Sigma(\theta_W)^{-1}
\]

Next, we consider $\sqrt{n} \widetilde{\psi}_n^* (\widehat{\theta}_{W}) $ in \eqref{eq:proof_taylor}, 
\begin{align}
\sqrt{n} \widetilde{\psi}_n^* (\widehat{\theta}_{W})&= \sqrt{n} \begin{pmatrix}
g_n^*( \widehat{\theta}_{W}) - g_n(\widehat{\theta}_W)\\
 vec(G_n^{*}(\widehat{\theta}_{W})^{\prime}) - vec (G_n(\widehat{\theta}_W)^{\prime})
\end{pmatrix}. 
\end{align}
Note that each of the two terms has conditional mean zero given $\{X_i\}_{i=1}^n$, since $E^*[g_n^*(\theta)] = g_n(\theta)$, $E^*[G_n^*(\theta)] =G_n(\theta)$. Then,  the triangular array CLT under Assumption \ref{assump:boot} gives $\sqrt{n} \widetilde{\psi}_n^* (\widehat{\theta}_{W}) \overset{d^{*}}{\to} N(0, \Sigma(\theta_W))$. The conclusion then follows by the Slutsky's theorem, 
\begin{eqnarray*}
\sqrt{n} (\widehat{\theta}_{ME}^{*}(W) - \widehat{\theta}_W )& \overset{d^{*}}{\to}& (\Gamma(\theta_W)^{\prime} \Sigma(\theta_W)^{-1}\Gamma(\theta_W))^{-1} \Gamma(\theta_W)^{\prime} \Sigma(\theta_W)^{-1} N(0, \Sigma(\theta_W)) \\
&=&N(0, (\Gamma(\theta_W)^{\prime} \Sigma(\theta_W)^{-1}\Gamma(\theta_W))^{-1} ).\qed 
\end{eqnarray*}

\noindent
\textbf{Proof of Corollary \ref{cor:bootDRHH}}

The DR estimator $\widehat{\theta}_{DR}^{*}$  in \eqref{eq:drhh} satisfies the following; 
\[
A_n (W, \widehat{\theta}_W) \widetilde{\psi}_n^* (\widehat{\theta}_{DR}^{*}) =  0
\]
with $A_n (W, \widehat{\theta}_W)$ instead of  $\Gamma_n^{*} (\widehat{\theta}_{ME}^{*})^{\prime} \Delta^*$ as in the  $\widehat{\theta}_{ME}^{*}$ case. Then, it can be similarly shown as in the proof of Theorem \ref{thm:bootME},
that the DR estimator satisfies;
\begin{eqnarray*}
\sqrt{n} (\widehat{\theta}_{DR}^{*} - \widehat{\theta}_W )& = & - (A (W, \theta_{W} ) \Gamma (\theta_{W} ))^{-1} A (W, \theta_{W} )\sqrt{n} \widetilde{\psi}_n^* (\widehat{\theta}_{W})  + o_{p^{*}}(1) \\
& \overset{d^{*}}{\to} &N(0, V_{GMM}(W,\theta_W))
\end{eqnarray*}
which is the same asymptotic distribution as the GMM estimator,  with asymptotic variance $V_{GMM}(W,\theta)$ in \eqref{asymptotic_dist}.\\ 
\qed

\noindent
\textbf{Proof of Theorem \ref{thm:semiparametric_efficiency}}

By construction,  for any finite $n$, the following holds,
\begin{eqnarray}
\sup_{(W, (F,\theta))\in \mathcal{G}^{\Gamma}} E_{F}\left\{l\left[\sqrt{n}(T_{n}-\theta_{1})\right]\right\} &=& \sup_{W\in \mathcal{W}} \sup_{(F,\theta) \in \Gamma^W} E_{F}\left\{l\left[\sqrt{n}(T_{n}-\theta_{1})\right]\right\} \nonumber \\
&&\ge \sup_{(F,\theta) \in \Gamma^{\widetilde{W}}} E_{F}\left\{l\left[\sqrt{n}(T_{n}-\theta_{1})\right]\right\}  \label{proof:risk1}
\end{eqnarray}
for any arbitrary element $\widetilde{W} \in \mathcal{W}$. Since the inequality \eqref{proof:risk1} holds for all $n$, we have that
\begin{equation}
 \liminf_{n\rightarrow\infty}  \sup_{(W, (F,\theta))\in \mathcal{G}^{\Gamma}} E_{F}\left\{l\left[\sqrt{n}(T_{n}-\theta_{1})\right]\right\} \ge \liminf_{n\rightarrow\infty}  \sup_{(F,\theta) \in \Gamma^{\widetilde{W}}} E_{F}\left\{l\left[\sqrt{n}(T_{n}-\theta_{1})\right]\right\} \label{proof:risk2}
\end{equation}

Now, we fix $W = \widetilde{W} \in \mathcal{W}$. Then, the remaining proof similarly follows the argument of Chamberlain (1987). Conditions (i)-(iv) of $\mathcal{F}^W$ are exactly condition (C0) of Chamberlain (1987, Section 2) for the moment functions $\widetilde{\psi}(X, \theta)$. Then, we can construct the multinomial distribution $G_0$ which is arbitrarily close to  $(F_0, \theta_W)$ that satisfies the conditions in $\mathcal{F}^W$, and $(G_0, \theta_W) \in \Gamma^W$. To be more specific, by Lemma 3 of Chamberlain (1987), there exists a probability measure $G_0$ whose support is a finite set in the support of $F_0$ with probabilities $p = (p_1, \cdots, p_r)$ on the support points, which satisfies the moment condition $\sum_{j=1}^{r} p_j \psi(X_j, \theta_W) = E_{F_0}[\psi(X_i, \theta_W)]$. Then, we apply the local asymptotic minimax bounds for the multinomial distribution (Lemma 2 in Chamberlain (1987)). 

By Theorem 2 of Chamberlain (1987) for the moment function $\widetilde{\psi}(X, \theta)$ and the pairs $(F, \theta)$ that satisfy the regularity conditions the space of probability distributions $\mathcal{F}^{\widetilde{W}}$,  \eqref{proof:risk2} is bounded below by the following,
\[
 \int_{-\infty}^{\infty}l(\sigma_{\widetilde{W}}u)d\Phi(u),
\]
which is the expected loss of $\mathcal{N} (0, \sigma_{\widetilde{W}}^2)$, where $\sigma_{\widetilde{W}}^2$ is the (1,1) element of
\begin{eqnarray*}
&& \left( \left[ E_{F_{0}}\frac{\partial\widetilde{\psi}(X,\theta_{\widetilde{W}})}{\partial\theta^{\prime}} \right]^{\prime} \left[ E_{F_{0}}\widetilde{\psi}(X,\theta_{\widetilde{W}})\widetilde{\psi}^{\prime}(X,\theta_{\widetilde{W}}) \right]^{-1} \left[ E_{F_{0}}\frac{\partial\widetilde{\psi}(X,\theta_{\widetilde{W}})}{\partial\theta^{\prime}} \right] \right)^{-1} \\
&=&(\Gamma(\theta_{\widetilde{W}})^{\prime} \Sigma (\theta_{\widetilde{W}})^{-1}\Gamma(\theta_{\widetilde{W}}))^{-1}.
\end{eqnarray*}

Since this bound holds for any choice of $\widetilde{W} \in \mathcal{W}$, we have 
\[
\liminf_{n\rightarrow\infty} \sup_{(W, (F,\theta))\in \mathcal{G}^{\Gamma}} E_{F}\left\{l\left[\sqrt{n}(T_{n}-\theta_{1})\right]\right\} \ge \sup_W \int_{-\infty}^{\infty}l(\sigma_{W}u)d\Phi(u) =  \int_{-\infty}^{\infty}l(\sup_W  \sigma_{W}u)d\Phi(u)
\]
where the last equality holds by the symmetry, monotonicity of the $l(\cdot)$, and $0<\sigma_W \le \sup_W \sigma_W <\infty$, and by the Monotone Convergence Theorem. \qed \\

\section*{References}

\begin{description}
\setlength{\itemsep}{3pt}\setlength{\parsep}{0pt}

\item[] Abadir, K. M., and Magnus, J. R. (2005). Matrix Algebra. \emph{Cambridge University Press}.

\item[] Andrews, D. W. K., and Cheng, X. (2012). Estimation and inference with weak, semi-strong and strong identification. \emph{Econometrica}, 80, 2153-2211.

\item[] Andrews, D. W. K., and Kwon, S. (2024). Misspecified Moment Inequality Models: Inference and Diagnostics. \emph{The Review of Economic Studies}, 91(1), 45-76.

\item[] Andrews, I. (2019). On the Structure of IV Estimands. \emph{Journal of Econometrics},  211 (1), 294-307.

\item[] Andrews, I., Barahona, N., Gentzkow, M., Rambachan, A., and Shapiro, J. M. (2025). Structural estimation under misspecification: Theory and implications for practice. \emph{Quarterly Journal of Economics}, 140(3):1801-1855.

\item[]Andrews, I., Barnhard, H., and Carlson, J. (2024). True and pseudo-true parameters. \emph{Working Paper}.

\item[]Andrews, I., Chen, J., and Tecchio, O. (2025). The Purpose of an Estimator is What it Does. \emph{Econometric Society 2025 World Congress Volume}.

\item[]  Andrews, I., Gentzkow, M., and Shapiro, J. M. (2017). Measuring the sensitivity 
of parameter estimates to estimation moments. \emph{Quarterly Journal of Economics}, 
132(4), 1553-1592.

\item[] Angrist, J. D., and Imbens, G. W. (1995). Two-stage least squares estimation of average causal effects in models with variable treatment intensity, \emph{Journal of the American Statistical Association}, 90, 431-442.

\item[] Angrist, J. D., and Krueger, A. B. (1991). Does Compulsory School Attendance Affect Schooling and Earnings?, \emph{The Quarterly Journal of Economics}, 106(4), 979-1014.

\item[] Angrist, J. D., and Krueger, A. B. (1995). Split-sample instrumental variables estimates of the return to schooling. \emph{Journal of Business and Economic Statistics}, 13(2), 225-235.

\item[]  Arellano, M., and Bond, S. (1991). Some tests of specification for panel data: Monte Carlo evidence and an application to employment equations.  \emph{Review of Economic Studies}, 58(2), 277-297.

\item[]Armstrong, T. B., and Koles\'{a}r, M. (2021). Sensitivity analysis using approximate moment condition models. \emph{Quantitative Economics}, 12(3), 705-740.

\item[] Armstrong, T. B., Kline, P., and Sun, L. (2025). Adapting to Misspecification. \emph{Econometrica}, 93 (6), 1981-2005.

\item[] Blandhol, C., Bonney, J., Mogstad, M., and Torgovitsky, A. (2025). When is TSLS Actually LATE?. \emph{Working Paper}

\item[]  Bonhomme, S., and Weidner, M. (2022). Minimizing sensitivity to model 
misspecification. \emph{Quantitative Economics}, 13(3), 907-954.

\item[] Bound, J., Jaeger, D. A., and Baker, R. M. (1995). Problems with instrumental variables estimation when the correlation between the instruments and the endogenous explanatory variable is weak. \emph{Journal of the American Statistical Association}, 90 (430), 443-450.

\item[] Brown, B. W., and Newey, W. K. (2002). Generalized method of moments, efficient bootstrapping, and improved inference. \emph{Journal of Business and Economic Statistics}, 20(4), 507-517.

\item[] Card, D. (1995). Using Geographic Variation in College Proximity to Estimate the Return to Schooling, \emph{in Aspects of Labour Market Behaviour: Essays in Honour of John Vanderkamp}, ed. by L. N. Christofides, K. E. Grant, and R. Swidinsky, Toronto: University of Toronto Press, 201-222.

\item[]Chamberlain, G. (1987). Asymptotic efficiency in estimation with conditional moment restrictions. \emph{Journal of Econometrics}, 34(3), 305-334.

\item[]  Christensen, T., and Connault, B. (2023). Counterfactual sensitivity and 
robustness. \emph{Econometrica}, 91(1), 263-298.

\item[]  Conley, T.G., Hansen, C.B., and Rossi, P.E. (2012). Plausibly exogenous, \emph{The Review of Economics and Statistics}, 94(1), 260-272.

\item[] Dovonon, P, and Gon\c{c}alves, S. (2017). Bootstrapping the GMM overidentification test under first-order underidentification. \emph{Journal of Econometrics}, 201, 43-71.

\item[] Dovonon, P., and Renault, E. (2013). Testing for common conditionally heteroskedastic factors. \emph{Econometrica}, 81, 2561-2586.

\item[] Gin\'{e}, E., and Zinn, J. (1990). Bootstrapping general empirical measures, \emph{The Annals of Probability}, 18(2), 851-869.

\item[]Godambe, V. P. (1960). An optimum property of regular maximum likelihood estimation. \emph{Annals of Mathematical Statistics}, 31(4), 1208-1211.

\item[] Gospodinov, N., Kan, R., and Robotti, C. (2014). Misspecification-robust inference in linear asset-pricing models with possibly nonlinear pricing kernels. \emph{Review of Financial Studies}, 27(7), 2139-2170.

\item[]Hahn, J. (1996). A note on bootstrapping generalized method of moments estimators. \emph{Econometric Theory}, 12(1), 187-197.

\item[] H\'{a}jek, J. (1972). Local asymptotic minimax and admissibility in estimation. In Proceedings of the Sixth Berkeley Symposium on Mathematical Statistics and Probability, Vol. 1, \emph{University of California Press}, 175-194.

\item[]Hall, A. R., and Inoue, A. (2003). The large sample behavior of the generalized method of moments estimator in misspecified models. \emph{Journal of Econometrics}, 114(2), 361-394.

\item[]Hall, A. R., and Pelletier, D. (2011). Nonnested testing in models estimated via generalized method of moments. \emph{Econometric Theory}, 27(2), 443-456.

\item[]Hall, P., and Horowitz, J. L. (1996). Bootstrap critical values for tests based on generalized-method-of-moments estimators. \emph{Econometrica}, 64(4), 891-916.

\item[] Hansen, B.E., and Lee, S. (2021). Inference for Iterated GMM Under Misspecification. \emph{Econometrica}, 89, 1419-1447

\item[]Hansen, L. P. (1982). Large sample properties of generalized method of moments estimators. \emph{Econometrica}, 50(4), 1029-1054.

\item[]Hansen, L. P., and Jagannathan, R. (1997). Assessing specification errors in stochastic discount factor models. \emph{Journal of Finance}, 52(2), 557-590.

\item[]Hwang, J., Kang, B., and Lee, S. (2022). A Doubly Corrected Robust Variance Estimator for Linear GMM. \emph{Journal of Econometrics}, 229(2), 276-298.

\item[]Imbens, G. W. (1997). One-step estimators for over-identified generalized method of moments models. \emph{Review of Economic Studies}, 64(3), 359-383.

\item[] Kan, R., and Robotti, C. (2009). Model Comparison Using the Hansen-Jagannathan Distance. \emph{The Review of Financial Studies},  22(9), 3449-3490.

\item[] Kleibergen, F. (2005). Testing Parameters in GMM Without Assuming that They Are Identified. \emph{Econometrica}, 73(4), 1103-1123.

\item[] Kleibergen, F., and Zhan, Z. (2025). Double robust inference for continuous updating GMM, \emph{Quantitative Economics}, 16(1), 295-327.

\item[]Koles\'{a}r, M. (2013). Estimation in an Instrumental Variables Model With Treatment Effect Heterogeneity, \emph{Working Paper}.

\item[] Lee, J.H., and Liao, Z. (2018). On Standard Inference for GMM with Local Identification Failure of Known Forms. \emph{Econometric Theory}, 34, 790-814.

\item[]Lee, S. (2014). Asymptotic refinements of a misspecification-robust bootstrap for generalized method of moments estimators, \emph{Journal of Econometrics}, 178(3), 398-413.

\item[] Marmer, V., and Otsu, T. (2012). Optimal comparison of misspecified moment restriction models under a chosen measure of fit. \emph{Journal of Econometrics}, 170(2), 538-550.

\item[]Mogstad, M., and Torgovitsky, A. (2024). Instrumental Variables with Unobserved Heterogeneity in Treatment Effects. In \emph{Handbook of Labor Economics}, Volume 5, forthcoming.

\item[]Mogstad, M., Torgovitsky, A., and Walters, C. (2021). The Causal Interpretation of Two-Stage Least Squares with Multiple Instrumental Variables. \emph{American Economic Review}, 111(11), 3663-3698.

\item[] Mukhin, Y. (2019). Sensitivity of Regular Estimators. \emph{Working Paper}.

\item[]M\"{u}ller, U. K. (2013). Risk of Bayesian inference in misspecified models, and the sandwich covariance matrix. \emph{Econometrica}, 81(5), 1805-1849.

\item[]Newey, W. K., and McFadden, D. (1994). Large sample estimation and hypothesis testing. \emph{In Handbook of Econometrics}, Vol. 4, 2111-2245.

\item[] Ramsey, J. B. (1969). Tests for specification errors in classical linear least-squares regression analysis. \emph{Journal of the Royal Statistical Society: Series B}, 31(2), 350-371.

\item[]Rivers, D., and Vuong, Q. (2002). Model selection tests for nonlinear dynamic models. \emph{The Econometrics Journal}, 5(1), 1-39.

\item[]S\l oczy\'{n}ski, T. (2024). When Should We (Not) Interpret Linear IV Estimands as LATE?, \emph{Working Paper}.

\item[] Stock, J.H., and Wright, J.H. (2000). GMM with weak identification, \emph{Econometrica}, 68, 1055-1096.

\item[]White, H. (1982). Maximum likelihood estimation of misspecified models. \emph{Econometrica}, 50(1), 1-25.

\item[] Windmeijer, F. (2005). A finite sample correction for the variance of linear efficient two-step GMM estimators, \emph{Journal of Econometrics}, 126(1), 25-51.

\item[] Windmeijer, F., and Santos Silva, J.M.C. (1997). Endogeneity in count data models: An application to demand for health care, \emph{Journal of Applied Econometrics}, 12(3), 281-294. 

\end{description}

%======================================================================
%  SUPPLEMENTAL APPENDIX
%  The settings below reproduce the preamble of the stand-alone supplement:
%  pages S1, S2, ...; appendices A, B, ...; equations and tables numbered
%  within each appendix (e.g., (A.1), Table C.1).
%======================================================================
\clearpage
\setcounter{page}{1}
\renewcommand{\thepage}{S\arabic{page}}
\setcounter{section}{0}
\setcounter{equation}{0}
\setcounter{table}{0}
\setcounter{footnote}{0}
\setlength{\emergencystretch}{3em}
\renewcommand{\thesection}{\Alph{section}}
\numberwithin{equation}{section}
\numberwithin{table}{section}
% Same as \usepackage[section]{placeins} in the stand-alone supplement:
% a \FloatBarrier before each \section keeps each table within its own appendix.
\makeatletter
\expandafter\renewcommand\expandafter\section\expandafter
  {\expandafter\@fb@secFB\section}%
\newcommand\@fb@secFB{\FloatBarrier
  \gdef\@fb@afterHHook{\@fb@topbarrier \gdef\@fb@afterHHook{}}}
\g@addto@macro\@afterheading{\@fb@afterHHook}
\gdef\@fb@afterHHook{}
% Restore the title macros for the supplement title page.
\let\maketitle\SuppSavedMaketitle
\let\@maketitle\SuppSaved@maketitle
\let\title\SuppSavedTitle
\let\author\SuppSavedAuthor
\let\date\SuppSavedDate
\let\thanks\SuppSavedThanks
\let\and\SuppSavedAnd
\makeatother

\title{Supplement to ``Efficient GMM and Weighting Matrix under Misspecification''}
\date{October 15, 2026}

\author{Byunghoon Kang\footnote{Department of Economics,  Lancaster University Management School.\\ Email: b.kang1@lancaster.ac.uk. } }
\maketitle

\begin{abstract}
This supplement contains additional theoretical results, estimators, simulations, and empirical results to accompany the main paper, ``Efficient GMM and Weighting Matrix under Misspecification.'' Appendix \ref{sec:linear_iv} provides additional details on the linear IV example. Appendix \ref{sec:worstcase_variance} considers worst-case variance bounds over misspecification and identification strength. Appendix \ref{sec:splitsample} considers the split-sample ME estimator. Appendix \ref{sec:app_weak_sim} reports additional simulations. Appendix \ref{sec:empirical_ab} contains the dynamic panel data application of Arellano and Bond (1991). Appendix \ref{sec:decomp_ME} provides the decomposition of the ME efficiency gain.
\end{abstract}

\section{Additional Details on the Linear IV Example}
\label{sec:linear_iv}

This appendix provides the closed-form expressions of the ME-GMM estimator, the bootstrap ME-GMM estimator, and the double-recentered (DR) bootstrap estimator in the linear IV example of Sections \ref{subsec:ME_estimator_linear} and \ref{sec:bootstrap} of the main paper.

Based on the FOC from \eqref{def:ME-GMM} in the main paper and using the linearity of the moment condition, the ME-GMM has a closed-form

\begin{multline}
\widehat{\theta}_{ME-GMM}(\Delta) =  (\frac{1}{n}X'Z \Delta_{11} Z'X)^{-1}  \biggl(X'Z \Delta_{11} (\frac{1}{n}Z'Y -  E[g(X_i, \theta_{W})])  \\ 
+ X'Z \Delta_{12} ( \frac{1}{n}vec(X'Z) - E[vec(X_iZ_i')])  \biggr) \label{def:ME_linear_estimator} 
\end{multline}
where $\Delta =  \begin{bmatrix} \Delta_{11} & \Delta_{12} \\  \Delta_{21} & \Delta_{22} \end{bmatrix}$.  When $\Delta_{11} = (Z'Z)^{-1} $ and $\Delta_{12} = 0$, ME-GMM estimator is related to a recentered TSLS estimator, 
\begin{equation}
\widehat{\theta}_{ME-GMM}(\Delta) = (\frac{1}{n}X'Z (Z'Z)^{-1}  Z'X)^{-1} X'Z (Z'Z)^{-1}  (\frac{1}{n}Z'Y -  E[g(X_i, \theta_{W})]).\label{def:HH_linear_estimator} 
\end{equation}
The choice of $\Delta_{11} = (Z'Z)^{-1}, \Delta_{12} = 0$ for ME-GMM  is not efficient under misspecification, although its choice was motivated by the efficiency of GMM under correct specification and homoskedasticity.

With the bootstrap samples $ \{ Y_i^{*}, X_i^{*}, Z_i^{*} \}_{i=1}^{n} $, independently sampled from the original sample $\{ Y_i, X_i, Z_i\}_{i=1}^{n}  $, the Hall and Horowitz (1996) version of the recentered bootstrap TSLS estimator is
\[
\widehat{\theta}_{HH}^{*} = (X^{* \prime}Z^{*} (Z^{*  \prime}Z^{*})^{-1}  Z^{*  \prime}X^{*})^{-1} X^{*  \prime}Z^{*  } (Z^{*  \prime}Z^{*})^{-1}   (Z^{*  \prime}Y^{* } -  Z^{\prime} \widehat{e})
\]
where $\widehat{e} = Y - X \widehat{\theta}_{TSLS}$. This is essentially related to a version of recentered bootstrap estimator of \eqref{def:HH_linear_estimator}.

Our proposed bootstrap ME-GMM estimator in \eqref{def:ME-GMM_bootstrap} of the main paper has the closed form
\begin{multline}
\widehat{\theta}_{ME}^{*}(W) = (X^{* \prime}Z^{*} \widehat{\Sigma}_{11, 2}^{-1} Z^{* \prime}X^{*})^{-1} \biggl( X^{* \prime}Z^{*} \widehat{\Sigma}_{11, 2}^{-1} (Z^{* \prime} Y^{*} -  Z^{\prime} \widehat{e})  \\- X^{* \prime}Z^{*}  \widehat{\Sigma}_{11,2}^{-1} \widehat{\Sigma}_{12} \widehat{\Sigma}_{22}^{-1} ( vec(X^{* \prime}Z^{*}) - vec(X'Z))  \biggr) \label{def:ME_linear_estimator_efficient_bootstrap} 
\end{multline}
where $\widehat{e} = Y - X \widehat{\theta}_{W}$, $\widehat{\Sigma}_{11, 2} = (\widehat{\Sigma}_{11} - \widehat{\Sigma}_{12} \widehat{\Sigma}_{22}^{-1} \widehat{\Sigma}_{21}),  \widehat{\Sigma} (\widehat{\theta}_{W})  = \begin{bmatrix} \widehat{\Sigma}_{11} & \widehat{\Sigma}_{12} \\ \widehat{\Sigma}_{21}&\widehat{\Sigma}_{22}\end{bmatrix}$, $\widehat{\Sigma} (\theta) = \frac{1}{n} \sum_{i=1}^{n} \psi(X_i, \theta) \psi(X_i, \theta)^{\prime} - \frac{1}{n} \sum_{i=1}^{n} \psi(X_i, \theta) \frac{1}{n} \sum_{i=1}^{n} \psi(X_i, \theta)^{\prime}$, and 
$\psi(X_i, \theta) = ((Z_i Y_i-Z_i X_i'\theta)', vec(X_iZ_i^{\prime})')'$. In the linear IV setup, $\widehat{\theta}_{ME}^{*}$ is not just a linear combination of the $\widehat{\theta}_{HH}^{*}$ and the recentered Jacobian moments, but an efficient linear combination with  $\widehat{\Sigma}_{11, 2}$.

Since the sample and population Jacobians do not depend on $\theta$,  $\widehat{\theta}_{DR}^{*}$ with TSLS weights $W  = (Z'Z)^{-1}$  admits a closed-form solution by direct algebra from \eqref{eq:drhh} of the main paper:
\begin{eqnarray}
\widehat{\theta}_{DR}^{*} &=&  \big(X'Z (Z'Z)^{-1} Z^{*\prime}X^*\big)^{-1} \Big(X'Z (Z'Z)^{-1} (Z^{*\prime}Y^* - Z'\hat{e}) \nonumber \\
&&\qquad + \quad (\hat{e}' Z (Z'Z)^{-1} \otimes I_p) \mathrm{vec}(X^{*\prime}Z^* - X'Z) \Big).
 \nonumber
\end{eqnarray}
Thus, the DR bootstrap estimator $\widehat{\theta}_{DR}^{*}$ is approximately equal to $\widehat{\theta}_{HH}^{*} +  \tilde{\Lambda}_2 \cdot \mathrm{vec}(X^{*\prime}Z^* - X'Z)$ with some weights $ \tilde{\Lambda}_2 $. The last term is a correction term for the Jacobian variability. Under correct specification, $g_n(\widehat{\theta}_{W}) = n^{-1}Z'\hat{e} = O_p(n^{-1/2})$, so the last term is a  finite-sample correction.

\section{Worst-Case Variance Bounds over Misspecification and Identification Strength}
\label{sec:worstcase_variance}

In this section, we consider the inference on $\theta_W$ without assuming that the recentering vector $\gamma (W) = (\gamma_1(W)', vec(\gamma_2(W)')')'$ is known. We assume only that the support of $\gamma$ is known and consider a confidence interval for $\theta_W$ by taking a union of valid confidence intervals specific to each $\gamma$. This approach is closely related to Conley, Hansen, and Rossi (2012), where they construct the union of confidence intervals over the nuisance parameter (``violation of the exclusion restrictions") which is related to $\gamma_1$ we consider here. In our framework, the nuisance parameter $\gamma$ not only includes the degree of misspecification $(\gamma_1)$, but also contains the identification
strength $(\gamma_2)$. Interestingly, we can consider the trade-off relationship in the construction of the parameter spaces from the FOC of the overidentified GMM conditions; $\gamma_2'W\gamma_1 = 0$, if $\gamma$ is equal to the true value $\gamma(W)$.

Define the parameter space
\begin{equation}\label{eq:Gamma_space}
\begin{split}
\overline{\Gamma}_W = \bigl\{\gamma = (\gamma_1', vec(\gamma_2')')' \in \mathbb{R}^{m + mp} :\; & 0 \leq ||\gamma_1|| \leq \overline{\gamma}_1, \;\; \underline{\gamma}_2 \leq || \gamma_2 || \leq \overline{\gamma}_2, \\
& \underline{\gamma}_2>0, \;\; \gamma_2'W\gamma_1 = 0 \bigr\}.
\end{split} 
\end{equation}
\noindent
$\overline{\gamma}_1$ denotes the worst-case degree of misspecification and $ ||\gamma_1|| = 0$ corresponds to the correct specification. $\underline{\gamma}_2>0$ rules out weak identification. We  do not consider data-dependent choice of $\overline{\gamma}_1, \underline{\gamma}_2$ in this paper, but practically we can consider  $\overline{\gamma}_1$ as J-statistics (which estimate $|| \gamma_1||^2$ with $W$), and some pre-test identification strength measure for $ \underline{\gamma}_2$ such as effective F-statistics in the linear IV setup. The FOC condition requires $\gamma_1$ to lie in the null space of $\gamma_2'W$, and thus the direction of $\gamma_1$ is restricted. When $W = \Sigma_{11,2}^{-1}$ in the linear model, the recentering part $\gamma_1(W)$ will be eliminated by the FOC $G'\Sigma_{11,2}^{-1} \gamma_1(W) = 0 $ from Corollary \ref{cor:optimal_matrix}. Since $\gamma_1(W)$ is not needed for $\widehat{\theta}_{ME}(W)$, the parameter space $\overline{\Gamma}_W $ only requires restrictions on $\gamma_2$. This reduces not only computational costs, but also the conservativeness of the CI considered here.

For any $\gamma \in \overline{\Gamma}_W $, we can define the ME-estimator in \eqref{def:ME-GMM}  of the main paper as follows;
\begin{equation}\label{def:ME-gamma}
\widehat{\theta}_{ME} (W; \gamma) = \argmin_{\theta}  \biggl(\psi_n(\theta) - \gamma \biggr)^{\prime}  \widehat{\Sigma} (\widehat{\theta}_W)^{-1} \biggl(\psi_n(\theta) - \gamma \biggr) 
\end{equation}
where $\psi_n (\theta) = \frac{1}{n} \sum_{i=1}^{n} \psi(X_i, \theta), \psi(X_i, \theta) = (g(X_i,\theta)', vec(G(X_i,\theta)')')'$, and $\widehat{\Sigma} (\widehat{\theta}_W)$ is a consistent estimator of $\Sigma(\theta_W)$ with the preliminary GMM estimator $\widehat{\theta}_W$. Although we can consider the pseudo-true value $\theta_W(\gamma)$ as  a minimizer of the population objective function of \eqref{def:ME-gamma}, for simplicity, we assume that for all $\gamma \in \overline{\Gamma}_{W}$, there exists a unique value of $\theta = \theta_W(\gamma)$ such that $E[\psi(X_i, \theta_W(\gamma))] = \gamma$. Then, it will follow under the conventional regularity conditions that
\begin{equation}
\sqrt{n} (\widehat{\theta}_{ME} (W; \gamma) - \theta_W(\gamma)) \overset{d}{\longrightarrow} N(0, V_{ME}(W;\gamma))
\end{equation}
where
\begin{multline*}
V_{ME}(W;\gamma) = \big(\Gamma(\theta_W(\gamma))'\Sigma(\theta_W)^{-1}\Gamma(\theta_W(\gamma))\big)^{-1} \Gamma(\theta_W(\gamma))'\Sigma(\theta_W)^{-1} \Sigma_{\gamma} \\
\times \Sigma(\theta_W)^{-1} \Gamma(\theta_W(\gamma)) \big(\Gamma(\theta_W(\gamma))'\Sigma(\theta_W)^{-1}\Gamma(\theta_W(\gamma))\big)^{-1},
\end{multline*}
and $\Sigma_{\gamma} = E [ \psi(X_i, \theta_W(\gamma)) \psi(X_i, \theta_W(\gamma))^{\prime}]  -\gamma \gamma^{\prime}$. Note that at the true values $\gamma = \gamma(W)$, we have $\theta_W(\gamma(W)) = \theta_W$ since $E[\psi(X_i, \theta_W)] = \gamma(W)$ and $\theta_W(\gamma)$ is unique, and thus $\Sigma_{\gamma} = \Sigma(\theta_W)$, $V_{ME}(W;\gamma) =V_{ME}(W) $ when evaluated at the $\gamma(W)$. Theorem \ref{thm:semiparametric_efficiency} can be extended to the worst-case minimax bounds $\sup_{\gamma \in \overline{\Gamma}_W} V_{ME}(W;\gamma)$ uniformly over $\gamma \in \overline \Gamma_{W}$. In the linear model, $V_{ME}(W) = (\gamma_2'\Sigma_{11,2}^{-1}\gamma_2)^{-1}$ and it depends only on $\gamma_2$, and the bounds for $V_{ME}(W)$ can be easily calculated with the $||\gamma_2||^2$ and the Rayleigh quotient for $\Sigma_{11,2}^{-1}$,  which reduces the dimension significantly.

For simplicity, we suppose that $\theta$ is a scalar ($p=1$) in the following construction. We can construct the union of confidence intervals as follows
 \[
CI_{1-\alpha} = \bigcup_{\gamma \in \overline{\Gamma}_W} CI_{1-\alpha}(\gamma)
\]
where $CI_{1-\alpha}(\gamma) = \widehat{\theta}_{ME} (W; \gamma) \pm z_{1-\alpha/2} \sqrt{\widehat{V}_{ME}(W;\gamma)/n}$ and $z_{1-\alpha/2}$ is the $(1-\alpha/2)$ quantile of the standard normal distribution. Reporting $CI_{1-\alpha}(\gamma)$ as a function of $(\gamma_1, \gamma_2)$ allows researchers to do transparent sensitivity analysis for assumptions about misspecification and identification strength. Since for each $\gamma \in \overline\Gamma_W$, $CI_{1-\alpha}(\gamma)$ covers $\theta_W(\gamma)$ with probability $1-\alpha$ asymptotically, and $\gamma(W) \in \overline\Gamma_W$, we have  valid coverage for $\theta_W$,
\[
\liminf_{n\rightarrow \infty} P(\theta_W \in CI_{1-\alpha} ) \geq 1-\alpha.
\]

Conley et al. (2012) consider a similar type of CI based on TSLS estimates $\widehat{\theta}(\gamma)$ and the conventional TSLS standard errors, with $\gamma$  denoting the direct effect of instruments. The key difference from  our approach is that we use the misspecification-efficient estimator $\widehat{\theta}_{ME}$ with the misspecification-efficient variance $V_{ME}(W)$ which was shown to be more efficient than the standard misspecification-robust GMM variances. We use the  augmented moment system and exploit additional information from the Jacobian with the potential efficiency gains from the moment and Jacobian correlations.

\section{Split-Sample ME Estimator}
\label{sec:splitsample}

The bootstrap provides a consistent estimate of the misspecification-efficient variance, but the oracle ME point estimate that actually achieves that variance is computable only if  the ``degree of misspecification" $(\gamma_1(W) = E[g(X_i, \theta_W)])$ and the ``identification strength" $(\gamma_2(W) = E[G(X_i, \theta_W)])$ parameters are known, which were used for recentering in the ME estimator. While these components can be consistently estimated from the full sample, the population moments and Jacobian cannot be replaced by their sample analogs, which suffers from a degeneracy in our setup by collapsing ME estimator back to the standard GMM estimator that invalidates the ME variance reduction.

We consider a sample-splitting ME estimator which uses one-half of a sample to estimate $\gamma(W) = (\gamma_1(W)', vec(\gamma_2(W)')')'$. Then estimated parameters $\widehat{\gamma}(W)$ are used to construct the final ME estimates using the other half of the samples. This construction parallels the split-sample IV estimator of Angrist and Krueger (1995), who use the first-half of the sample to estimate the population-first stage $E[Z'Z]^{-1}E[Z'X]$ and use it in the IV estimates in the other half sample. 

Specifically, we consider the following steps to construct the sample-split estimator for $\widehat{\theta}_{ME} (W)$;

\begin{enumerate}
\item Randomly divide the samples into $K=2$ independent subsamples of roughly equal size $n/2$, and denote $\{Y_1,X_1\}, \{Y_2, X_2\}$ as data matrices for each subsample. 
\item Use all observations with $\{ Y_2,X_2\}$ to estimate the $\gamma(W) =(\gamma_1(W)',\allowbreak vec(\gamma_2(W)')')'$ with the GMM estimator $\widehat{\theta}_{W} $ and let these leave-one-fold-out estimators as $\widehat{\gamma}_{-1} (W)$.
\item Solve the single-split ME estimator $\widehat{\theta}^{SS}(W)$ using observations $\{X_1,Y_1\}$ with the $\widehat{\gamma}_{-1}(W)$, and calculate the standard error using ``efficient" misspecification-robust variance formula $\widehat{V}_{ME}(W)$.
\item Repeat the above procedure, $s = 1, ..., S$, and set $\widehat{\theta}^{RSS}(W) = \frac{1}{S} \sum_{s=1}^{S} \widehat{\theta}_{s}^{SS}(W)$, $SD (\widehat{\theta}^{RSS}(W))^2 = \frac{1}{S-1} \sum_{s=1}^{S} (\widehat{\theta}_{s}^{SS}(W) - \widehat{\theta}^{RSS}(W))^2$. We also use the variance estimators considered in Chernozhukov et al. (2018) based on median, i.e., $ median \{ \widehat{V}_{ME,s}(W) +  (\widehat{\theta}_{s}^{SS}(W) - \widehat{\theta}_{median}^{SS}(W))(\widehat{\theta}_{s}^{SS}(W) - \widehat{\theta}_{median}^{SS}(W))'  \}_{s=1}^{S}$.
\end{enumerate}

Because the sample size is $ n/2$ in the split-sample ME estimator, the asymptotic variance is at least twice the asymptotic variance of the oracle-ME. Repeated single sample-splitting procedure can be done to reduce the variance from the single sample. Furthermore, when we construct ME estimator based on $W = \Sigma_{11,2}^{-1}$ in the linear model, we don't need the recentering part $\gamma_1(W)$ due to the FOC with $W^* = \Sigma_{11,2}^{-1}$. This helps reduce computational costs and estimation noise of the sample-splitting estimate. We do not consider the general symmetric sample-splitting estimator (i.e., swap the role of each subsample and then average) or cross-fitting estimator because the symmetry of averaging again has similar degeneracy problems.

Table \ref{tab:sim_rss} reports the Monte Carlo results of the repeated sample-splitting ME estimator for the simulation design in Section \ref{sec:simulation} under $W = I$ and $W = \widehat\Sigma_{11,2}^{-1}$, which complement Table \ref{tab:sim} in the main paper. We report the median of the $S=100$ repeated sample-split estimates (RSS ME median) with the Wald CI based on the median variance estimator described above. The RSS ME median CI has coverage close to or above the nominal level across all designs (0.945-0.987), and its length is larger than those of the Wald CI with the misspecification-robust SE and the DR percentile CI in Table \ref{tab:sim}, as expected since each split uses only half of the sample. The SD of the RSS ME median is close to that of the standard GMM estimator.

\begin{table}[h]
\footnotesize
\setlength{\tabcolsep}{3pt}
\centering
\caption{Monte Carlo results for the repeated sample-splitting ME estimator}
\label{tab:sim_rss}
\begin{tabular}{l *{12}{c}}
\toprule
 & \multicolumn{3}{c}{$\delta=0$} & \multicolumn{3}{c}{$\delta=0.5$} & \multicolumn{3}{c}{$\delta=1$} & \multicolumn{3}{c}{$\delta=2$} \\
\cmidrule(lr){2-4} \cmidrule(lr){5-7} \cmidrule(lr){8-10} \cmidrule(lr){11-13}
RSS ME median & SD & Cov & Len & SD & Cov & Len & SD & Cov & Len & SD & Cov & Len \\
\midrule
\multicolumn{13}{l}{\textit{Panel A: $W = I$}} \\
\qquad $n=200$   & 1.04 & 0.974 & 0.466 & 1.04 & 0.976 & 0.471 & 1.67 & 0.963 & 0.710 & 3.18 & 0.970 & 1.378 \\
\qquad $n=500$   & 0.63 & 0.961 & 0.265 & 0.66 & 0.955 & 0.273 & 1.04 & 0.948 & 0.413 & 1.95 & 0.956 & 0.804 \\
\qquad $n=1000$  & 0.44 & 0.958 & 0.181 & 0.46 & 0.955 & 0.187 & 0.72 & 0.945 & 0.284 & 1.38 & 0.945 & 0.556 \\
\midrule
\multicolumn{13}{l}{\textit{Panel B: $W = \widehat\Sigma_{11,2}^{-1}$}} \\
\qquad $n=200$   & 1.04 & 0.971 & 0.458 & 0.97 & 0.982 & 0.480 & 1.52 & 0.985 & 0.737 & 3.01 & 0.987 & 1.440 \\
\qquad $n=500$   & 0.63 & 0.958 & 0.264 & 0.61 & 0.975 & 0.280 & 0.95 & 0.977 & 0.439 & 1.88 & 0.978 & 0.858 \\
\qquad $n=1000$  & 0.44 & 0.957 & 0.180 & 0.43 & 0.971 & 0.195 & 0.67 & 0.973 & 0.304 & 1.32 & 0.970 & 0.600 \\
\bottomrule
\end{tabular}
\begin{flushleft}
\scriptsize
\textit{Notes.} Repeated sample-splitting ME estimator (RSS ME): the median of the repeated sample-split estimates with $\pm 1.96\,\mathrm{SD}$. ``SD" is the MC standard deviation of the point estimate across MC draws, reported relative to the correctly-specified benchmark SD as in Table \ref{tab:sim} (GMM with $\delta = 0, W = I$, $n=200$). ``Cov" is empirical $95\%$ coverage; ``Len" is median CI length. Based on $2{,}000$ Monte Carlo replications and $S=100$ sample-split replications. See Table \ref{tab:sim} for the results of the other estimators under the same designs.
\end{flushleft}
\end{table}

Tables \ref{tab:emp_card_rss} and \ref{tab:emp_ak_rss} report the RSS ME median estimates for the empirical applications in Section \ref{sec:empirical}, with the same specifications and weighting matrices as in Tables \ref{tab:emp_card} and \ref{tab:emp_ak}. The RSS ME median estimates are close to the GMM point estimates, while their standard errors are larger than the misspecification-robust SE.

\begin{table}[h]
\footnotesize
\centering
\caption{Repeated Sample-Splitting (RSS)  ME: Card (1995)}
\label{tab:emp_card_rss}
\begin{tabular}{l cc cc}
\toprule
 & \multicolumn{2}{c}{ $exp/exp^2$ as exogenous} & \multicolumn{2}{c}{$exp/exp^2$ as endogenous} \\
\cmidrule(lr){2-3} \cmidrule(lr){4-5}
 & (1) & (2) & (3) & (4) \\
\midrule
$W = I$ (identity)  & 0.1831 & 0.1899 & 0.1792  &0.1786  \\
\quad  & (0.0648) & (0.0619) & (0.0646) & (0.0645) \\
$W = (Z'Z)^{-1}$ (TSLS)  & 0.1748  & 0.1740  & 0.1693 & 0.1726 \\
\quad & (0.0674) & (0.0580) & (0.0613) & (0.0592) \\
$W = \widehat\Sigma_{11,2}^{-1}$ (efficient)  & 0.1619 & 0.1692& 0.1642 & 0.1654  \\
\quad  & (0.0650) & (0.0602) & (0.0632) & (0.0603) \\
\bottomrule
\end{tabular}
\begin{flushleft}
\scriptsize
\textit{Notes.} We report the RSS ME median estimates of the returns to education with standard errors (in parentheses), for the specifications in Table \ref{tab:emp_card}; see the notes in Table \ref{tab:emp_card} for the specifications and the other estimators. RSS ME median estimates are the median across $S=100$ replications, with standard errors for the median similarly computed as in Chernozhukov et al. (2018). $n=3010$.
\end{flushleft}
\end{table}

\begin{table}[h]
\footnotesize
\centering
\setlength{\tabcolsep}{3pt}
\caption{Repeated Sample-Splitting (RSS)  ME: Angrist and Krueger (1991)}
\label{tab:emp_ak_rss}
\resizebox{\textwidth}{!}{%
\begin{tabular}{l cc cc cc}
\toprule
 & \multicolumn{2}{c}{(1) QOB (m=3)} & \multicolumn{2}{c}{(2) + QOB $\times$ YOB (m=30)} & \multicolumn{2}{c}{(3) + QOB $\times$ State (m=180)} \\
\cmidrule(lr){2-3} \cmidrule(lr){4-5} \cmidrule(lr){6-7}
 & basic & flexible & basic & flexible & basic & flexible \\
\midrule
$W = I$ (identity)        & 0.1102 & 0.1099 & 0.0828 & 0.0804 & 0.0972 & 0.0890 \\
\quad                    & (0.0342) & (0.0366) & (0.0282) & (0.0337) & (0.0243) & (0.0264) \\
$W = (Z'Z)^{-1}$ (TSLS)   & 0.0998 & 0.1012 & 0.0734 & 0.0787 & 0.0774 & 0.0740 \\
\quad                    & (0.0332) & (0.0343) & (0.0303) & (0.0293) & (0.0255) & (0.0202) \\
$W = \widehat\Sigma_{11,2}^{-1}$ (efficient) & 0.0998 & 0.0979 & 0.0781 & 0.0707 & 0.0795 & 0.0754 \\
\quad                    & (0.0322) & (0.0348) & (0.0288) & (0.0300) & (0.0267) & (0.0215) \\
\bottomrule
\end{tabular}}
\begin{flushleft}
\scriptsize
\textit{Notes.} We report the RSS ME median estimates of the returns to education with standard errors (in parentheses), for the specifications in Table \ref{tab:emp_ak}; see the notes in Table \ref{tab:emp_ak} for the specifications and the other estimators. RSS ME median estimates are the median across $S=100$ replications, with standard errors for the median similarly computed as in Chernozhukov et al. (2018). $n=329,509$.
\end{flushleft}
\end{table}

\section{Additional Simulations}\label{sec:app_weak_sim}

In this appendix, we report the additional simulation results for the same simulation design in Section  \ref{sec:simulation} with the conventional optimal weighting $W = \widehat\Sigma_{11}^{-1}$ and the moderately weak instrument setup. 

Table \ref{tab:sim_s11} reports the simulation results with the correct-specification optimal weighting matrix  $W = \widehat\Sigma_{11}^{-1}$. The GMM estimators under the two weights are nearly identical when $\delta \in \{0, 0.5\}$ and the population pseudo-true values are meaningfully different only when the misspecification is large, i.e., $\delta =2$. Comparing Table \ref{tab:sim} (with the efficient ME weight $W = \widehat\Sigma_{11,2}^{-1}$) and Table \ref{tab:sim_s11}, we find that SD of the GMM estimator under $W=\widehat\Sigma_{11,2}^{-1}$ is slightly smaller than  that of $W=\widehat\Sigma_{11}^{-1}$, especially when the degree of misspecification is large ($\delta \in \{1, 2\}$). The Wald CI with the misspecification-robust SE under-covers under $W= \widehat\Sigma_{11}^{-1}$ (0.90-0.92 for $\delta=1,2$). The oracle ME bound is identical across the two weights, and the RSS-ME median CI lengths, and DR CI lengths are very similar across all $(n, \delta)$, with DR CI coverage close to the nominal level.

\begin{table}[h]
\footnotesize

\setlength{\tabcolsep}{3pt}
\centering
\caption{Monte Carlo results under $W = \widehat\Sigma_{11}^{-1}$}
\label{tab:sim_s11}
\begin{tabular}{l *{12}{c}}
\toprule
 & \multicolumn{3}{c}{$\delta=0$} & \multicolumn{3}{c}{$\delta=0.5$} & \multicolumn{3}{c}{$\delta=1$} & \multicolumn{3}{c}{$\delta=2$} \\
\cmidrule(lr){2-4} \cmidrule(lr){5-7} \cmidrule(lr){8-10} \cmidrule(lr){11-13}
Estimator & SD & Cov & Len & SD & Cov & Len & SD & Cov & Len & SD & Cov & Len \\
\midrule
\multicolumn{13}{l}{ \qquad $n=200$} \\
GMM + Conv.\ SE        & 1.00 & 0.942 & 0.384 & 0.98 & 0.929 & 0.347 & 1.57 & 0.867 & 0.471 & 2.99 & 0.870 & 0.902 \\
GMM + Robust SE        & 1.00 & 0.942 & 0.386 & 0.98 & 0.938 & 0.367 & 1.57 & 0.903 & 0.527 & 2.99 & 0.913 & 1.019 \\
Oracle ME              & 0.93 & 0.944 & 0.367 & 0.87 & 0.949 & 0.349 & 1.08 & 0.945 & 0.422 & 1.78 & 0.942 & 0.695 \\
RSS ME median          & 1.04 & 0.968 & 0.455 & 1.01 & 0.975 & 0.477 & 1.60 & 0.978 & 0.725 & 3.07 & 0.981 & 1.401 \\
HH Bootstrap           & 1.00 & 0.942 & 0.394 & 0.90 & 0.926 & 0.357 & 1.24 & 0.872 & 0.483 & 2.39 & 0.875 & 0.924 \\
DR Bootstrap         & 1.02 & 0.943 & 0.400 & 1.02 & 0.951 & 0.404 & 1.55 & 0.941 & 0.604 & 3.02 & 0.947 & 1.181 \\
\midrule
\multicolumn{13}{l}{\qquad $n=500$} \\
GMM + Conv.\ SE        & 0.62 & 0.955 & 0.245 & 0.63 & 0.919 & 0.223 & 0.99 & 0.875 & 0.303 & 1.85 & 0.881 & 0.585 \\
GMM + Robust SE        & 0.62 & 0.955 & 0.246 & 0.63 & 0.935 & 0.234 & 0.99 & 0.911 & 0.338 & 1.85 & 0.920 & 0.654 \\
Oracle ME              & 0.58 & 0.945 & 0.232 & 0.56 & 0.945 & 0.221 & 0.69 & 0.944 & 0.267 & 1.09 & 0.952 & 0.440 \\
RSS ME median          & 0.63 & 0.962 & 0.263 & 0.64 & 0.967 & 0.277 & 1.00 & 0.967 & 0.429 & 1.86 & 0.965 & 0.827 \\
HH Bootstrap           & 0.62 & 0.954 & 0.247 & 0.56 & 0.918 & 0.224 & 0.77 & 0.874 & 0.305 & 1.49 & 0.880 & 0.591 \\
DR Bootstrap         & 0.63 & 0.955 & 0.249 & 0.64 & 0.957 & 0.254 & 0.96 & 0.939 & 0.382 & 1.87 & 0.946 & 0.742 \\
\midrule
\multicolumn{13}{l}{\qquad $n=1000$} \\
GMM + Conv.\ SE        & 0.43 & 0.953 & 0.174 & 0.44 & 0.922 & 0.158 & 0.69 & 0.877 & 0.216 & 1.34 & 0.874 & 0.415 \\
GMM + Robust SE        & 0.43 & 0.953 & 0.174 & 0.44 & 0.935 & 0.167 & 0.69 & 0.916 & 0.241 & 1.34 & 0.905 & 0.465 \\
Oracle ME              & 0.40 & 0.958 & 0.164 & 0.39 & 0.948 & 0.156 & 0.47 & 0.950 & 0.189 & 0.79 & 0.951 & 0.311 \\
RSS ME median          & 0.44 & 0.957 & 0.180 & 0.45 & 0.964 & 0.191 & 0.70 & 0.961 & 0.296 & 1.35 & 0.956 & 0.568 \\
HH Bootstrap           & 0.44 & 0.953 & 0.175 & 0.40 & 0.921 & 0.159 & 0.54 & 0.881 & 0.217 & 1.04 & 0.873 & 0.418 \\
DR Bootstrap         & 0.44 & 0.955 & 0.176 & 0.45 & 0.953 & 0.180 & 0.67 & 0.942 & 0.270 & 1.31 & 0.939 & 0.521 \\
\bottomrule
\end{tabular}
\begin{flushleft}
\scriptsize
\textit{Notes.} We consider (i) standard GMM estimator with the conventional SE; (ii) GMM with the misspecification-robust SE; (iii) the oracle ME estimator with efficient SE, $V_{ME}(W)$; (iv) Repeated sample-splitting ME estimator (RSS ME), the median repeated sample-split estimates with $\pm 1.96\,\mathrm{SD}$; (v) Hall and Horowitz (1996) (HH) percentile bootstrap; and (vi) the Double-Recentered percentile bootstrap  (DR). For GMM, Oracle ME, and RSS ME, ``SD" is the MC standard deviation of the point estimate across MC draws; for HH and DR bootstrap, ``SD" is the average bootstrap standard deviation across MC draws. ``SD" is reported relative to the correctly-specified benchmark SD (GMM with $\delta = 0, W = I$, $n=200$), so values $<1$ ($>1$) indicate higher (lower) precision than this benchmark. ``Cov" is empirical $95\%$ coverage; ``Len" is median CI length. Based on $2{,}000$ Monte Carlo replications, $B=1{,}000$ bootstrap draws, $S=100$ sample-split replications.
\end{flushleft}
\end{table}

Tables \ref{tab:sim_I_weak} and \ref{tab:sim_eff_weak} report simulation results under the moderately weak instrument setup. In particular, we set the first-stage coefficient so that the scaled concentration parameter ($\mu^2/m$) is 10 for $n=200$. We note that our theory does not cover the weak instruments setup. We report the median bootstrap standard deviation as a robust measure in this setup.

The results are qualitatively similar to the strong instrument cases; conventional SE/HH bootstrap undercover as $\delta$ increases, while the misspecification-robust SE Wald CI/sample-splitting methods, and DR percentile CI have coverage close to the nominal level (the robust-SE Wald CI slightly undercovers). The oracle ME has smaller SD/lengths, while the coverage for weak instruments is slightly below the nominal level at $n=200$, although the coverage improves when $n$ increases. 

\begin{table}[h]
\footnotesize

\setlength{\tabcolsep}{3pt}
\centering
\caption{Monte Carlo results under weak instrument setup ($W = I$)}
\label{tab:sim_I_weak}
\begin{tabular}{l *{12}{c}}
\toprule
 & \multicolumn{3}{c}{$\delta=0$} & \multicolumn{3}{c}{$\delta=0.5$} & \multicolumn{3}{c}{$\delta=1$} & \multicolumn{3}{c}{$\delta=2$} \\
\cmidrule(lr){2-4} \cmidrule(lr){5-7} \cmidrule(lr){8-10} \cmidrule(lr){11-13}
Estimator & SD & Cov & Len & SD & Cov & Len & SD & Cov & Len & SD & Cov & Len \\
\midrule
\multicolumn{13}{l}{\qquad $n=200$} \\
GMM + Conv.\ SE        & 1.00 & 0.946 & 0.847 & 1.51 & 0.898 & 1.055 & 2.92 & 0.859 & 2.036 & 6.29 & 0.844 & 4.173 \\
GMM + Robust SE        & 1.00 & 0.950 & 0.873 & 1.51 & 0.949 & 1.275 & 2.92 & 0.917 & 2.485 & 6.29 & 0.918 & 5.135 \\
Oracle ME              & 0.89 & 0.920 & 0.781 & 0.93 & 0.936 & 0.846 & 1.34 & 0.927 & 1.169 & 2.47 & 0.931 & 2.061 \\
RSS ME median          & 1.08 & 0.989 & 1.393 & 1.53 & 0.980 & 1.854 & 2.99 & 0.965 & 3.398 & 6.24 & 0.958 & 6.950 \\
HH Bootstrap           & 0.95 & 0.951 & 0.914 & 1.21 & 0.914 & 1.140 & 2.36 & 0.895 & 2.200 & 4.81 & 0.898 & 4.524 \\
DR Bootstrap         & 1.03 & 0.950 & 0.989 & 1.55 & 0.950 & 1.466 & 3.15 & 0.949 & 2.930 & 6.50 & 0.949 & 6.087 \\
\midrule
\multicolumn{13}{l}{\qquad $n=500$} \\
GMM + Conv.\ SE        & 0.59 & 0.958 & 0.541 & 0.94 & 0.891 & 0.690 & 1.74 & 0.893 & 1.354 & 3.63 & 0.892 & 2.865 \\
GMM + Robust SE        & 0.59 & 0.961 & 0.548 & 0.94 & 0.949 & 0.821 & 1.74 & 0.942 & 1.623 & 3.63 & 0.945 & 3.377 \\
Oracle ME              & 0.54 & 0.938 & 0.494 & 0.57 & 0.945 & 0.535 & 0.78 & 0.948 & 0.739 & 1.40 & 0.936 & 1.303 \\
RSS ME median          & 0.62 & 0.981 & 0.650 & 0.97 & 0.966 & 0.940 & 1.75 & 0.956 & 1.799 & 3.63 & 0.956 & 3.728 \\
HH Bootstrap           & 0.58 & 0.958 & 0.560 & 0.74 & 0.896 & 0.713 & 1.46 & 0.904 & 1.411 & 3.09 & 0.910 & 2.972 \\
DR Bootstrap         & 0.59 & 0.960 & 0.576 & 0.89 & 0.942 & 0.867 & 1.79 & 0.951 & 1.738 & 3.74 & 0.957 & 3.602 \\
\midrule
\multicolumn{13}{l}{\qquad $n=1000$} \\
GMM + Conv.\ SE        & 0.41 & 0.951 & 0.388 & 0.62 & 0.898 & 0.489 & 1.24 & 0.899 & 0.966 & 2.54 & 0.902 & 2.041 \\
GMM + Robust SE        & 0.41 & 0.952 & 0.390 & 0.62 & 0.948 & 0.581 & 1.24 & 0.936 & 1.157 & 2.54 & 0.952 & 2.404 \\
Oracle ME              & 0.36 & 0.946 & 0.349 & 0.39 & 0.956 & 0.378 & 0.55 & 0.948 & 0.523 & 0.96 & 0.941 & 0.922 \\
RSS ME median          & 0.42 & 0.965 & 0.424 & 0.63 & 0.959 & 0.616 & 1.24 & 0.948 & 1.207 & 2.54 & 0.959 & 2.498 \\
HH Bootstrap           & 0.41 & 0.950 & 0.395 & 0.51 & 0.898 & 0.496 & 1.01 & 0.894 & 0.983 & 2.15 & 0.911 & 2.079 \\
DR Bootstrap         & 0.41 & 0.952 & 0.401 & 0.61 & 0.952 & 0.596 & 1.24 & 0.942 & 1.198 & 2.57 & 0.955 & 2.495 \\
\bottomrule
\end{tabular}
\begin{flushleft}
\end{flushleft}
\end{table}

\begin{table}[h]
\footnotesize

\setlength{\tabcolsep}{3pt}
\centering
\caption{Monte Carlo results under weak instrument setup ($W = \widehat\Sigma_{11,2}^{-1}$)}
\label{tab:sim_eff_weak}
\begin{tabular}{l *{12}{c}}
\toprule
 & \multicolumn{3}{c}{$\delta=0$} & \multicolumn{3}{c}{$\delta=0.5$} & \multicolumn{3}{c}{$\delta=1$} & \multicolumn{3}{c}{$\delta=2$} \\
\cmidrule(lr){2-4} \cmidrule(lr){5-7} \cmidrule(lr){8-10} \cmidrule(lr){11-13}
Estimator & SD & Cov & Len & SD & Cov & Len & SD & Cov & Len & SD & Cov & Len \\
\midrule
\multicolumn{13}{l}{\qquad $n=200$} \\
GMM + Conv.\ SE        & 1.00 & 0.943 & 0.841 & 1.48 & 0.895 & 1.026 & 2.91 & 0.852 & 1.879 & 6.23 & 0.818 & 3.686 \\
GMM + Robust SE        & 1.00 & 0.946 & 0.862 & 1.48 & 0.941 & 1.234 & 2.91 & 0.923 & 2.458 & 6.23 & 0.916 & 5.220 \\
Oracle ME              & 0.89 & 0.920 & 0.781 & 0.93 & 0.936 & 0.846 & 1.34 & 0.927 & 1.169 & 2.47 & 0.931 & 2.061 \\
RSS ME median          & 1.09 & 0.987 & 1.382 & 1.46 & 0.982 & 1.885 & 2.88 & 0.973 & 3.593 & 6.02 & 0.974 & 7.333 \\
HH Bootstrap           & 0.94 & 0.950 & 0.910 & 1.16 & 0.919 & 1.113 & 2.16 & 0.890 & 2.031 & 4.20 & 0.865 & 3.970 \\
DR Bootstrap         & 1.04 & 0.953 & 0.990 & 1.54 & 0.948 & 1.456 & 3.14 & 0.942 & 2.938 & 6.78 & 0.930 & 6.260 \\
\midrule
\multicolumn{13}{l}{\qquad $n=500$} \\
GMM + Conv.\ SE        & 0.59 & 0.955 & 0.540 & 0.92 & 0.897 & 0.675 & 1.77 & 0.872 & 1.250 & 3.88 & 0.841 & 2.535 \\
GMM + Robust SE        & 0.59 & 0.957 & 0.546 & 0.92 & 0.940 & 0.800 & 1.77 & 0.941 & 1.658 & 3.88 & 0.932 & 3.564 \\
Oracle ME              & 0.54 & 0.938 & 0.494 & 0.57 & 0.945 & 0.535 & 0.78 & 0.948 & 0.739 & 1.40 & 0.936 & 1.303 \\
RSS ME median          & 0.61 & 0.981 & 0.651 & 0.94 & 0.978 & 0.970 & 1.76 & 0.971 & 1.911 & 3.81 & 0.969 & 3.990 \\
HH Bootstrap           & 0.57 & 0.958 & 0.559 & 0.72 & 0.895 & 0.695 & 1.34 & 0.882 & 1.301 & 2.70 & 0.846 & 2.615 \\
DR Bootstrap         & 0.59 & 0.962 & 0.574 & 0.89 & 0.941 & 0.861 & 1.83 & 0.950 & 1.771 & 3.93 & 0.942 & 3.804 \\
\midrule
\multicolumn{13}{l}{\qquad $n=1000$} \\
GMM + Conv.\ SE        & 0.41 & 0.950 & 0.387 & 0.61 & 0.900 & 0.479 & 1.28 & 0.861 & 0.889 & 2.72 & 0.844 & 1.814 \\
GMM + Robust SE        & 0.41 & 0.952 & 0.390 & 0.61 & 0.943 & 0.571 & 1.28 & 0.945 & 1.181 & 2.72 & 0.945 & 2.556 \\
Oracle ME              & 0.36 & 0.946 & 0.349 & 0.39 & 0.956 & 0.378 & 0.55 & 0.948 & 0.523 & 0.96 & 0.941 & 0.922 \\
RSS ME median          & 0.42 & 0.962 & 0.424 & 0.62 & 0.970 & 0.644 & 1.28 & 0.963 & 1.291 & 2.72 & 0.959 & 2.745 \\
HH Bootstrap           & 0.41 & 0.952 & 0.394 & 0.50 & 0.897 & 0.487 & 0.93 & 0.863 & 0.904 & 1.90 & 0.842 & 1.837 \\
DR Bootstrap         & 0.41 & 0.951 & 0.401 & 0.61 & 0.950 & 0.598 & 1.26 & 0.944 & 1.224 & 2.72 & 0.947 & 2.635 \\
\bottomrule
\end{tabular}
\begin{flushleft}
\scriptsize
\end{flushleft}
\end{table}

\clearpage
\section{Dynamic Panel Data Model}
\label{sec:empirical_ab}

In this Appendix, we consider the following dynamic panel regression of Arellano and Bond (1991), 
\begin{equation}\label{eq:empirica:ab}
 y_{it} = \alpha_1  y_{i t-1} + \alpha_2  y_{i t-2} + x_{it}'\beta + \lambda_t + \eta_i + u_{i t},
 \end{equation}
where $y_{it}$ is the log of employment for company $i$ at time $t$,  $x_{it} = (w_{it}, w_{i t-1}, k_{it},\allowbreak ys_{it}, ys_{i t-1})$ is a vector of strictly exogenous regressors, $w_{it}$ is the log of the real product wage, $k_{it}$ is the log of gross capital, and $ys_{it}$ is the log of industry output. The data set includes an unbalanced panel of 140 UK firms over the years 1976-1984. We consider the same model as Windmeijer (2005, Table 2), which is the same as Arellano and Bond (1991, Table 4(b)). 

In Table \ref{tab:ab1991}, we consider the first-difference GMM estimator of Arellano and Bond (1991) across different weighting matrices and the instrument sets $Z_i$. Besides $W=I$, we consider the one-step weights $W = (\sum_i Z_i'HZ_i)^{-1}$ (TSLS-type), where $H$ is a matrix with 2 on the diagonal, and -1 on the first off-diagonals.  We also consider the conventional two-step efficient weight $W = (\sum_i Z_i' \hat{u}_i \hat{u}_i^{\prime}Z_i)^{-1}$ with the one-step residuals $\hat{u}_i$. The ME efficient weight $W=\widehat\Sigma_{11,2}^{-1}$ is computed from the firm-level moments and Jacobians.  Panel A uses the full instrument sets using all available lagged levels, and Panel B uses the collapsed GMM-style instruments due to the concern of many instruments (Roodman (2009)). When using one-step and two-step with the full instrument set in the second and third columns in Table \ref{tab:ab1991}, we replicate Windmeijer (2005) Table 2 and Arellano and Bond (1991) Table 4(b) for the point estimates/conventional SE; we add the misspecification-robust SE, and the ME efficiency bounds, and the DR percentile bootstrap CI. 

The results from Panel A should be read with caution. First of all, with the same concern as in Roodman (2009), the full instrument set does not reject (J-test p-value 0.220), while the collapsed instrument set rejects the J-test at the 5\%  (p-value 0.040). Furthermore, the ME bound in Panel A is much smaller than the conventional/misspecification-robust SEs, which can be due to the small-sample problems of the estimated weighting matrix with many instruments, including $\widehat{\Sigma}_{11,2}^{-1}$, rather than an efficiency gain.  In Panel B with the collapsed instruments, the estimates of the AR(1) coefficient $\alpha_1$ range from $0.65$ to $0.85$, and the misspecification-robust SE  (0.31-1.60) exceeds the conventional SE (0.26-0.40) with the ME bound 0.194, in which the ME-GMM bootstrap SD tracks it closely (0.165-0.173). Under $W=I$, the gap between the ME bound and the conventional/robust SEs is much bigger, and the DR bootstrap CIs are wider. Under the efficient weights (two-step, ME-efficient), the ME efficiency gain is much larger than in the cross-sectional examples. This is partly because the first-differencing makes the error terms and the lagged dependent variable correlated, so that the first term of the decomposition in Remark \ref{rem:gap} is large.

\begin{table}[p]
\scriptsize\centering
\caption{Dynamic panel: Arellano and Bond (1991) employment equation}
\label{tab:ab1991}
\begin{tabular}{l l cccc}
\toprule
 & & \multicolumn{4}{c}{Weight matrix $W$} \\
\cmidrule(lr){3-6}
 & & Identity & One-step & Two-step & ME-efficient \\
Coefficient & & $I$ & $(\sum_i Z_i'HZ_i)^{-1}$ & $(\sum_i Z_i'\hat u_i\hat u_i'Z_i)^{-1}$ & $\widehat\Sigma_{11,2}^{-1}$ \\
\midrule
\multicolumn{6}{l}{\textit{Panel A. Full instrument set}} \\
\multicolumn{6}{l}{\qquad $m=27$, $F=3.2$, $J=30.11$ ($p=0.220$)} \\
\addlinespace
$\alpha_1 $ & GMM point & $0.572$ & $0.535$ & $0.474$ & $0.442$ \\
           & Conv.\ SE & $(0.311)$ & $(0.166)$ & $(0.086)$ & $(0.284)$ \\
           & Misspec.\ SE & $(0.412)$ & $(0.338)$ & $(0.101)$ & $(0.486)$ \\
           & ME bound & $(0.027)$ & $(0.027)$ & $(0.027)$ & $(0.027)$ \\
           & ME-GMM Boot SD & $(0.023)$ & $(0.024)$ & $(0.024)$ & $(0.024)$ \\
           & DR 95\% CI & [$-1.09,1.95$] & [$-0.41,1.26$] & [$0.18,0.71$] & [$-0.49,1.67$] \\
\addlinespace
$\alpha_2$ & GMM point & $0.103$ & $-0.075$ & $-0.053$ & $-0.055$ \\
           & Conv.\ SE & $(0.224)$ & $(0.068)$ & $(0.027)$ & $(0.100)$ \\
           & Misspec.\ SE & $(0.359)$ & $(0.104)$ & $(0.029)$ & $(0.143)$ \\
           & ME bound & $(0.008)$ & $(0.008)$ & $(0.008)$ & $(0.008)$ \\
           & ME-GMM Boot SD & $(0.008)$ & $(0.008)$ & $(0.008)$ & $(0.008)$ \\
           & DR 95\% CI & [$-0.95,0.97$] & [$-0.32,0.23$] & [$-0.12,0.02$] & [$-0.41,0.20$] \\
\midrule
\multicolumn{6}{l}{\textit{Panel B. Collapsed instrument set}} \\
\multicolumn{6}{l}{\qquad $m=7$, $F=5.3$,  $J=11.63$ ($p=0.040$)} \\
\addlinespace
$\alpha_1$ & GMM point & $0.705$ & $0.823$ & $0.854$ & $0.650$ \\
           & Conv.\ SE & $(0.400)$ & $(0.293)$ & $(0.270)$ & $(0.264)$ \\
           & Misspec.\ SE & $(1.601)$ & $(0.773)$ & $(0.306)$ & $(0.593)$ \\
           & ME bound & $(0.194)$ & $(0.194)$ & $(0.194)$ & $(0.194)$ \\
           & ME-GMM Boot SD & $(0.165)$ & $(0.165)$ & $(0.167)$ & $(0.173)$ \\
           & DR 95\% CI & [$-6.25,5.32$] & [$-1.84,4.23$] & [$-0.24,2.96$] & [$-1.75,3.04$] \\
\addlinespace
$\alpha_2$ & GMM point & $0.449$ & $-0.145$ & $-0.170$ & $-0.088$ \\
           & Conv.\ SE & $(0.312)$ & $(0.069)$ & $(0.065)$ & $(0.067)$ \\
           & Misspec.\ SE & $(1.760)$ & $(0.163)$ & $(0.065)$ & $(0.134)$ \\
           & ME bound & $(0.046)$ & $(0.046)$ & $(0.046)$ & $(0.046)$ \\
           & ME-GMM Boot SD & $(0.042)$ & $(0.041)$ & $(0.041)$ & $(0.042)$ \\
           & DR 95\% CI & [$-5.04,3.63$] & [$-0.90,0.41$] & [$-0.53,0.05$] & [$-0.64,0.36$] \\
\bottomrule 
\end{tabular}
\par\smallskip\parbox{0.85\linewidth}{\scriptsize\textit{Notes:} We report  point estimates of $\alpha_1, \alpha_2$, which are the coefficients of $y_{i t-1}$ and $y_{i t-2}$ in the dynamic panel model of Windmeijer (2005), Arellano and Bond (1991) (equation \eqref{eq:empirica:ab}). We use the first-differenced dynamic panel GMM estimator. ``Conv. SE'' is the conventional standard errors, and  ``Misspec. SE'' is the misspecification-robust GMM SEs. ``ME efficiency bound'' is the oracle ME standard error $\widehat V_{ME}(W)^{1/2}/\sqrt{n}$ with $V_{ME}(W)=(G'\Sigma_{11,2}^{-1}G)^{-1}$ as in Corollary~\ref{cor:optimal_matrix}.  ``ME-GMM Boot SD'' is the SD of the bootstrap ME-GMM estimator. ``DR 95\% CI'' is the double-recentered percentile interval and is reported in brackets. We  use $B=2,000$ bootstrap draws. The dataset includes $N=140$ firms, 611 first-differenced observations in 1979-1984. $m$ is the number of excluded instruments. $F$ is the first-stage $F$ for $\Delta y_{i,t-1}$ and $J$ is the Sargan-Hansen J statistic for the two-step estimator.}
\end{table}

\section{The Decomposition of the ME Efficiency Gain}
\label{sec:decomp_ME}

This appendix provides a detailed derivation of the decomposition in Remark \ref{rem:gap} and the closed-form formula in the linear homoskedastic IV model. The closed form is useful for understanding the ratio of the conventional SE, the misspecification-robust SE and the ME efficiency bound. 

We first derive the general form of the misspecification-robust GMM variance $V_{GMM}(W, \theta_W)$ which is helpful for the calculations. Consider the linear moment $g(X_i,\theta)=g(X_i,0)+G(X_i)\theta$ with $G=E[G(X_i)]$ and $F(\theta)=0$. Fix $W>0$, and let $r_i = g(X_i, \theta_W) - g(\theta_W) - \Sigma_{12}\Sigma_{22}^{-1} (vec(G(X_i)') - E[vec(G(X_i)')])$. Recall that $A(W,\theta_W)=[G'W,\; g(\theta_W)'W\otimes I_p]$, and we have the FOC, $G'W g(\theta_W)=0$. We note that
\[
A(W, \theta_W) (\psi(X_i,\theta_W) - E[\psi(X_i,\theta_W)]) = G'W r_i + C_W (vec(G(X_i)') - E[vec(G(X_i)')]), 
\]
where $C_W = G'W \Sigma_{12}\Sigma_{22}^{-1} + (g(\theta_W)'W\otimes I_p)$ and $Var(r_i) = \Sigma_{11,2}$ and $Cov(r_i,\allowbreak vec(G(X_i)')) = 0$ by construction. Therefore, we have that

\[
V_{GMM}(W, \theta_W) = (G'WG)^{-1} (G'W \Sigma_{11,2} WG + C_W \Sigma_{22}C_W') (G'WG)^{-1}, 
\]
and  $V_{ME} = (G'\Sigma_{11,2}^{-1}G)^{-1}$ as in Corollary \ref{cor:optimal_matrix}.

Now we consider the linear IV model
\[
Y_i=X_i\theta_0+Z_i'\eta+\varepsilon_i,\qquad X_i=Z_i'\pi+v_i,
\]
with a scalar endogenous regressor $X_i$,  excluded instruments $Z_i\in\mathbb R^m$, the first-stage coefficient $\pi$, the direct effect of instruments $\eta$, and $E[\varepsilon_i\mid Z_i]=E[v_i\mid Z_i]=0$. The moment is $g(X_i, \theta) = Z_i (Y_i-X_i\theta)$, and the 
Jacobian is $G(X_i) = -Z_i X_i$, $G = E[G(X_i)] = -Q_{ZZ} \pi$ where $Q_{ZZ} = E[Z_i Z_i']>0$. The pseudo-true value $\theta_W$ is defined by the first-order condition $G'W g(\theta_W)=0$, and we can deduce that $\theta_W = \theta_0 + (\pi'Q_{ZZ}\eta)/(\pi'Q_{ZZ}\pi)$ when $W \propto Q_{ZZ}^{-1}$ (TSLS).  We define the residual evaluated at $\theta_W$,
\[
\varepsilon_{W,i} \equiv Y_i - X_i \theta_W = e_i + a_i
\]
where $e_i = \varepsilon_i - v_i (\theta_W-\theta_0), a_i = Z_i'(\eta - \pi(\theta_W-\theta_0))$.

We further assume the following conditions for the derivation: 1) (conditional homoskedasticity) $E[(\varepsilon_i,v_i)'(\varepsilon_i,v_i)\mid Z_i]$
does not depend on $Z_i$; 2) (local first-stage and the misspecification), $\pi=c_\pi/\sqrt n$ and $\eta=c_\eta/\sqrt n$. Under these assumptions, we have $E[e_i^2\mid Z_i]=\sigma_e^2$, $E[e_i v_i\mid Z_i]=\sigma_{ev}$ and $E[v_i^2\mid Z_i]=\sigma_v^2$ that are constants, and let $\rho=\sigma_{ev}/(\sigma_e\sigma_v)$. We impose the local misspecification for the moment as $g(\theta_W) = O(n^{-1/2})$, but not in the parameter as  $\theta_W-\theta_0=O(1)$. We also note that $a_i$ and $g(\theta_W)$ are $O(n^{-1/2})$, and thus $\rho = \lim corr (\varepsilon_{W,i}, X_i)$. 

We can then calculate the terms that appear in the variance formula of the $V_{ME}, V_{GMM}$, up to the order $O(n^{-1})$,
\[
\Sigma_{11} = \sigma_e^2 Q_{ZZ}, \quad \Sigma_{22} = \sigma_v^2 Q_{ZZ}, \quad \Sigma_{12} = -\sigma_{ev} Q_{ZZ}, \quad \Sigma_{12}\Sigma_{22}^{-1} = -\sigma_{ev}/\sigma_v^2 I_m,
\]
\[
\Sigma_{11,2} = \Sigma_{11}-\Sigma_{12}\Sigma_{22}^{-1}\Sigma_{21} =(\sigma_e^2 - \frac{\sigma_{ev}^2}{\sigma_v^2}) Q_{ZZ}= \sigma_e^2 (1-\rho^2) Q_{ZZ},
\]
and thus we have
\[
V_{ME} =  \sigma_e^2 (1-\rho^2) (\pi'Q_{ZZ}\pi)^{-1}. 
\]
Finally, using the formula for $V_{GMM}$
\[
V_{GMM}(W, \theta_W) = (G'WG)^{-1} (G'W \Sigma_{11,2} WG + C_W \Sigma_{22}C_W') (G'WG)^{-1} 
\]
with $C_W = G'W \Sigma_{12}\Sigma_{22}^{-1} + (g(\theta_W)'W\otimes I_p)$, where we can further have $G'W \Sigma_{11,2} WG  = \sigma_e^2(1-\rho^2) (\pi'Q_{ZZ}\pi)$, and $C_W\Sigma_{22}C_W' = \sigma_{ev}^2/\sigma_v^2 (\pi'Q_{ZZ}\pi) + \sigma_v^2 g(\theta_W)'Q_{ZZ}^{-1}g(\theta_W)$ by using the FOC, $G'Wg(\theta_W)=0$. Thus, we have

\begin{eqnarray*}
V_{GMM}(W, \theta_W) &=& (\pi'Q_{ZZ}\pi)^{-1} (\sigma_e^2(1-\rho^2) (\pi'Q_{ZZ}\pi) + \sigma_{ev}^2/\sigma_v^2  (\pi'Q_{ZZ}\pi) \\
&& \qquad + \sigma_v^2 g(\theta_W)'Q_{ZZ}^{-1}g(\theta_W)  )  (\pi'Q_{ZZ}\pi)^{-1}\\
&=& (\pi'Q_{ZZ}\pi)^{-1} (\sigma_e^2 (1-\rho^2) + \frac{\sigma_{ev}^2}{\sigma_v^2}) \\
&& \qquad + (\pi'Q_{ZZ}\pi)^{-1}  \frac{ng(\theta_W)'Q_{ZZ}^{-1} g(\theta_W)}{n (\pi'Q_{ZZ}\pi)/\sigma_v^2}.
\end{eqnarray*}

In sum, we have the decomposition as in \eqref{eq:gain_decomp} in the main paper,

\begin{equation*}
\frac{V_{GMM}(W, \theta_W)}{V_{ME}} -1 = \frac{\rho^2}{1-\rho^2}  + \frac{J_n}{\mu_n^2 (1-\rho^2)}  + o(1)
\end{equation*}
where $J_n = n g(\theta_W)'\Sigma_{11}^{-1} g(\theta_W) = O(1)$, and $\mu_n^2 = n \pi'Q_{ZZ}\pi/\sigma_v^2 = O(1)$. The ME efficiency gap is approximately measured by the correlation between the moments and the Jacobian $(\frac{\rho^2}{1-\rho^2})$, and the Hansen $J$ statistic scaled by the strength of identification $\frac{J_n}{\mu_n^2 (1-\rho^2)}$.

Furthermore, we can also compare the conventional/misspecification-robust SE with the ME bound. Note that the conventional variance $V_0 = (G'\Sigma_{11}^{-1}G)^{-1} = \sigma_e^2 (\pi'Q_{ZZ}\pi)^{-1}$. Therefore,  to leading order, we have
\[
 V_{ME}\approx V_0 (1-\rho^2), \quad V_{GMM} \approx V_0\Big(1+\frac{J_n}{\mu_n^2}\Big),\qquad
\]
Thus, 
\begin{equation}
\frac{SE_{ME}}{SE_{conv}}\approx \sqrt{1-\rho^2},\qquad
\frac{SE_{robust}}{SE_{conv}}\approx \sqrt{1+\frac{J_n}{\mu_n^2}},\qquad
\frac{SE_{robust}}{SE_{ME}}\approx\sqrt{\frac{1+J_n/\mu_n^2}{1-\rho^2}}.
\end{equation}
The misspecification-robust SE inflates the conventional SE by the scaled $J$
statistic $J_n/\mu_n^2$. In practice, $J_n$ can be replaced by the usual $J$ statistic, and the concentration parameter $\mu_n^2$ can be approximated by $mF$ with the number of excluded instruments $m$ and the first-stage $F$ statistic. It is also worth noting that the above derivation is useful to compare  the ratio of the variance formulas, but under the weak first-stage, all three variances are proportional to $(\pi'Q_{ZZ}\pi)^{-1} = O(n)$.

\section*{Supplemental Appendix References}

\begin{description}
\setlength{\itemsep}{3pt}\setlength{\parsep}{0pt}

\item[] Angrist, J. D., and Krueger, A. B. (1991). Does Compulsory School Attendance Affect Schooling and Earnings?, \emph{The Quarterly Journal of Economics}, 106(4), 979-1014.

\item[] Angrist, J. D., and Krueger, A. B. (1995). Split-sample instrumental variables estimates of the return to schooling. \emph{Journal of Business and Economic Statistics}, 13(2), 225-235.

\item[]  Arellano, M., and Bond, S. (1991). Some tests of specification for panel data: Monte Carlo evidence and an application to employment equations.  \emph{Review of Economic Studies}, 58(2), 277-297.

\item[] Card, D. (1995). Using Geographic Variation in College Proximity to Estimate the Return to Schooling, \emph{in Aspects of Labour Market Behaviour: Essays in Honour of John Vanderkamp}, ed. by L. N. Christofides, K. E. Grant, and R. Swidinsky, Toronto: University of Toronto Press, 201-222.

\item[] Chernozhukov, V., Chetverikov, D., Demirer, M., Duflo, E., Hansen, C., Newey, W., and Robins, J. (2018). Double/debiased machine learning for treatment and structural parameters. \emph{The Econometrics Journal}, 21, 1-68.

\item[]  Conley, T.G., Hansen, C.B., and Rossi, P.E. (2012). Plausibly exogenous, \emph{The Review of Economics and Statistics}, 94(1), 260-272.

\item[]Hall, P., and Horowitz, J. L. (1996). Bootstrap critical values for tests based on generalized-method-of-moments estimators. \emph{Econometrica}, 64(4), 891-916.

\item[] Roodman, D. (2009). A note on the theme of too many instruments, \emph{Oxford Bulletin of Economics and Statistics}, 71(1), 135-158.

\item[] Windmeijer, F. (2005). A finite sample correction for the variance of linear efficient two-step GMM estimators, \emph{Journal of Econometrics}, 126(1), 25-51.

\end{description}

\end{document}